\newif\ifsubmission
 \submissionfalse

\ifsubmission
\documentclass{llncs}
\else
\documentclass{article}
\usepackage{fullpage}
\fi
\usepackage[T1]{fontenc} 
\usepackage{lmodern} 
\usepackage[utf8]{inputenc} 
\usepackage{cmap} 
\usepackage[pdfstartview=FitH,pdfpagemode=None,colorlinks,linkcolor=blue,filecolor=blue,citecolor=Violet,urlcolor=red]{hyperref}
\usepackage{multirow}
\usepackage{float}
\usepackage[section]{placeins}         
\usepackage{mathtools}  
\usepackage[dvipsnames]{xcolor}
\usepackage{amsmath,amsthm,amssymb,algpseudocode,algorithm,cryptocode}
\usepackage{mathrsfs} 
\usepackage[capitalize]{cleveref}
\usepackage{xspace}  
\usepackage{braket}
\usepackage{comment}

\usepackage{tikz}

\newif\ifnotes
\ifsubmission
\notestrue
\else
\notestrue
\fi
\newcommand{\authnote}[3]{\textcolor{#3}{[{\footnotesize {\bf #1:} { {#2}}}]}}

\newcommand{\andrea}[1]{\ifnotes \authnote{A}{#1}{magenta} \fi}
\newcommand{\dakshita}[1]{\ifnotes \authnote{D}{#1}{blue} \fi}
\newcommand{\fermi}[1]{\ifnotes \authnote{F}{#1}{orange} \fi}

\ifsubmission
\else
\newtheorem{theorem}{Theorem}[section]

\newtheorem{lemma}[theorem]{Lemma}
\newtheorem{conjecture}{Conjecture}

\newtheorem{definition}[theorem]{Definition}

\fi

\Crefname{importedtheorem}{Imported Theorem}{Imported Theorems}
\Crefname{theorem}{Theorem}{Theorems}
\Crefname{proposition}{Proposition}{Propositions}
\Crefname{claim}{Claim}{Claims}
\Crefname{lemma}{Lemma}{Lemmas}
\Crefname{conjecture}{Conjecture}{Conjectures}
\Crefname{corollary}{Corollary}{Corollaries}
\Crefname{construction}{Construction}{Constructions}
\Crefname{property}{Property}{Properties}

\Crefname{definition}{Definition}{Definitions}
\Crefname{assumption}{Assumption}{Assumptions}
\Crefname{notation}{Notation}{Notations}

\theoremstyle{remark}

\Crefname{question}{Question}{Questions}
\Crefname{remark}{Remark}{Remarks}
\Crefname{comment}{Comment}{Comments}
\Crefname{fact}{Fact}{Facts}

\newcommand{\secp}{\lambda}

\def\cA{{\mathsf{Adv}}}
\def\cB{{\cal B}}

\def\cD{{\cal D}}

\def\cF{{\cal F}}
\def\cG{{\cal G}}
\def\cH{{\cal H}}
\def\cI{{\cal I}}

\def\cL{{\cal L}}
\def\cM{{\cal M}}

\def\cO{{\cal O}}
\def\cP{{\cal P}}
\def\cQ{{\cal Q}}
\def\cR{{\cal R}}
\def\cS{{\cal S}}

\def\cV{{\cal V}}

\def\bbF{{\mathbb F}}

\def\bbI{{\mathbb I}}

\def\bbN{{\mathbb N}}

\newcommand{\bx}{\mathbf{x}}

\newcommand{\by}{\mathbf{y}}
\newcommand{\bz}{\mathbf{z}}

\newcommand{\bct}{\mathbf{ct}}

\newcommand{\bd}{\mathbf{d}}
\newcommand{\be}{\mathbf{e}}
\newcommand{\bt}{\mathbf{t}}

\newcommand{\bm}{\mathbf{m}}

\newcommand{\bw}{\mathbf{w}}
\newcommand{\bp}{\mathbf{p}}
\newcommand{\bn}{\mathbf{n}}
\newcommand{\brho}{\boldsymbol{\rho}}
\newcommand{\bpi}{\boldsymbol{\pi}}

\newcommand{\bst}{\mathbf{st}}

\def\poly{{\rm poly}}

\def\negl{{\rm negl}}

\newcommand{\pk}{\mathsf{pk}}
\newcommand{\sk}{\mathsf{sk}}

\newcommand{\KeyGen}{\mathsf{KeyGen}}
\newcommand{\Setup}{\mathsf{Setup}}
\newcommand{\VSetup}{\mathsf{VSetup}}
\newcommand{\Prove}{\mathsf{Prove}}
\newcommand{\pvk}{\mathsf{pvk}}
\newcommand{\svk}{\mathsf{svk}}
\newcommand{\QMA}{\mathsf{QMA}}
\newcommand{\yes}{\mathsf{yes}}
\newcommand{\no}{\mathsf{no}}

\newcommand{\Com}{\mathsf{Com}}

\newcommand{\Sim}{\mathsf{Sim}}

\newcommand{\Enc}{\mathsf{Enc}}
\newcommand{\Dec}{\mathsf{Dec}}

\newcommand{\Verify}{\mathsf{Verify}}

\newcommand{\ct}{\mathsf{ct}}

\newcommand{\Eval}{\mathsf{Eval}}

\DeclareMathOperator*{\expectation}{\mathbb{E}}
\newcommand{\E}{\expectation}

\newcommand{\PRF}{\mathsf{PRF}}

\newcommand{\Gen}{\mathsf{Gen}}

\newcommand{\Tstate}{\mathbf{T}}
\newcommand{\Zstate}{\mathbf{0}}

\newcommand{\chck}{\mathsf{check}}
\newcommand{\decchck}{\mathsf{dec}-\mathsf{check}}
\newcommand{\decchckenc}{\mathsf{dec}-\mathsf{check}-\mathsf{enc}}
\newcommand{\decchckrerand}{\mathsf{dec}-\mathsf{check}-\mathsf{rerand}}
\newcommand{\rerandenc}{\mathsf{rerand}-\mathsf{enc}}
\newcommand{\rerand}{\mathsf{rerand}}
\newcommand{\Rerand}{\mathsf{Rerand}}
\newcommand{\Trap}{\mathsf{Trap}}
\newcommand{\trap}{\mathsf{trap}}
\newcommand{\baux}{\mathbf{aux}}
\newcommand{\Q}{\mathsf{Q}}
\newcommand{\C}{\mathsf{C}}
\newcommand{\Full}{\mathsf{Full}}
\newcommand{\Half}{\mathsf{Half}}
\newcommand{\secplev}{{\lambda_{\mathsf{lev}}}}
\newcommand{\inpcor}{\mathsf{inp}-\mathsf{cor}}
\newcommand{\comb}{\mathsf{comb}}
\newcommand{\extract}{\mathsf{extract}}

\newcommand{\trace}{\mathrm{Tr}}
\newcommand{\CM}{\mathsf{C+M}}
\newcommand{\dst}{\mathsf{dist}}

\newcommand{\proref}[1]{Protocol~\protect\ref{#1}}

\newenvironment{boxfig}[2]{\begin{figure}[#1]\fbox{
    \begin{minipage}{\linewidth}
    \vspace{0.2em}\makebox[0.025\linewidth]{}    \begin{minipage}{0.95\linewidth}{{#2 }}
    \end{minipage}\vspace{0.2em}\end{minipage}}}{\end{figure}}

\newcommand{\pprotocol}[4]{
\begin{boxfig}{h!}{
\begin{center}
\textbf{#1}
\end{center}
    #4
\vspace{0.2em} } \caption{\label{#3} #2}
\end{boxfig}
}

\newcommand{\protocol}[4]{
\pprotocol{#1}{#2}{#3}{#4} }

\newcommand{\gray}[1]{\textcolor{gray}{\mathsf{#1}}}

\newcommand{\Garble}{\mathsf{Garble}}
\newcommand{\GEval}{\mathsf{GEval}}
\newcommand{\GSim}{\mathsf{GSim}}

\newcommand{\QGarble}{\mathsf{QGarble}}
\newcommand{\QGEval}{\mathsf{QGEval}}
\newcommand{\QGSim}{\mathsf{QGSim}}

\newcommand{\labset}{\overline{\mathsf{lab}}}
\newcommand{\labsimset}{\widetilde{\mathsf{lab}}}
\newcommand{\labenc}{\mathsf{LabEnc}}
\newcommand{\lab}{\mathsf{lab}}

\newcommand{\crs}{\mathsf{crs}}
\newcommand{\inp}{\mathsf{inp}}

\newcommand{\twopc}{\mathsf{2PC}}
\newcommand{\gen}{\mathsf{Gen}}
\newcommand{\st}{\mathsf{st}}
\newcommand{\out}{\mathsf{out}}
\newcommand{\Real}{\mathsf{REAL}}
\newcommand{\Ideal}{\mathsf{IDEAL}}
\newcommand{\cmt}{\mathsf{cmt}}
\newcommand{\Ver}{\mathsf{Ver}}
\newcommand{\Open}{\mathsf{Open}}
\newcommand{\cmp}{\mathsf{cmp}}
\newcommand{\abort}{\mathsf{abort}}

\newcommand{\qmfhe}{\mathsf{QMFHE}}
\newcommand{\mfhe}{\mathsf{MFHE}}

\newcommand{\CEnc}{\mathsf{CEnc}}

\newcommand{\eval}{\mathsf{Eval}}

\newcommand{\MPC}{\mathsf{MPC}}

\newcommand{\Circle}{\mathsf{circle}}
\newcommand{\test}{\mathsf{test}}
\newcommand{\QGC}{\mathsf{QGC}}
\newcommand{\enc}{\mathsf{enc}}
\begin{document}

\title{On the Round Complexity of Secure Quantum Computation\footnote{This paper subsumes an earlier version \cite{cryptoeprint:2020:1471} that only considers secure two-party quantum computation.}}
\author{James Bartusek\thanks{UC Berkeley. Email: \texttt{bartusek.james@gmail.com}} \and Andrea Coladangelo\thanks{UC Berkeley. Email: \texttt{andrea.coladangelo@gmail.com}} \and Dakshita Khurana\thanks{UIUC. Email: \texttt{dakshita@illinois.edu}} \and Fermi Ma\thanks{Princeton University and NTT Research. Email: \texttt{fermima@alum.mit.edu}}}
\date{}
\maketitle

\begin{abstract}

We construct the first \emph{constant-round} protocols for secure quantum computation in the two-party (2PQC) and multi-party (MPQC) settings with security against \emph{malicious} adversaries. Our protocols are in the common random string (CRS) model.

\begin{itemize}
    \item Assuming two-message oblivious transfer (OT), we obtain ($i$) three-message 2PQC, and ($ii$) five-round MPQC with only three rounds of \emph{online} (input-dependent) communication; such OT is known from quantum-hard Learning with Errors (QLWE).
    \item Assuming sub-exponential hardness of QLWE, we obtain ($i$) three-round 2PQC with two online rounds and ($ii$) four-round MPQC with two online rounds.
    \item When only one (out of two) parties receives output, we achieve \emph{minimal interaction} (two messages) from two-message OT; classically, such protocols are known as non-interactive secure computation (NISC), and our result constitutes the first maliciously-secure quantum NISC. 
    
     Additionally assuming reusable malicious designated-verifier NIZK arguments for $\mathsf{NP}$ (MDV-NIZKs), we give the first MDV-NIZK for $\mathsf{QMA}$ that only requires one copy of the quantum witness.
        \end{itemize}

    Finally, we perform a preliminary investigation into \emph{two-round} secure quantum computation where each party must obtain output. On the negative side, we identify a broad class of simulation strategies that suffice for \emph{classical} two-round secure computation that are \emph{unlikely} to work in the quantum setting. Next, as a proof-of-concept, we show that two-round secure quantum computation exists with respect to a quantum oracle.

\end{abstract}
\newpage

{
  \hypersetup{linkcolor=Violet}
  \setcounter{tocdepth}{2}
  \tableofcontents
}
\newpage

\section{Introduction}


Secure computation allows mutually distrusting parties to compute arbitrary functions on their private inputs, revealing only the outputs of the computation while hiding all other private information~\cite{FOCS:Yao86,STOC:GolMicWig87,STOC:BenGolWig88,C:ChaCreDam87}.
With the emergence of quantum computers, it becomes important to understand the landscape of secure \emph{quantum} computation over distributed, private quantum (or classical) states. In the most general setting, $n$ parties hold (possibly entangled) quantum inputs $\bx_1,\dots,\bx_n$, and would like to evaluate a quantum circuit $Q(\bx_1,\dots,\bx_n)$. The output is of the form $(\by_1,\dots,\by_n)$, so at the end of the protocol party $i$ holds state $\by_i$. 

Secure computation with classical inputs and circuits forms a centerpiece of classical cryptography. Solutions to this problem in the classical setting were first obtained nearly 35 years ago, when~\cite{FOCS:Yao86} built garbled circuits to enable secure two-party computation, and~\cite{STOC:GolMicWig87,STOC:BenGolWig88,C:ChaCreDam87} obtained the first secure multi-party computation protocols. Since then, there has been an extensive body of work in this area, of which a large chunk focuses on understanding the amount of back-and-forth interaction required to implement these protocols. Notably, the work of Beaver, Micali and Rogaway~\cite{STOC:BeaMicRog90} obtained the first constant-round classical multi-party computation protocols in the dishonest majority setting. There have been several subsequent works including but not limited to~\cite{C:KatOst04,EC:GMPP16,C:AnaChoJai17,TCC:BraHalPol17,C:BGJKKS18,TCC:CCG0O20} that have nearly completely characterized the {\em exact} round complexity of classical secure computation.

The problem of secure {\em quantum} computation on distributed quantum states is not nearly as well-understood as its classical counterpart. The quantum setting was first studied by~\cite{STOC:CreGotSmi02,FOCS:BCGHS06}, who obtained unconditional maliciously-secure multi-party quantum computation with honest majority. 
Just like the classical setting, when half (or more) of the players are malicious, secure quantum computation also requires computational assumptions due to the impossibility of unconditionally secure quantum bit commitment~\cite{mayers1997unconditionally,lo1998quantum,d2006quantum}. 

In the dishonest majority setting,~\cite{C:DupNieSal10} gave a two-party quantum computation (2PQC) protocol secure against the quantum analogue of semi-honest adversaries (specious adversaries); this was later extended to the malicious setting by~\cite{DNS12}. A work of~\cite{EC:DGJMS20} constructed maliciously-secure \emph{multi-party} quantum computation (MPQC) with dishonest majority from any maliciously-secure post-quantum classical MPC, where the round complexity grows with the size of the circuit {\em and} the number of participants. Very recently,~\cite{cryptoeprint:2020:1464} constructed MPQC with identifiable abort, and with round complexity that does not grow with the circuit size but grows with the number of participants. 

However, the feasibility of maliciously-secure MPQC with \emph{constant} rounds has remained open until this work. In addition to settling this question, we also make several headways in understanding the {\em exact} round complexity of secure quantum computation with up to all-but-one malicious corruptions.

\subsection{Our Results}

We assume that parties have access to a common random string (CRS), and obtain answers to a range of fundamental questions, as we discuss next\footnote{We point out that the post-quantum MPC protocol of~\cite{Agarwal2020PostQuantumMC} can be used to generate a CRS in constant rounds. This, combined with our results, yields the first constant round multi-party quantum computation protocols without trusted setup in the standard model.}. 

\subsubsection{Quantum Non-Interactive Secure Computation}

Our first result pertains to the most basic setting for secure (quantum) computation: a sender holds input $\mathbf{y}$, a receiver holds input $\mathbf{x}$, and the goal is for the receiver to obtain $Q(\mathbf{x},\mathbf{y})$ for some quantum circuit $Q$. We construct a protocol achieving \emph{minimal interaction} --- commonly known as non-interactive secure computation (NISC)~\cite{EC:IKOPS11} --- where the receiver publishes an encryption of $\mathbf{x}$, the sender replies with an encryption of $\mathbf{y}$, and the receiver subsequently obtains $Q(\mathbf{x},\mathbf{y})$. Our result constitutes the first maliciously-secure NISC for quantum computation (Q-NISC).

\begin{theorem}(Informal)\label{thm:informal-3-message}
Maliciously-secure NISC for quantum computation exists assuming post-quantum maliciously-secure two-message oblivious transfer (OT) with straight-line simulation.
\end{theorem}

Such OT protocols are known from the post-quantum hardness of Learning with Errors (LWE)~\cite{C:PeiVaiWat08}. We remark that our Q-NISC result also extends to the {\em reusable} setting where the receiver has a classical input that they would like to reuse across multiple quantum computations on different sender inputs.

\paragraph{Application: Malicious Designated-Verifier NIZK Arguments for $\mathsf{QMA}$.} As an application of our maliciously-secure Q-NISC, we construct (reusable) \emph{malicious designated-verifier non-interactive zero-knowledge arguments} (MDV-NIZKs) for $\mathsf{QMA}$ in the common random string model. Specifically, our MDV-NIZK enables the following interaction for any $\mathsf{QMA}$ language: a verifier can publish a classical public key $\mathsf{\pk}$ that enables a prover to send an instance $x$ and quantum message $\mathbf{m}$, such that the verifier holding the corresponding secret key $\mathsf{sk}$ can determine if $x$ is a valid instance.

\begin{theorem}(Informal)
\label{thm:inf-dv}
There exists a reusable MDV-NIZK for $\mathsf{QMA}$ with a classical CRS and classical proving key assuming the existence of post-quantum maliciously-secure two-message oblivious transfer with straight-line simulation in the CRS model, and post-quantum (adaptively sound) reusable MDV-NIZK for NP. All of the underlying primitives exist assuming the quantum hardness of learning with errors. 
\end{theorem}

We briefly elaborate on the security guarantees of our reusable MDV-NIZK. Reusability means that soundness holds for multiple proofs (of potentially different statements) computed with respect to the same setup (i.e., the common random string and the public key), even if the prover learns whether or not the verifier accepted each proof; we remark that reusable security is sometimes referred to as multi-theorem security. Malicious security means that the zero-knowledge property holds even against verifiers that generate the public key maliciously. Previously, such a reusable MDV-NIZK for QMA required the prover to have access to multiple copies of the quantum witness~\cite{shmueli2020multitheorem}, while our MDV-NIZK only requires the prover to have a single copy.

\subsubsection{Constant-round 2PQC and MPQC}


Our next set of results concerns the general setting for 2PQC and MPQC \emph{where all parties obtain output}. We focus on minimizing total round complexity as well as \emph{online} round complexity, where the latter refers to the number of \emph{input-dependent} rounds; if a protocol has round complexity $d$ and online round complexity $k$, then the parties can perform the first $d-k$ rounds \emph{before} they receive their inputs.\footnote{We remark that a $k$-online round protocol can also be viewed as a $k$-round protocol in a quantum pre-processing model, i.e. a model where parties receive some quantum correlations as setup.}

We obtain various results, some from the generic assumption of quantum polynomially-secure two-message oblivious transfer, and others from the specific assumption of sub-exponential QLWE. Our results in this section are summarized in~\cref{table: results}.\footnote{The results below are in the setting of security with abort, as opposed to security with \emph{unanimous} abort (which is only a distinction in the multi-party setting). If one wants security with unanimous abort, the overall round complexity will not change, but one more round of online communication will be required.} 



\begin{table}[ht!]
\caption{Maliciously-Secure Quantum Computation in the CRS Model}
\begin{center}
    \begin{tabular}{c|c|c}
         &  From OT & From sub-exp QLWE\\
         \hline
         Two-party & 3 rounds (3 online) & 3 rounds (2 online) \\
         \hline
         Multi-party & 5 rounds (3 online) & 4 rounds (2 online)
    \end{tabular}
\end{center}
\label{table: results}
\end{table}

In order to prove the security of these protocols, we develop a delayed simulation technique, which we call ``simulation via teleportation'', which may be of independent interest.

\subsubsection{Is Two-Round Secure Quantum Computation Possible?}
A natural next question is whether it is possible to construct two-round secure quantum computation \emph{without} pre-processing. This appears to be a challenging question to resolve, either positively or negatively. We provide some preliminary results on both fronts: we give a negative result indicating that common simulation strategies from the classical setting will not suffice in the quantum setting, but we also provide a proof-of-concept positive result, with a new simulation strategy, assuming virtual-black-box obfuscation of quantum circuits. We stress that the latter result is primarily to give intuition, as virtual-black-box obfuscation is known to be impossible even for classical circuits~\cite{C:BGIRSVY01}. We limit the scope of this preliminary investigation to the \emph{two-party} setting.

First, we give some intuition for why it seems hard to design a two-round two-party protocol by showing that, under a plausible quantum information-theoretic conjecture, a large class of common simulation techniques would \emph{not} suffice. We consider any simulator that learns which player (between Alice and Bob) is corrupted only {\em after} it has generated the simulated CRS. We call such a simulator an \emph{oblivious simulator}. To the best of our knowledge, all existing classical and quantum two-party computation protocols in the CRS model either (1) already admit oblivious simulation, or (2) can generically be transformed to admit oblivious simulation via post-quantum NIZK proofs of knowledge for $\mathsf{NP}$.


In the quantum setting, we show, roughly, that any two-round 2PQC protocol for general quantum functionalities \emph{with an oblivious simulator} would yield an \textit{instantaneous nonlocal quantum computation} protocol \cite{PhysRevLett.90.010402,Beigi_2011,DBLP:conf/tqc/Speelman16,DBLP:journals/tit/GonzalesC20} for general quantum functionalities, with polynomial-size pre-processing. 


Instantaneous nonlocal quantum computation is well-studied in the quantum information literature \cite{PhysRevLett.90.010402,Beigi_2011,DBLP:conf/tqc/Speelman16,DBLP:journals/tit/GonzalesC20}, and the best known protocols for general functionalities require exponential-size pre-processing \cite{Beigi_2011}. Thus, a two-round 2PQC for general functionalities with oblivious simulation would immediately yield an exponential improvement in the size of the pre-processing for this task.

\begin{theorem} (Informal)
Under the conjecture that there exists a quantum functionality that does not admit an instantaneous nonlocal quantum computation protocol with polynomial-size pre-processing, there exists a quantum functionality that cannot be securely computed in two rounds in the classical CRS model with an oblivious simulator.
\end{theorem}


\ifsubmission
\else
Towards getting around this potential barrier, we give a proof-of-concept construction of a protocol with non-oblivious simulation. Specifically, we assume a (strong) form of VBB obfuscation for quantum circuits that contain unitary and measurement gates, where the former may be classically controlled on the outcome of measurement gates. We point out, however, that VBB-obfuscation of circuits with measurement gates is potentially even more powerful than the VBB obfuscation for unitaries that was formalized in~\cite{Alagic2016OnQO} (see discussion in \cref{subsec:vbb-protocol}). Under this assumption, we obtain a two-round two-party secure quantum computation protocol in the CRS model (that is straightforward to extend to the multi-party setting).
\fi

\begin{theorem}(Informal)
Two-round two-party secure quantum computation in the common reference string model exists assuming a strong form of VBB or ideal obfuscation for quantum circuits as discussed above.
\end{theorem}

We remark that while there exist (contrived) examples of functionalities that cannot be VBB obfuscated~\cite{Alagic2016OnQO,alagic2020impossibility,ananth2020secure}, it is still plausible that many quantum functionalities can be obfuscated. However, without any candidate constructions of obfuscation for quantum circuits, we stress that our result should only be taken as a proof-of-concept.

\section{Technical Overview}
\label{sec:techoverview}

\subsection{Quantum Background}

We briefly recap some relevant concepts from quantum computation.

\paragraph{Notation.} We use bold letters to write the density matrix of a quantum state $\bx$. We use the shorthand $U(\bx)$ to mean $U \bx U^\dagger$, the result of applying unitary $U$ to $\bx$. The notation $(\bx,\by)$ denotes a state on two registers, where $\bx$ and $\by$ are potentially entangled. The $k$-fold tensor product of a state $\bx \otimes \bx \otimes \cdots \otimes \bx$ will be written as $\bx^k$.

\paragraph{The Pauli Group.} The Pauli group on a single qubit, denoted by $\mathscr{P}_1$, is generated by the unitary operations $X$ (bit flip) and $Z$ (phase flip), defined as $X = \begin{bmatrix} 0 & 1 \\ 1 & 0 \end{bmatrix}, Z = \begin{bmatrix} 1 & 0 \\ 0 & -1 \end{bmatrix}.$ The Pauli group on $n$ qubits, denoted by $\mathscr{P}_n$, is the $n$-fold tensor product of the single qubit Pauli group. Any unitary in the Pauli group $\mathscr{P}_n$ can be written (up to global phase) as $\bigotimes_{i \in [n]} X^{r_i}Z^{s_i}$ for $r, s \in \{0,1\}^n$.

\paragraph{The Clifford Group.} The Clifford group on $n$ qubits, denoted by $\mathscr{C}_n$, is the group of unitaries that normalize $\mathscr{P}_n$, i.e. $C \in \mathscr{C}_n$ if and only if for all $U \in \mathscr{P}_n$, we have $CUC^\dagger \in \mathscr{P}_n$. Alternatively, we can think of a Clifford unitary $C$ as an operation where for any choice of $r,s \in \{0,1\}^n$, there exists a choice of $r',s' \in \{0,1\}^n$ such that 
\[ C \left(\bigotimes_{i \in [n]} X^{r_i}Z^{s_i}\right) = \left(\bigotimes_{i \in [n]} X^{r'_i}Z^{s'_i}\right) C.\]

Intuitively, this means that with a suitable update of the Pauli operation, one can swap the order in which a Clifford and a Pauli are applied.

\paragraph{Clifford Authentication Codes.} We will make extensive use of Clifford authentication codes. Clifford authentication codes are an information-theoretic encoding scheme for quantum states that provides both secrecy and authentication. An $n$-qubit quantum state $\mathbf{x}$ can be encoded in a Clifford authentication code as follows: prepare a $\lambda$-qubit all $0$'s state which we denote as $\mathbf{0}^\lambda$ (where $\lambda$ is a security parameter), sample a random Clifford unitary $C \leftarrow \mathscr{C}_{n+\lambda}$, and output $C (\mathbf{x}, \mathbf{0}^\lambda)$. The Clifford $C$ serves as a secret key, while the $\mathbf{0}^\lambda$ qubits enable authentication, and are called ``trap'' qubits. A party without knowledge of $C$ cannot modify the encoding without modifying the trap qubits (except with negligible probability). Therefore, decoding works by applying $C^\dagger$ and then measuring the $\lambda$ trap qubits in the computational basis. If these measurements are all $0$, this ensures that with all but negligible probability, the $n$ remaining registers hold the originally encoded state $\bx$.

\paragraph{Clifford + Measurement Circuits.} We will rely heavily on the ``Clifford + Measurement'' representation of quantum circuits (henceforth ``C+M circuits'') due to \cite{BravyiKitaev05}. In this representation, a quantum circuit can be decomposed into layers. Each layer consists of a Clifford unitary whose output wires are partitioned into wires that will be fed as inputs into the next layer, and wires that will be measured. The latter group of wires are measured in the computational basis, resulting in a classical bitstring that is used to select the Clifford unitary to be applied in the subsequent layer. The first layer takes in all of the inputs to the quantum circuit, ancilla $\mathbf{0}$ states, and ``magic'' $\mathbf{T}$ states defined as $\mathbf{T} \coloneqq (\ket{0} + e^{i \pi/4}\ket{1})/\sqrt{2}$. The final layer only produces output wires (i.e. its output registers have no wires to be measured), which are interpreted as the output of the circuit. \cite{BravyiKitaev05} demonstrate that, with constant multiplicative factor overhead in size, any quantum circuit can be written as a ``C + M circuit'' or equivalently, in a magic state representation.

Therefore, for the purposes of this technical overview, we will assume that any quantum circuit $F$ is written as a C+M circuit $F_{\mathrm{CM}}$, and its evaluation on an input $\bx$ is computed as $F(\bx) = F_{\mathrm{CM}}(\bx, \mathbf{T}^k, \mathbf{0}^k).$ For simplicity, we use the same $k$ to denote the number of $\mathbf{T}$ states and the number of ancilla $\mathbf{0}$ states.

\paragraph{Magic State Distillation.} In settings where malicious parties are tasked with providing the $\mathbf{T}$ states, we will use cryptographic techniques such as ``cut-and-choose'' to ensure that $F_{\mathrm{CM}}$ is evaluated on an input of the form $(\mathbf{x},\widehat{\mathbf{T}^k},\mathbf{0}^k)$ where $\widehat{\mathbf{T}^k}$ is a state guaranteed to be ``somewhat'' close to $\mathbf{T}^k$. However, correctness of $F_{\mathrm{CM}}$ will require states that are negligibly close to real magic states. To that end, we will make use of a magic state distillation C+M circuit $D$ due to~\cite{EC:DGJMS20} which takes in somewhat-close magic states $\widehat{\mathbf{T}^k}$ and outputs states negligibly close to $\mathbf{T}^{k'}$, for $k' < k$. Therefore, the representation of any functionality $F$ will in fact be a C+M circuit $F_{\mathrm{CM},D}$ that first applies $D$ to $\widehat{\mathbf{T}^k}$, and then runs $F_{\mathrm{CM}}$.

\subsection{Why is Malicious Security Hard to Achieve?}

We begin this technical overview by describing our results in the two-party setting. Before this, we briefly explain why malicious security does not follow readily from existing techniques. Indeed, a candidate two-message 2PQC (where one party receives output) with \emph{specious} security (the quantum analogue of classical semi-honest security~\cite{C:DupNieSal10}) was recently proposed in~\cite{ARXIV:BrakerskiYuen20}. Alternatively, any construction of quantum fully-homomorphic encryption (QFHE) naturally yields a two-message 2PQC protocol: (1) Alice QFHE-encodes her input and sends it to Bob, (2) Bob evaluates the functionality on his input and Alice's encoded input, and (3) Bob sends Alice the encryption of her output.

One might hope to compile this QFHE-based protocol or the \cite{ARXIV:BrakerskiYuen20} protocol into a maliciously secure protocol by having the parties include proofs that their messages are well-formed. Unfortunately, it is unclear how to implement this in the quantum setting. In both of these approaches, the parties would have to prove (in zero-knowledge) statements of the form ``$\by$ is the result of evaluating quantum circuit $C$ on $\bx$.'' Crucially, the \emph{statement} the parties need to prove explicitly makes reference to a quantum state. This is beyond the reach of what one can prove with, say, NIZKs for $\mathsf{QMA}$, in which witnesses are quantum but the statements are entirely classical.

Therefore, we design our malicious 2PQC so that parties do not have to prove general statements about quantum states. A core ingredient in our protocol is a quantum garbled circuit construction sketched in~\cite[\S2.5]{ARXIV:BrakerskiYuen20}, where the circuit garbling procedure is entirely classical.\footnote{We remark that the 2PQC proposed in \cite{ARXIV:BrakerskiYuen20} is based on their ``main'' quantum garbled circuit construction, which crucially does \emph{not} have a classical circuit garbling procedure. The advantage of their main construction is that garbling can be done in low depth, whereas the alternative construction requires an expensive but classical garbling procedure.} Combining this with a post-quantum maliciously-secure \emph{classical} 2PC, we will ensure valid circuit garbling against malicious quantum adversaries.

\subsection{A Garbling Scheme for $\CM$ Circuits} Our first step is to formalize the proposal sketched in \cite[\S2.5]{ARXIV:BrakerskiYuen20} for garbling $\CM$ circuits. The starting point for the \cite[\S2.5]{ARXIV:BrakerskiYuen20} construction is a simple technique for garbling any quantum circuit that consists of a single Clifford unitary $F$.\footnote{ \cite{ARXIV:BrakerskiYuen20} call this \emph{group-randomizing quantum randomized encoding}.} The idea is to sample a random Clifford $E$ and give out $FE^{\dagger}$ as the garbled circuit; note that the description of $FE^{\dagger}$ will be entirely classical. Since the Clifford unitaries form a group, $FE^{\dagger}$ is a uniformly random Clifford unitary independent of $F$. To garble the input quantum state $\mathbf{x}$, simply compute $E(\mathbf{x})$. The construction in \cite[\S2.5]{ARXIV:BrakerskiYuen20} extends this simple construction to any circuit. 

To build intuition, we will consider a two-layer $\CM$ circuit $Q = (F_1,f)$, where $F_1$ is the first layer Clifford unitary, and $f$ is a classical circuit that takes as input a single bit measurement result $m$, and outputs a classical description of $F_2$, the second layer Clifford unitary. On input $\mathbf{x}$, the circuit operates as follows:
\begin{enumerate}
\item Apply $F_1$ to $\mathbf{x}$.
\item Measure the last output wire in the computational basis to obtain $m \in \{0,1\}$, and feed the remaining wires to the next layer. Compute the second layer Clifford unitary $F_2 = f(m)$.
\item Apply $F_2$ to the non-measured output wires from the first layer. Return the result.
\end{enumerate}

One could try to extend the simple idea for one-layer garbling to this circuit. We still sample a random input-garbling Clifford $E_0$ and compute $F_1E_0^{\dagger}$. To hide the second layer Clifford, a natural idea is to sample yet another random Clifford $E_1$ to be applied to the non-measured output wires of $F_1$. That is, we replace $F_1E_0^{\dagger}$ with $(E_1 \otimes \bbI)F_1 E_0^{\dagger}$, and release the description of a function $g$ such that $g(m) = f(m) E_1^{\dagger}$.

However, this may in general be insecure. Let $F_2^{(0)}$ be the Clifford output by function $f$ when $m = 0$, and $F_2^{(1)}$ the Clifford output by function $f$ when $m = 1$. Suppose $ F_2^{(0)} -  F_2^{(1)} = A$ for some invertible matrix $A$. 
Then, an attacker with access to $g$ could obtain $F_2^{(0)}E_1^{\dagger} - F_2^{(1)}E_1^{\dagger}$, and multiplying the result by $A^{-1}$ yields $ A^{-1}(F_2^{(0)}E_1^{\dagger} - F_2^{(1)}E_1^{\dagger}) = A^{-1}AE_1^{\dagger} = E_1^{\dagger}$.

Therefore, instead of giving out $g$, the construction of \cite[\S2.5]{ARXIV:BrakerskiYuen20} gives out a classical garbling of $g$. To accommodate this, the output wire from the first layer that is measured to produce $m \in \{0,1\}$ must be replaced by a collection of wires that produces the corresponding label $\mathsf{lab}_{m}$ for the garbled circuit. This can be easily achieved by applying a suitable ``label unitary'' to the $m$ wire (and ancilla wires) within the garbled gate for the first layer.

There is one last issue with this approach: an attacker that chooses not to measure the wires containing $\mathsf{lab}_{m}$ can obtain a superposition over two valid labels. Recall that the standard definition of security for classical garbled circuits only guarantees simulation of one label, not a quantum superposition of both labels. To ensure the attacker cannot get away with skipping the computational basis measurement, the \cite[\S2.5]{ARXIV:BrakerskiYuen20} construction applies a $Z$-twirl to $m$ before the ``label unitary'' is applied. Recall that a $Z$-twirl is simply a random application of a Pauli $Z$ gate, i.e. $Z^b$ for a uniformly random bit $b$; applying $Z^b$ to a wire is equivalent to performing a computational basis measurement (without recording the result). 

To recap, a garbled $2$-layer $\CM$ circuit $Q$ consists of three components: an ``input garbling'' Clifford $E_0$, an initial Clifford unitary to be applied to the garbled input $D_0 \coloneqq (E_1 \otimes \bbI)F_1 E_0^{\dagger}$, and a classical garbled circuit $\widetilde{g}$. Extrapolating, we see that in general a garbled $\CM$ circuit takes the form $$(E_0,D_0,\widetilde{g}_1,\dots,\widetilde{g}_d) \coloneqq (E_0,\widetilde{Q}),$$
where the $\widetilde{g}_i$'s are garblings of classical circuits.
Crucially, all of these components can be generated by an entirely classical circuit. The only quantum operation involved in the garbling process is the application of $E_0$ to the input $\bx$ to garble the input. Next, we show how we can take advantage of this mostly classical garbling procedure to obtain maliciously-secure 2PQC.

\subsection{A Three-Message Protocol with Malicious Security}
\label{subsec:3-round-malicious}

In this section, we describe a three-message 2PQC protocol where both parties obtain output. This implies the two-message 2PQC result with one-sided output described in the first part of our results section, and fills in the upper left corner of \cref{table: results}.

We begin with a plausible but \emph{insecure} construction of a three-message 2PQC based on the above quantum garbled circuit construction. We will then highlight the ways a malicious attacker might break this construction, and arrive at our final construction by implementing suitable modifications.

Our protocol relies only on a \emph{classical} two-message 2PC with one-sided output that is (post-quantum) secure against malicious adversaries; this can be realized by combining (post-quantum) classical garbled circuits~\cite{FOCS:Yao86} with (post-quantum) two-message oblivious transfer~\cite{C:PeiVaiWat08} following eg.~\cite{EC:IKOPS11}.

We will consider two parties: Alice with input $\mathbf{x}_A$ and Bob with input $\mathbf{x}_B$.
They wish to jointly compute a quantum circuit $Q$ on their inputs whose output is delivered to both players. 
$Q$ is represented as a Clifford+Measurement circuit that takes input $(\mathbf{x}_A, \mathbf{x}_B, \mathbf{T}^k, \mathbf{0}^k)$. We denote by $(\mathbf{y}_A, \mathbf{y}_B)$ the joint outputs of Alice and Bob. At a high level, the parties will use the first two messages (Bob $\to$ Alice, Alice $\to$ Bob) to jointly encode their quantum inputs, while in parallel computing a two-message classical 2PC that outputs the classical description of a quantum garbled circuit to Bob. By evaluating the garbled circuit, Bob can learn his own output, as well as Alice's encoded output, which he sends to Alice in the 3rd message.

In more detail, the classical functionality $\cF[Q]$ to be computed by the classical 2PC is defined as follows. It takes as input (the classical description of) a Clifford unitary $C_{B,\mathrm{in}}$ from Bob and Clifford unitaries $(C_{A,\mathrm{in}},C_{A,\mathrm{out}})$ from Alice. Let $Q_B$ be a modification of $Q$ that outputs $(C_{A,\mathrm{out}}(\by_A,\Zstate^\secp),\mathbf{y}_B)$ in place of $(\mathbf{y}_A, \mathbf{y}_B)$; looking ahead, this will enable Bob to evaluate (a garbling of) $Q_B$ on (a garbling of) their joint inputs without learning Alice's output. The functionality computes a garbling $(E_0,\widetilde{Q_B})$ of $Q_B$. Finally, it computes $W \coloneqq E_0 \cdot (\bbI \otimes C_{B,\mathrm{in}}^{-1} \otimes \bbI) \cdot C_{A,\mathrm{in}}^{-1}$ (where the registers implied by the tensor product will become clear below), 
and outputs $(W,\widetilde{Q_B})$ to Bob.


The (insecure) protocol template is as follows:
\begin{itemize}
\item \textbf{First Message (Bob $\rightarrow$ Alice).} Bob picks a random Clifford $C_{B,\mathrm{in}}$ and uses it to encrypt and authenticate his input $\mathbf{x}_B$ as $\mathbf{m}_1 \coloneqq C_{B,\mathrm{in}} (\bx_B,\Zstate^\secp)$. He also computes the first round message $m_1$ of the classical 2PC, using $C_{B,\mathrm{in}}$ as his input. He sends $(\mathbf{m}_1,m_1)$ to Alice.
\item \textbf{Second Message (Alice $\rightarrow$ Bob).} After receiving $(\mathbf{m}_1,m_1)$, Alice picks a random Clifford $C_{A,\mathrm{in}}$ and uses it to encrypt her input $\mathbf{x}_A$ along with Bob's encoding $\mathbf{m}_1$, $k$ copies of a $\mathbf{T}$ state, and $k+\lambda$ copies of a $\mathbf{0}$ state. The result of this is $\mathbf{m}_2 \coloneqq C_{A,\mathrm{in}}(\mathbf{x}_A, \mathbf{m}_1, \mathbf{T}^k, \mathbf{0}^{k+\lambda})$. Alice also samples another random Clifford $C_{A,\mathrm{out}}$ that will serve to encrypt and authenticate her output, and computes the second round message $m_2$ of the classical 2PC using input $(C_{A,\mathrm{in}},C_{A,\mathrm{out}})$. She sends $(\mathbf{m}_2,m_2)$ to Bob.
\item \textbf{Third Message (Bob $\rightarrow$ Alice).} After receiving $(\mathbf{m}_2,m_2)$, Bob can compute his output of the classical 2PC, which is $(W,\widetilde{Q_B})$. He computes $$W(\mathbf{m}_2) = 
E_0 \cdot (\bbI \otimes C_{B,\mathrm{in}}^{-1} \otimes \bbI) \cdot C_{A,\mathrm{in}}^{-1} \left( C_{A,\mathrm{in}}(\mathbf{x}_A, \mathbf{m}_1, \mathbf{T}^k, \mathbf{0}^{k+\lambda}) \right)  
= E_0 (\mathbf{x}_A,\mathbf{x}_B,\mathbf{T}^k, \mathbf{0}^{k+\lambda}).$$ Recall that $E_0(\mathbf{x}_A, \mathbf{x}_B, \mathbf{T}^k, \mathbf{0}^{k+\lambda})$ corresponds to a garbled input for $\widetilde{Q_B}$.
He evaluates $\widetilde{Q_B}$ on this garbled input and obtains $(C_{A,\mathrm{out}}(\mathbf{y}_A,\Zstate^\secp), \mathbf{y}_B)$.

At this point, Bob has his output $\mathbf{y}_B$ in the clear. Next he sets $\mathbf{m}_3 = C_{A,\mathrm{out}}(\mathbf{y}_A,\Zstate^\secp)$, and sends $\mathbf{m}_3$ to Alice.
Upon receiving $\mathbf{m}_3$, Alice can recover her output by computing $C_{A,\mathrm{out}}^{-1}(\mathbf{m}_3)$.
\end{itemize}

The above protocol can already be shown to be secure against malicious Bob by relying on security of the classical two-party computation protocol against malicious adversaries. But malicious Alice can break security by generating ill-formed auxiliary states. We now describe this issue in some more detail and then present modifications to address the problem.

\paragraph{Malicious Generation of Auxiliary States.} In the second message of the protocol, Alice is instructed to send a quantum state $C_{A,\mathrm{in}}(\bx_A,\bm_1,\mathbf{T}^k,\mathbf{0}^{k+\lambda})$. A malicious Alice can deviate from honest behavior by submitting arbitrary states in place of the magic $\mathbf{T}$ states and the auxiliary $\mathbf{0}$ states, either of which may compromise security. 

We therefore modify the classical 2PC to include randomized checks that will enable Bob to detect if Alice has deviated from honest behavior.

We check validity of $\mathbf{0}$ states using the ``random linear map'' technique of~\cite{EC:DGJMS20}. The classical 2PC will sample a uniformly random matrix $M \in \mathbb{F}_2^{k \times k}$, and apply a unitary $U_{M}$ that maps the quantum state $\mathbf{v} = \ket{v}\bra{v}$ for any $v \in \mathbb{F}_2^k$ to the state $\mathbf{Mv} = \ket{Mv}\bra{Mv}$. For any $M \in \mathbb{F}_2^{k \times k}$, there exists an efficient Clifford unitary $U_M$ implementing this map. This check takes advantage of the fact that $U_M(\mathbf{0}^k) = \mathbf{0}^k$ for any $M$, but on any other pure state $\mathbf{v} = \ket{v}\bra{v}$ for non-zero $v \in \mathbb{F}_2^k$, we have $U_M(\mathbf{v}) \neq \mathbf{0}^k$ with overwhelming probability in $k$.

More precisely, our protocol will now ask Alice to prepare twice $(2k)$ the required number of $\mathbf{0}$ states. The classical 2PC will generate a Clifford unitary $U_{M}$ implementing a random linear map $M \in \mathbb{F}_2^{2k \times 2k}$, and incorporate $U_{M}$ into its output Clifford $W$, which is now $W = (E_0 \otimes \bbI) \cdot (\bbI \otimes C_{B,\mathrm{in}}^{-1} \otimes \bbI) \cdot (\bbI \otimes U_M) \cdot C_{A,\mathrm{in}}^{-1}$. Now when Bob applies $W$ to Alice's message $C_{A,\mathrm{in}}(\bx_A,C_{B,\mathrm{in}}(\bx_B,\Zstate^\secp),\mathbf{T}^k,\mathbf{0}^{2k})$, it has the effect of stripping off $C_{A,\mathrm{in}}$ by applying $C_{A,\mathrm{in}}^{-1}$, and then applying $U_{M}$ to the last $2k$ registers. The rest of the application of $W$ has the same effect as before the modification, so it undoes the application of $C_{B,\mathrm{in}}$, and then re-encodes \emph{all but the last $k$ registers} under the input garbling Clifford $E_0$ to produce a garbled input. Crucially, the last $k$ registers are designated ``$\mathbf{0}$-state check registers'', which Bob can simply measure in the computational basis to detect if Alice prepared the $\mathbf{0}$ states properly.

Unfortunately, this technique does not extend to checking validity of $\mathbf{T}$ states. To do so, we would have to map $\mathbf{T}$ states to $\mathbf{0}$ states, but there is no Clifford unitary that realizes this transformation.\footnote{The existence of such a Clifford would imply that Clifford + Measurement circuits \emph{without} magic states are universal for quantum computing, contradicting the Gottesman–Knill theorem (assuming $\mathsf{BPP} \neq \mathsf{BQP}$).} The problem with using a non-Clifford unitary is that security of $W$ relies on the fact that it is the product of a random Clifford $C_{A,\mathrm{in}}$ and some other Clifford $W'$. Since the Clifford unitaries form a group, multiplication by a random $C_{A,\mathrm{in}}$ perfectly masks the details of $W'$, but only when $W'$ is Clifford. 

We will therefore employ the ``cut-and-choose'' technique from~\cite{EC:DGJMS20}. The protocol will now have Alice prepare $\lambda(k+1)$-many $\mathbf{T}$ states instead of just $k$. The classical 2PC will generate a random permutation $\pi$ on $[\lambda(k+1)]$, which will move a random selection of $\lambda$ of the $\mathbf{T}$ states into ``$\mathbf{T}$-state check registers.'' The application of $\pi$ will be implemented by a unitary $U_\pi$ incorporated into $W$. After applying $W$, Bob will apply a projective measurement onto $\mathbf{T}$ to each of the $\mathbf{T}$-state check registers, and will abort if any of the $\lambda$ measurements fails. 

If all of the $\lambda$ measurements pass, this means the remaining $\lambda k$ un-tested $\mathbf{T}$ states are ``somewhat close'' to being real $\mathbf{T}$ states. However, being ``somewhat close'' will not be sufficient; for instance, an attacker who prepares exactly one completely invalid $\mathbf{T}$ state will only be caught with $1/(k+1)$ probability. 

We will therefore need to apply magic-state distillation to transform these into states which are negligibly close to real $\mathbf{T}$ states. For this, we use a magic-state distillation circuit of~\cite[\S2.5]{EC:DGJMS20} (which builds on \cite{BravyiKitaev05}). This circuit consists solely of Clifford gates and computational basis measurements. To apply this circuit we modify our underlying functionality, so that we now give out a garbling of a circuit that first implements magic-state distillation and only then applies $Q_B$.

This completes an overview of our protocol, and a formal construction and analysis can be found in Section \ref{sec:three-message}.

\subsection{Application: Reusable MDV-NIZK for QMA}

Now we briefly describe how the above techniques readily give a reusable malicious designated-verifier NIZK for QMA in the CRS model. Note that NIZK for QMA is a special case of two-party quantum computation, where the functionality being computed is the verification circuit $\cV$ for some QMA language, the prover (previously Alice) has the quantum witness $\bw$ as input, and the verifier (previously Bob) has no input and receives a binary output indicating whether $\cV(x,\bw)$ accepts or rejects, where $x$ is the (classical) description of the instance they are considering.

Since the prover does not receive output, there is no need for the third message in the protocol of Section \ref{subsec:3-round-malicious}. Furthermore, since the verifier has no input, there is no need for any quantum message from him in the first message. The verifier only needs to send a first-round classical 2PC message which then functions as a proving key. The (classical) left-over state is the verifier's secret verification key. After this, the prover just sends one quantum message (the Second Message in the above protocol), proving that $\cV(x,\bw) = 1$.

In order to make the above template reusable, we can first instantiate the underlying classical 2PC with a reusable 2PC. Once this is in place, the verifier's first-round message is necessarily instance-indepedent. Then, to ensure that a cheating prover cannot break soundness by observing whether the verifier accepts its proofs or not, we modify the classical functionality to take as input a PRF key from the verifier, and generate all required randomness (used for the $\Zstate$ and $\Tstate$ checks, and the quantum garbling procedure) by applying this PRF to the (classical) description of the instance $x$. By security of the reusable 2PC and the PRF, a verifier will never accept a maliciously sampled proof for any instance $x$ not in the language.


\subsection{Challenges in Achieving a Two-Round Protocol in the Quantum Setting}\label{subsec:two-round-challenges}

The previous sections show that we can achieve 2PQC in two messages if only one party receives output, which is optimal in terms of round complexity. Now we ask whether both parties can obtain output with just two rounds of simultaneous exchange. Indeed, in the classical setting, there is a natural approach to obtaining a two-round protocol, given a two-message protocol where one party receives output. The parties simply run two parallel executions of the two-message protocol on the same inputs - one in which Alice speaks first and the functionality only computes her part of the output, and another in which Bob speaks first and the functionality only computes his part of the output. Unfortunately, this natural approach completely fails in the quantum setting, for at least two reasons.

\begin{itemize}
    \item Running two parallel executions of the same protocol on the same set of inputs seems to require \emph{cloning} those inputs, which is in general impossible if the inputs may be arbitrary quantum states.
    \item Running two parallel executions of a randomized functionality requires the parties to fix the same random coins to be used in each execution, as otherwise their outputs may not be properly jointly distributed. This is not possible in the quantum setting, since randomness can come from measurement, and measurement results cannot be fixed and agreed upon beforehand.
\end{itemize}

These issues motivate the rest of our work. Since running two protocols in parallel on the same inputs is problematic, we take as our guiding principle that one party must be performing the actual computation at some point in the protocol, and then distributing the outputs. 


Interestingly, while the first issue mentioned above is unique to the setting of quantum inputs, the second issue applies even if the parties wish to compute a quantum circuit over just \emph{classical} inputs, which we regard as a very natural setting. Thus, while this paper focuses on the most general case of secure quantum computation over potentially quantum inputs, we stress that all the results we achieve are the best known even for the classical input setting. Furthermore, note that both issues also exist in the specious setting, so it doesn't appear to be straightforward to achieve two-round 2PQC even in this setting. While the focus of this paper is on the setting of malicious security, exploring these questions in the specious setting is also an interesting direction.

\subsection{A Two-Round Protocol with Pre-Processing} 
\label{sec: two-round protocol tech ovw}

Our next result is a three-round protocol for 2PQC which requires only two \emph{online} rounds of communication, filling in the upper right corner of \cref{table: results}.


In fact, we construct a protocol in which the pre-input phase only consists of a \emph{single} message from Bob to Alice (computed with respect to a CRS). We take our three sequential message protocol as a starting point, and introduce several modifications. The first modification will immediately achieve the goal of removing input-dependence from Bob's first message, and all the subsequent modifications will be necessary to restore correctness and security.

\paragraph{Modification 1: Removing Input-Dependence via Teleportation.} Before sending his first message, Bob samples $n$ EPR pairs, where $n$ is the number of qubits of the input $\bx_B$. We denote these EPR pairs by $(\mathbf{epr}_1,\mathbf{epr}_2)$, where $\mathbf{epr}_1$ denotes the left $n$ qubits, and $\mathbf{epr}_2$ denotes the right $n$ qubits. In place of sending $C_{B,\mathrm{in}}(\bx_B,\Zstate^\secp)$, Bob sends $\bm_{B,1} \coloneqq C_{B,\mathrm{in}}(\mathbf{epr}_1,\Zstate^\secp)$. Note that the classical 2PC only requires input $C_{B,\mathrm{in}}$, which is a random Clifford that Bob samples for himself, so Bob's entire first round message $(\bm_{B,1},m_{B,1})$ can now be sent \emph{before} Bob receives his input. The idea is that later on, when Bob learns his input $\bx_B$, he will perform Bell measurements on $(\bx_B,\mathbf{epr}_2)$ to teleport $\bx_B$ into $\mathbf{epr}_1$.

\paragraph{Issue: Incorporating Bob's Teleportation Errors.} Teleporting $\bx_B$ into $\mathbf{epr}_1$ will require Bob to somehow correct $\mathbf{epr}_1$ later in the protocol using the results of his Bell measurements on $(\bx_B,\mathbf{epr}_2)$. But enabling Bob to do this in a way that does not compromise security will be tricky, as we now explain.

After receiving the second round message from Alice in our original malicious 2PQC protocol, Bob learns the output of the classical 2PC, which includes (1) a (classical description of a) quantum garbled circuit $\widetilde{Q}$, and (2) a Clifford unitary $W$. Bob applies $W$ to Alice's quantum message $\bm_{A,2}$, performs the appropriate $\mathbf{0}$ and $\mathbf{T}$ state checks, and conditioned on the checks passing, is left with a state of the form $E_0(\bx_A,\bx_B,\widehat{\mathbf{T}},\mathbf{0})$, where $\widehat{\mathbf{T}}$ is a state ``somewhat close'' to $\mathbf{T}^k$. But at this point in our newly modified protocol, Bob is holding the state $E_0(\bx_A,\mathbf{epr}_1,\widehat{\mathbf{T}},\mathbf{0})$. To restore correctness, we somehow need to modify the protocol so that Bob can apply $X^{x_{\mathrm{inp}}}Z^{z_{\mathrm{inp}}}$ to $\mathbf{epr}_1$ ``inside'' the $E_0$ mask, where $x_{\mathrm{inp}},z_{\mathrm{inp}}$ are the result of Bell basis measurements on $(\bx_B,\mathbf{epr}_2)$.

Recall that the structure of $W$ is $W = E_0 \cdot U_{\mathrm{dec-check}}^\dagger$, where $E_0$ is the input garbling Clifford for the quantum garbled circuit, and $U_{\mathrm{dec-check}}$ is the matrix that undoes $C_{A,\mathrm{in}}$, undoes $C_{B,\mathrm{in}}$, and then applies a permutation $\pi$ and a random linear map $M$, and rearranges all the to-be-checked registers to the last few (rightmost) register slots. The multiplication by $E_0$ is applied only to the non-checked registers. 

Thus, it seems like correctness would have to be restored by inserting the unitary $(\bbI \otimes X^{x_{\mathrm{inp}}}Z^{z_{\mathrm{inp}}} \otimes \bbI)$ in between $E_0$ and $U_{\mathrm{dec-check}}^\dagger$. But if Bob can learn $E_0(\bbI \otimes X^{x_{\mathrm{inp}}}Z^{z_{\mathrm{inp}}} \otimes \bbI)U_{\mathrm{dec-check}}^\dagger$ for even two different values of $x_{\mathrm{inp}}$ and $z_{\mathrm{inp}}$, security of the input garbling Clifford $E_0$ may be lost entirely. 



\paragraph{Modification 2: Classical Garbling + Quantum Multi-Key Fully Homomorphic Encryption} In order to resolve this issue, we will split up the matrix $E_0(\bbI \otimes X^{x_{\mathrm{inp}}}Z^{z_{\mathrm{inp}}} \otimes \bbI)U_{\mathrm{dec-check}}^\dagger$ into two matrices 
\begin{align*}
U_{x_{\mathrm{inp}},z_{\mathrm{inp}}} &\coloneqq E_0(\bbI \otimes X^{x_{\mathrm{inp}}}Z^{z_{\mathrm{inp}}}\otimes \bbI)U_{\mathrm{rand}}^\dagger\\ U_{\mathrm{check}} &\coloneqq U_{\mathrm{rand}}U_{\mathrm{dec-check}}^\dagger
\end{align*}
where $U_{\mathrm{rand}}$ is a ``re-randomizing'' Clifford. 

The matrix $U_{\mathrm{check}}$ is independent of Bob's teleportation errors, and will now be output to Bob by the classical 2PC. But to preserve security, we will have Bob obtain $U_{x_{\mathrm{inp}},z_{\mathrm{inp}}}$ by evaluating a \emph{classical} garbled circuit $\widetilde{f}_{\mathrm{inp}}$ where $f_{\mathrm{inp}}(x_{\mathrm{inp}},z_{\mathrm{inp}}) \coloneqq U_{x_{\mathrm{inp}},z_{\mathrm{inp}}}$; the garbled circuit $\widetilde{f}_{\mathrm{inp}}$ is included in the output of the classical 2PC.

But now we are faced with a new problem: how does Bob obtain the (classical) labels for $\widetilde{f}_{\mathrm{inp}}$? Since we only have one round of interaction remaining, Bob won't be able to run an OT to learn the correct labels (Bob could learn the labels by the end of the two online rounds, but then we would still need another round for Bob to send Alice her encrypted output).

We resolve this problem with \emph{quantum multi-key fully-homomorphic encryption} ($\mathsf{QMFHE}$), which we will use in tandem with our classical garbled circuit $\widetilde{f}_{\mathrm{inp}}$ to enable Bob to compute (a homomorphic encryption of) $U_{x_\mathrm{inp},z_\mathrm{inp}}$ without leaking anything else. Before we continue, we give a brief, intuition-level recap of $\mathsf{QMFHE}$ (we refer the reader to \ifsubmission\else~\cref{sec: qmfhe}\fi for a formal description). Recall that a standard fully-homomorphic encryption ($\mathsf{FHE}$) allows one to apply arbitrary efficient computation to encrypted data (without needing to first decrypt). \emph{Multi-key} $\mathsf{FHE}$ ($\mathsf{MFHE}$) extends $\mathsf{FHE}$ to enable computation over multiple ciphertexts encrypted under different keys; the output of such a homomorphic computation is a ``multi-key'' ciphertext which can only be decrypted given all the secret keys for all of the ciphertexts involved in the computation~\cite{STOC:LopTroVai12}. Finally, $\mathsf{QMFHE}$ extends $\mathsf{MFHE}$ a step further to allow arbitrary efficient \emph{quantum} computation over encrypted (classical or quantum) data~\cite{EPRINT:Goyal18,C:Brakerski18,FOCS:Mahadev18b,Agarwal2020PostQuantumMC}.

We will encrypt each of the garbled circuit labels for $\widetilde{f}_{\mathrm{inp}}$ under an independent $\mathsf{QMFHE}$ key. All of these encrypted labels along with the corresponding $\mathsf{QMFHE}$ public keys (to enable quantum computations over these ciphertexts) will also be output to Bob as part of the classical 2PC. We remark that this requires a $\mathsf{QMFHE}$ scheme where encryptions of classical plaintexts are themselves classical; such schemes are known assuming the quantum hardness of the learning with errors (QLWE) assumption~\cite{Agarwal2020PostQuantumMC}.\footnote{We only require \emph{leveled} $\mathsf{QMFHE}$, which can be based solely on the QLWE assumption. Unleveled $\mathsf{QMFHE}$ requires an additional circularity security assumption.}

To recap, Bob obtains from the classical 2PC a collection of $\mathsf{QMFHE}$ ciphertexts, one for each of the garbled circuit labels for $\widetilde{f}_{\mathrm{inp}}$. Bob picks out the ciphertexts corresponding to $x_{\mathrm{inp}},z_{\mathrm{inp}}$ and performs quantum multi-key evaluation of $\widetilde{f}_{\mathrm{inp}}$ over these ciphertexts, obtaining a $\mathsf{QMFHE}$ encryption of the output of $\widetilde{f}_{\mathrm{inp}}$, i.e. $\mathsf{QMFHE}.\mathsf{Enc}(\pk_{x_{\mathrm{inp}},z_{\mathrm{inp}}},U_{x_{\mathrm{inp}},z_{\mathrm{inp}}})$ where $\pk_{x_{\mathrm{inp}},z_{\mathrm{inp}}}$ denotes the collection of $\mathsf{QMFHE}$ public keys corresponding to $x_{\mathrm{inp}},z_{\mathrm{inp}}$. The classical 2PC output also includes $U_{\mathrm{check}}$ in the clear, which Bob can apply to $\bm_{A,2}$ to obtain $U_{\mathrm{rand}}(\bx_A,\mathbf{epr}_1,\widehat{\mathbf{T}},\mathbf{0})$ (after performing appropriate measurement checks). Then Bob can homomorphically compute the ciphertext $\mathsf{QMFHE}.\mathsf{Enc}(\pk_{x_{\mathrm{inp}},z_{\mathrm{inp}}},E_0(\bx_A,\bx_B,\widehat{\mathbf{T}},\mathbf{0}))$, and proceed to homomorphically evaluate his quantum garbled circuit to obtain $\mathsf{QMFHE}.\mathsf{Enc}(\pk_{x_{\mathrm{inp}},z_{\mathrm{inp}}},(C_{A,\mathrm{out}}(\by_A,\Zstate^\secp),\by_B))$.

In order for Bob to obtain his final output in the clear, we will have Bob send Alice $x_{\mathrm{inp}},z_{\mathrm{inp}}$ in the first online round. In response, in the second online round Alice will reply with $\sk_{x_{\mathrm{inp}},z_{\mathrm{inp}}}$; security of the $\mathsf{QMFHE}$ will guarantee that Bob cannot decrypt ciphertexts corresponding to any other choice of the teleportation errors. In the second online round, Bob will send Alice $\mathsf{QMFHE}.\mathsf{Enc}(\pk_{x_{\mathrm{inp}},z_{\mathrm{inp}}},(C_{A,\mathrm{out}}(\by_A,\Zstate^\secp))$, which she can decrypt to obtain $\by_A$. Finally, Bob produces his output by performing $\mathsf{QMFHE}$ decryption with $\sk_{x_{\mathrm{inp}},z_{\mathrm{inp}}}$.

\paragraph{Issue: Simulating a Quantum Garbled Circuit with Unknown Output.} At this point, we have a correct protocol whose first round is completely input-independent. However, we will run into issues when attempting to prove malicious security.


The problem arises in the security proof for a malicious Bob. In the original three-round maliciously secure protocol, the simulator is able to extract $\bx_B$ from Bob’s first round message to Alice; this is done by first extracting $C_{B,\mathrm{in}}$ from Bob’s first round classical message for the classical 2PC, and then applying $C_{B,\mathrm{in}}^{-1}$ to Bob’s first round quantum message. Extracting $\bx_B$ from Bob’s first round message to Alice is crucial for proving security, since it enables the simulator to query the ideal functionality on $\bx_B$, learn the output $\by_B$, and finally simulate the quantum garbled circuit using Bob's output $\by_B$ before computing Alice's simulated second round message to be sent to Bob. This second round message reveals to Bob the quantum garbled circuit, so it is crucial that the quantum garbled circuit simulator has been executed at this point.



Not surprisingly, this simulation strategy runs into a major problem in our newly modified protocol. Bob’s first message is independent of $\bx_B$, so the simulator cannot query the ideal functionality, and therefore seemingly cannot simulate the quantum garbled circuit before computing Alice's message, which in particular reveals the quantum garbled circuit to Bob. In summary, the simulator must provide Bob with the quantum garbled circuit (part of Alice's first online round message), \emph{before} it has enough information to extract Bob's input. This appears quite problematic since simulating a garbled circuit certainly requires knowing the output. However, since Bob can only obtain an \emph{encryption} of the output of the garbled circuit after receiving Alice's first message, it is still reasonable to expect that the protocol is secure.

\paragraph{Modification 3: Simulation via Teleportation.} We fix this problem through a new technique we call \emph{simulation via teleportation}. The idea is as follows. Instead of running the quantum garbled circuit simulator on the output of the circuit (which the simulator does not yet know), the simulator will first prepare fresh EPR pairs $\mathbf{epr}_1’,\mathbf{epr}_2’$ and then run the quantum garbled circuit simulator on $(C_{A,\mathrm{out}}(\Zstate,\mathbf{0}^\secp),\mathbf{epr}_1’)$ (where $\Zstate$ takes the place of Alice's input $\bx_A$ and $\mathbf{epr}_1’$ takes the place of Bob's output $\by_B$). In the following round, after Bob has teleported over his input state $\bx_B$, the simulator will query the ideal functionality, learn $\by_B$, and then \emph{teleport $\by_B$ into $\mathbf{epr}_1’$}. 

Implementing the final teleportation step requires some care. When the simulator learns $\by_B$, it performs Bell measurements on $(\by_B,\mathbf{epr}_2’)$, obtaining measurement outcomes $x_{\mathrm{out}},z_{\mathrm{out}}$. It must then find some way to apply $x_{\mathrm{out}},z_{\mathrm{out}}$ to the state $\mathbf{epr}_1’$ so that Bob can obtain his correct output. 


So we further modify the protocol so that the garbled circuit Bob receives from the classical 2PC is modified to output $(C_{A,\mathrm{out}}(\by_A,\Zstate^\secp),X^{x_{\mathrm{out}}}Z^{z_{\mathrm{out}}}\by_B)$ instead of $(C_{A,\mathrm{out}}(\by_A,\Zstate^\secp),\by_B)$, as before. That is, in the real protocol, an honest Alice will sample random $x_{\mathrm{out}},z_{\mathrm{out}}$, and then the 2PC will output the circuit implementing this functionality. Alice will send $x_{\mathrm{out}},z_{\mathrm{out}}$ to Bob in the second online round, and Bob will first apply Pauli corrections $X^{x_{\mathrm{out}}}Z^{z_{\mathrm{out}}}$ to his output to obtain $\by_B$. In the simulated protocol, however, $x_{\mathrm{out}},z_{\mathrm{out}}$ are not sampled by the simulator. Instead, they are the result of the simulator's Bell measurements on $(\by_B,\mathbf{epr}_2’)$. The simulator thus simulates a garbled circuit that outputs $(C_{A,\mathrm{out}}(\Zstate,\Zstate^\secp),\mathbf{epr}_1’)$, and then sends $x_{\mathrm{out}},z_{\mathrm{out}}$ in the second online round. Note that this teleportation step occurs \emph{exclusively within the simulation}.

\paragraph{Modification 4: Alice (Equivocally) Commits to Pauli Corrections.} To arrive at a fully secure protocol, we need to address one last issue. As currently described, there is nothing that prevents a malicious Alice from misreporting her choice of $x_{\mathrm{out}}, z_{\mathrm{out}}$. This can introduce arbitrary Pauli errors into Bob’s output that he has no way of detecting. However, this can easily be fixed using equivocal commitments. That is, Alice inputs $x_{\mathrm{out}},z_{\mathrm{out}}$ to the classical 2PC, along with commitment randomness $s$. Bob obtains the commitment as part of the output of the classical 2PC, and later when Alice sends $x_{\mathrm{out}},z_{\mathrm{out}}$ in the second online round, she must also send along $s$. The equivocality property enables the simulation strategy to work as before, as the simulator will have the power to send Bob a commitment to an arbitrary value, and after learning $x_{\mathrm{out}},z_{\mathrm{out}}$ from its Bell measurements, use equivocation to produce a valid opening.

\subsection{The Multi-Party Setting}

In this section, we describe our results in the multi-party setting, filling in the bottom row of \cref{table: results}. 

We begin by describing our approach to obtaining a five-round protocol from quantum-secure OT. Our approach follows the same high-level idea as the three-message 2PQC protocol described in Section \ref{subsec:3-round-malicious}, where one party (the ``designated party'', or $P_1$) will evaluate a quantum garbled circuit on encodings of each party's input, and then distribute the encoded outputs to each party. However, implementing this template in the multi-party setting requires resolving a host of new challenges.

\paragraph{Input Encoding.} Recall that in our two-party protocol, Alice received an encoding of Bob's input, concatenated their own input, re-randomized the entire set of registers with a random Clifford $C$, and then sent the re-randomized state to Bob. This re-randomization ensures that the only meaningful computation Bob can perform is to apply the quantum garbled circuit, whose classical description is re-randomized with $C^\dagger$. A natural extension of this idea to the multi-party setting goes as follows. First, each party sends their encoded input to $P_1$. Then $P_1$ concatenates all inputs together and re-randomizes the resulting set of registers with their own random Clifford $C_1$. Then, these registers are passed around in a circle, each party $P_i$ applying their own re-randomizing Clifford $C_i$. Finally, $P_1$ receives the fully re-randomized state, along with some classical description of a quantum garbled circuit obtained via classical MPC, and re-randomized with $C_1^\dagger \dots C_n^\dagger$. The fact that each party applies their own re-randomizing Clifford is necessary, since we are in the dishonest majority setting. Indeed, if only one party $P_i$ is honest, their security will crucially rely on the fact that the adversary does not know their re-randomizing Clifford $C_i$. This approach of encrypting and sending a state around the circle of parties for re-randomization is similar to \cite{EC:DGJMS20}'s ``input encoding'' protocol, in which each individual party's input is sent around the circle of parties for re-randomization.

Unfortunately, the round complexity of this encoding step will grow linearly with the number of parties. To obtain a constant-round protocol, our idea is to round-collapse this input-encoding via the use of quantum teleportation. In the first round, parties will send EPR pairs to each other following the topology of the computation described above. That is, each party sets up EPR pairs with $P_1$ that will be used to teleport their encoded inputs to $P_1$, and each consecutive pair of parties will set up EPR pairs that will be used to teleport the encoded state around the circle. After this setup, the parties can \emph{simultaneously} apply re-randomization Cliffords and teleport the encoded state around the circle. This will introduce teleportation errors, but since the re-randomization operations are Clifford, these can be later corrected. Indeed, this correction will be facilitated by a classical MPC protocol that takes as input each party's Clifford and set of teleportation errors.

\paragraph{0 and T State Checks.} The next challenge is how to enforce 0 and $T$ state checks in the multi-party setting. Recall that in the two-party setting, we had the non-evaluator party (Alice) prepare the 0 and $T$ states, which were then checked by the garbled circuit evaluator (Bob). This approach works because we know that if Alice is malicious and tried to cheat during preparation of these states, then Bob must be honest and will then refuse to evaluate the garbled circuit. However, this does not carry over to the multi-party setting. If we try to fix some party $P_i$ to prepare the $0$ and $T$ states and then have the evaluator $P_1$ check them, it may be the case that \emph{both} $P_i$ and $P_1$ are malicious, which would be problematic.

Thus, we take a different approach, instructing $P_1$ to prepare the 0 and $T$ states, and designing a \emph{distributed} checking protocol, similar to that of \cite{EC:DGJMS20}. We now briefly describe the $T$ state check, leaving a description of the 0 state check to the body. $P_1$ will be instructed to concatenate all parties' inputs with their own $T$ states, and then send the resulting state around the circle for re-randomization. Later, they receive the re-randomized state, along with a unitary from the classical MPC that i) undoes the re-randomization, ii) samples a different subset of $T$ states for each party, iii) Clifford-encodes each subset, and iv) garbles the inputs together with the remaining $T$ states. Thus, $P_1$ obtains $n$ encoded subsets of $T$ states, and is supposed to send one to each party. Each party will then receive their encoded subset, decode (using information obtained from the classical MPC), and measure in the $T$-basis. Each party will then abort the protocol if their check failed. Only if \emph{no} parties abort will the classical MPC send information to each party allowing them to decrypt their output from the quantum garbled circuit. It is crucial that \emph{no} party receives output until all honest parties indicate that their $T$ state check passed, because using malformed $T$ states in the quantum garbled circuit could result in outputs that leak information about honest party inputs.

\paragraph{The Five-Round Protocol.} We give a high-level overview of the five rounds of the protocol.

\begin{itemize}
    \item Round 1: Each party $P_i$ generates EPR pairs and sends half of each pair to its neighbor $P_{i+1}$. Additionally, party $P_1$ generates enough EPR pairs so that it can send EPR pair halves to every other party $P_i$ for $i \neq 1$.
    \item Round 2: Teleport inputs to $P_1$ and teleport the resulting state around the circle (with re-randomization Cliffords $C_i$ applied along the way). Input teleportation errors and $\{C_i\}_{i \in [n]}$ to the classical MPC.
    \item Round 3: Classical MPC delivers unitary to $P_1$ that samples subsets of $T$ states and garbles inputs, along with classical description of the quantum garbled circuit.
    \item Round 4: $P_1$ evaluates the unitary and garbled circuit, then delivers encoded subsets of $T$ states and encrypted outputs to each party.
    \item Round 5: If no parties abort after their $T$ state check, the classical MPC delivers key to each party allowing them to decrypt their output.
\end{itemize}

Note that the distributed $T$ state check is the reason that the protocol requires five rounds. The first round is used for setting up EPR pairs. At this point the parties can perform quantum teleportation and obtain their Pauli errors. Now, these must be corrected by the classical MPC, which takes a minimum of two rounds. Thus, $P_1$ can only obtain output from the MPC, and thus from the quantum garbled circuit, after Round 3. Then, Round 4 must be used to distribute subsets of $T$ states, and Round 5 must be used to deliver decryption keys conditioned on all parties being happy with their $T$ states. As we describe in the body, the actual computation of the garbled circuit can be delayed one round (at the cost of settling for security with abort rather than unanimous abort), giving a five-round protocol with three online rounds.

Now we discuss how to instantiate the classical MPC. We are going to need an MPC that supports \emph{reactive} functionalities, where inputs may depend on previous outputs obtained from the MPC. Moreover, we need the MPC to be \emph{round-optimal}, in the sense that outputs delivered in round $i$ may depend on inputs from round $i-1$. We observe that the round-collapsing compiler of \cite{EC:GarSri18a} gives exactly this --- an $\ell+1$ round MPC for a reactive functionality with $\ell$ rounds of output. Thus, we can rely solely on quantum-secure two-message OT to construct the above five-round quantum MPC. \ifsubmission\else See~\cref{subsec:reactivempc} for more discussion about the classical reactive MPC.\fi

\paragraph{The Four-Round Protocol.} Finally, we observe that there is some slack in the aforementioned protocol. Indeed, $P_1$ does not obtain any output from the classical MPC until after round 3, when in principle the classical MPC can be used to compute some output in only two rounds. The reason we waited three rounds is that we wanted to include the parties' teleportation errors in the computation performed by the MPC, and these are not known until the beginning of the second round.

However, we can use ideas similar to those in \cref{sec: two-round protocol tech ovw} in order to allow the MPC to compute something meaningful during the first two rounds without yet knowing the teleportation errors. In particular, we make use of classical garbled circuits and quantum multi-key FHE to provide a mechanism by which the classical MPC can output information allowing $P_1$ to (homomorphically) compute a function of the teleportation errors after Round 2. This allows us to collapse the total number of required rounds to 4. Moreover, a similar idea allows the parties to delay teleportation of their inputs another round, giving a four-round protocol with (optimal) \emph{two} rounds of online interaction. Equivalently, our protocol can be seen as two-round MPQC in a quantum pre-processing model.

\subsection{Two Round 2PQC Without Pre-Processing: Challenges and Possibilities}
\label{subsec:tech-overview-two-round}
In this section, we explore the possibility of achieving a two-round 2PQC protocol in the CRS model \emph{without pre-processing}. We stress that this model \emph{does not permit pre-shared entanglement} between the two parties, as we consider sharing of entanglement to be a pre-processing step.

\paragraph{The Challenge of Oblivious Simulation.} In the classical setting, all known two-round two-party computation protocols (in the CRS model) can be modified so that security is proven via (what we call) an \emph{oblivious simulator}.\footnote{Each party will use a NIZK proof of knowledge to prove that their first message is well-formed, using their input and randomness as witness. Then, a simulator programming the CRS may extract either party's input.} That is, the simulator (1) only makes black-box queries to the adversary, (2) is straight-line (meaning it only runs the adversary a single time without rewinding), and (3) it generates the simulated CRS \emph{independently of the choice of corrupted party} (between Alice and Bob).

By focusing on protocols with oblivious simulation, we can highlight an apparent difficulty of building secure two-round protocols for quantum functionalities in the CRS model. Assume without loss of generality that Alice is adversarial (the identical argument applies to Bob). Observe that if the first message that Alice sends is not computationally binding to her input $\mathbf{x}_A$, she can potentially cheat by \emph{equivocating}, i.e. acting as if she had received a different input, and subsequently learn multiple outputs of the functionality. If the simulation is oblivious, then this reasoning applies simultaneously to Alice and Bob --- that is, both parties must, in the first round, send computationally-binding commitments to their respective inputs. This is immediately problematic for quantum inputs, since no-cloning implies that their leftover states will have no (computationally) useful information about their original inputs. Thus, it is unclear how a general computation can be performed on their \emph{joint} inputs before the start of the second round, as the parties have effectively swapped their initial states. And somehow, after just one more round of messaging, they must hold their correctly computed output states.

Our negative result formalizes this intuitive difficulty. If the simulator is oblivious, then by roughly following the above reasoning, at the end of the first round:
\begin{itemize}
\item Alice holds a computationally binding commitment to Bob's input $\mathbf{x}_B$,
\item Bob holds a computationally binding commitment to Alice's input $\mathbf{x}_A$, and
\item Neither party has information about their original inputs.
\end{itemize}
Moreover, the correctness of oblivious simulation implies that for a computationally indistinguishable CRS, there exists a ``trapdoor'' that would enable Alice to extract $\mathbf{x}_B$ and would enable Bob to extract $\mathbf{x}_A$. But now their states can be viewed as the states of two parties at the \emph{beginning of a one-round protocol with polynomial-size pre-processing} in which the parties' inputs are \emph{swapped}; the pre-processing step is necessary to give both parties the trapdoor information of the simulator. The resulting one-round protocol no longer satisfies any meaningful security guarantees, but crucially, it still satisfies correctness. Moreover, the one-round protocol falls into a model of ``instantaneous non-local computation'' that has been previously studied in the quantum information literature~\cite{Beigi_2011}. It is currently open whether this model enables general quantum computation with only polynomial-size preprocessing, and a positive result for two-round 2PQC with oblivious simulation would affirmatively answer this question.

\paragraph{A Proof-of-Concept Construction from Quantum VBB Obfuscation.}

\ifsubmission
\else
Given the above barrier, one could attempt to construct a two-round protocol whose security relies crucially on a \emph{non-oblivious} simulation strategy. In this work, we take an initial step in this direction by providing a proof-of-concept construction from a strong form of quantum VBB obfuscation that handles obfuscation of quantum circuits that include both unitary gates and measurement gates (see \cref{defn:vbb} and the discussion preceding it).
\fi

In our construction, Alice will send an encryption of her input to Bob in round 1, who will then homomorphically compute the functionality over their joint inputs and respond with Alice's encrypted output in round 2. Alice will also send a message in round 2 that allows Bob to decrypt his output. However, the key is that this interaction will actually be indistinguishable from an interaction in which the \emph{opposite} flow of computation is occuring. In particular, if the CRS if sampled differently (but in an indistinguishable way), it will be the case that Bob is actually sending his encrypted input to Alice in the first round, and then Alice homomorphically computes the functionality and sends Bob's encrypted output back in the second round.

To instantiate this template, we provide a number of quantum obfuscations in the CRS, three per party. First, there are the  ``input'' obfuscations $\cO_{A,\inp}$ and $\cO_{B,\inp}$. $\cO_{A,\inp}$ will take as input Alice's input $\bx_A$ along with a ``dummy'' input $\bd_A$, and output Clifford encodings of each. Alice is instructed to send the first output of this obfuscation as her first message, and keep the second output as her state. In the real protocol, the obfuscated functionality will be such that the first output will be the Clifford encoding of the first input (Alice's real input $\bx_A$), and the second output will be the Clifford encoding of the second input (Alice's dummy input $\bd_A$). On the other hand, $\cO_{B,\inp}$ will obfuscate the functionality that does the exact opposite, setting its first output to be a Clifford encoding of its second input, and its second output to be a Clifford encodings of its first input. Thus, in round 1, Alice sends a Clifford encoding of her real input and keeps a Clifford encoding of her dummy input in her state, while Bob sends a Clifford encoding of his dummy input and keeps a Clifford encoding of his real input in his state.

The next obfuscations $\cO_{A,\cmp}$ and $\cO_{B,\cmp}$ share secret randomness with the input obfuscations (in the form of PRF keys) and can thus decrypt Clifford encodings output by the input obfuscations. They each are defined to decrypt and check the authenticity of their inputs, apply the functionality $Q$ that the parties wish to compute, and then encode the outputs with freshly sampled Cliffords. Each party will run their respective obfuscation on their state and the other party's first round message. Note that then Alice is just using $\cO_{A,\cmp}$ to compute $Q$ over dummy inputs, while Bob is using $\cO_{B,\cmp}$ to compute $Q$ over their real inputs. Alice will send an encrypted dummy output to Bob in round 2, while Bob will send an encrypted real output to Alice.

Finally, each party applies their respective output obfuscation $\cO_{A,\out}$ and $\cO_{B,\out}$ to their final state and other party's second round message. $\cO_{A,\out}$ will ignore Alice's state (which contains Alice's dummy output) and decrypt and output Bob's second round message (which contains Alice's real output). On the other hand, $\cO_{B,\out}$ will ignore Alice's second round message and decrypt and output Bob's state.

Now, it is possible to argue (under the assumption that the obfuscations in the CRS are in fact VBB obfuscations\ifsubmission\else\footnote{Attempting to prove this based on just indistinguishability obfuscation runs into issues that arise due to the inherently probabilistic nature of the functionalities obfuscated. In particular, they generate randomness via measurement and then use this randomness to generate Clifford matrices. In the classical setting, one could usually generate the required randomness with a PRF applied to the input, but it is unclear how to do this when the input is a quantum state.}\fi) that, because all intermediate states and messages are Clifford-encoded, ``switching the direction'' of the input and output obfuscations cannot be noticed by the parties. Note that if each of $\cO_{A,\inp}$ and $\cO_{B,\inp}$ are re-defined to permute the order of their outputs, then the flow of computation will be completely reversed. In particular, Alice will be computing the functionality over real inputs with $\cO_{A,\cmp}$, and Bob will be computing the functionality over dummy inputs with $\cO_{B,\cmp}$. Thus, depending on how the simulator programs the CRS, it can either extract directly from Alice's first round message OR it can extract directly from Bob's first round message, but it could never extract from both simultaneously. 

Thus, this template represents a potential method for securely computing a quantum functionality in two rounds, where one of the two parties actually performs the computation between rounds 1 and 2 and then distributes the output in round 2. In other words, it is an instantiation of our guiding principle mentioned in \cref{subsec:two-round-challenges} in a model without pre-processing.

\ifsubmission\else Of course, since VBB obfuscation of quantum circuits is in general impossible \cite{Alagic2016OnQO}, one may wonder how to interpret this result. One may view this construction, in conjunction with our impossibility result for oblivious simulators, as suggesting a particular template for designing two-round 2PQC that with new ideas may eventually be instantiated to give a construction from plausible assumptions. On the other hand, one may view the construction as a potential barrier to obtaining a more general impossibility result. Indeed, showing that it is impossible to securely compute a particular quantum functionality $Q$ in two rounds now requires showing that (strong) VBB obfuscation of certain functionalities is impossible. Currently, we only know that some very specific functionalities are un-obfuscatable~\cite{Alagic2016OnQO,alagic2020impossibility,ananth2020secure}.  \fi

\section{Preliminaries}

\subsection{Notation}
Following \cite{ARXIV:BrakerskiYuen20}, we define a Quantum Random Variable, or QRV, to be a density matrix $\bx$ on register $\gray{X}$. We will generally refer to QRVs with lowercase bold font and to registers with uppercase gray font. A collection of QRVs $(\bx,\by,\bz)$ on registers $\gray{X},\gray{Y},\gray{Z}$ is also a QRV, and $\bx,\by,\bz$ may or may not be entangled with each other.

Let $\lambda$ denote the security parameter. We will consider non-uniform quantum polynomial-time adversaries, denoted by $\cA = \{\cA_\secp,\brho_\secp\}_{\secp \in \bbN}$, where each $\cA_\secp$ is the classical description of a $\poly(\secp)$-size quantum circuit, and each $\brho_\secp$ is some (not necesarily efficiently computable) non-uniform $\poly(\secp)$-qubit quantum advice.

We will denote the trace distance between two QRVs $\bx$ and $\by$ with $\|\bx - \by\|_1$ and for infinite sequences of QRVs $\{\bx_\secp\}_{\secp \in \bbN}$ and $\{\by_\secp\}_{\secp \in \bbN}$ we write $$\{\bx_\secp\}_{\secp \in \bbN} \approx_s \{\by_\secp\}_{\secp \in \bbN}$$ to indicate that there exists a negligible function $\mu(\cdot)$ such that $||\bx_\secp - \by_\secp||_1 \leq \mu(\secp)$. Here, the $s$ refers to ``statistical'' indistinguishability.

In addition, we write $$\{\bx_\secp\}_{\secp \in \bbN} \approx_c \{\by_\secp\}_{\secp \in \bbN}$$ to indicate that there exists a negligible function $\mu(\cdot)$ such that for all QPT distinguishers $\cD = \{\cD_\secp,\bd_\secp\}_{\secp}$, $$\left|\Pr[\cD_\secp(\bd_\secp,\bx_\secp) = 1] - \Pr[\cD_\secp(\bd_\secp,\by_\secp) = 1] \right| \leq \mu(\secp).$$ Here, the $c$ refers to ``computational'' indistinguishability.

Let $\mathscr{C}_n$ and $\mathscr{P}_n$ denote the $n$-qubit Clifford and Pauli groups, respectively. Let $\Zstate$ refer to a 0 state and $\Tstate$ refer to a $T$ state. $\Zstate^n$ denotes $n$ copies of a single qubit 0 state and likewise for $\Tstate^n$. Let $X$ and $Z$ be the Pauli matrices, i.e.
\[ X = \begin{bmatrix} 0 & 1 \\ 1 & 0 \end{bmatrix}, Z = \begin{bmatrix} 1 & 0 \\ 0 & -1 \end{bmatrix}. \]

\subsection{Clifford Authentication Code}

\begin{definition}[Clifford Authentication Code]
The $n$-qubit $\secp$-trap Clifford authentication code consists of the following algorithms, which encode an $n$-qubit state $\bx$ with key $C \in \mathscr{C}_{n+\secp}$.
\begin{itemize}
    \item $\Enc(C,\bx)$: Compute $C(\bx,\Zstate^\secp) \coloneqq \widehat{\bx}$.
    \item $\Dec(C,\widehat{\bx})$: Compute $(\bx',\by) \coloneqq C^\dagger(\widehat{\bx})$ (where $\bx'$ is the first $n$ registers of the result and $\by$ is the final $\secp$) and measure $\by$ in the standard basis. If the outcome is $0^\secp$, return $\bx'$, and otherwise return $\ket{\bot}\bra{\bot}$.
\end{itemize}
This authentication code satisfies the following property. For any QRV $(\bx,\bz)$ and any CPTP map $\cA$ acting on the encoding and side-information $\bz$, there exist maps $\cB_0$, $\cB_1$ acting on $\bz$ such that $\cB_0+\cB_1$ is CPTP (completely positive trace preserving), and 

$$\left\|\E_{C \gets \mathscr{C}_{n+\secp}}\left[\Dec(C,\cA(\Enc(C,\bx),\bz))\right] - \left((\bx,\cB_0(\bz)) + \left(\ket{\bot}\bra{\bot}, \cB_1(\bz)\right)\right)\right\|_1 = \negl(\secp).$$

\end{definition}

\subsection{Multi-Party Quantum Computation}

Below we give a definition of maliciously-secure multi-party quantum computation with abort, following the standard real/ideal world paradigm for defining secure computation~\cite{Goldbook}. It can be strengthened to security with \emph{unanimous} abort by changing the interface of the ideal functionality $\cI$ so that its final input is $\abort$ or $\mathsf{ok}$, and it either outputs all honest party outputs (if $\mathsf{ok}$) or none (if $\abort$).

Consider an $n$-party quantum functionality specified by a family of quantum circuits $\cQ = \{Q_\secp\}_{\secp \in \bbN}$ where $Q_\secp$ has $m_1(\secp) + \dots + m_n(\secp)$ input qubits and $\ell_1(\secp) + \dots + \ell_n(\secp)$ output qubits. We will consider a QPT adversary $\cA = \{\cA_\secp\}_{\secp \in \bbN}$ that corrupts any subset $M \subset [n]$ of parties.


Let $\Pi$ be an $n$-party protocol for computing $Q$. For security parameter $\secp$ and any collection of (potentially entangled) quantum states $(\bx_1,\dots,\bx_n,\baux_\cA,\baux_\cD)$, where $\bx_i$ is $P_i$'s input to $Q_\secp$ (on $m_i(\secp)$ registers), $\baux_\cA$ is some side information (on an arbitrary number of registers) given to the adversary, and $\baux_\cD$ is some information (on an arbitrary number of registers) given to the distinguisher, we define the quantum random variable $\Real_{\Pi,\Q}(\cA_\secp,\{\bx_i\}_{i \in [n]},\baux_\cA)$ as follows. $\cA_\secp(\{\bx_i\}_{i \in M},\baux_\cA)$ interacts with honest party algorithms on inputs $\{\bx_i\}_{i \in [n] \setminus M}$ participating in protocol $\Pi$, after which the honest parties output $\{\by_i\}_{i \in [n] \setminus M}$ and $\cA$ outputs a final state $\bz$ (an arbitrary function computed on an arbitrary subset of the registers that comprise its view). The random variable $\Real_{\Pi,\Q}(\cA_\secp,\{\bx_i\}_{i \in [n]},\baux_\cA)$ then consists of $\{\by_i\}_{i \in [n] \setminus M}$ along with $\bz$.

For any $\cA$, we require the existence of a simulator $\Sim= \{\Sim_\secp\}_{\secp \in \bbN}$ that takes as input $(\{\bx_i\}_{i \in M},\baux_\cA)$, has access to an ideal functionality $\cI[\{\bx_i\}_{i \in [n] \setminus M}](\cdot)$, and outputs a state $\bz$. The ideal functionality accepts an input $\{\bx_i\}_{i \in M}$, applies $Q_\secp$ to $(\bx_1,\dots,\bx_n)$ to recover $(\by_1,\dots,\by_n)$, and returns $\{\by_i\}_{i \in M}$ to $\Sim_\secp$. Then, for each $i \in [n] \setminus M$, it waits for either an $\mathsf{abort}_i$ or $\mathsf{ok}_i$ message from $\Sim_\secp$. In the case of $\mathsf{ok}_i$ it includes $\by_i$ in its output and in the case of $\mathsf{abort}_i$ it includes $\bot$ (note that these outputs are not given to $\Sim_\secp$). Now, we define the quantum random variable $\Ideal_{\Pi,\Q}(\Sim_\secp,\{\bx_i\}_{i \in [n]},\baux_\cA)$ to consist of the output of $\cI[\{\bx_i\}_{i \in [n] \setminus M}](\cdot)$ and the final state $\bz$ of $\Sim_\secp^{\cI[\{\bx_i\}_{i \in [n] \setminus M}](\cdot)}(\{\bx_i\}_{i \in M},\baux_\cA)$.

\begin{definition}[Secure Multi-Party Quantum Computation]
\label{def:mpqc} A protocol $\Pi$ securely computes $Q$ if for all QPT $\cA = \{\cA_\secp\}_{\secp \in \bbN}$ corrupting subset of parties $M \subset [n]$, there exists a QPT $\Sim = \{\Sim_\secp\}_{\secp \in \bbN}$ such that for all $\{\bx_{1,\secp},\dots,\bx_{n,\secp},\baux_{\cA,\secp},\baux_{\cD,\secp}\}_{\secp \in \bbN}$ and all QPT $\cD = \{\cD_\secp\}_{\secp \in \bbN}$, there exists a negligible function $\nu(\cdot)$ such that

\begin{align*}
    &\bigg| \Pr\left[\cD_\secp\left(\baux_{\cD,\secp},\Real_{\Pi,\Q}(\cA_\secp,\{\bx_{i,\secp}\}_{i \in [n]},\baux_{\cA,\secp})\right) = 1\right]\\ &- \Pr\left[\cD_\secp\left(\baux_{\cD,\secp},\Ideal_{\Pi,\Q}(\Sim_\secp,\{\bx_{i,\secp}\}_{i \in [n]},\baux_{\cA,\secp})\right) = 1\right] \bigg| \leq \nu(\secp).
\end{align*}


\end{definition}

\subsection{Useful Lemmas}

\begin{lemma}[Magic State Distillation~\cite{BravyiKitaev05,EC:DGJMS20}]\label{lemma:distillation} Let $p(\cdot)$ be a polynomial. Then there exists a $\poly(\secp)$ size $\CM$ circuit $Q$ from $\secp p(\secp)$ input qubits to $p(\secp)$ output qubits such that the following holds. Take any state $\bx$ on $\secp p(\secp) +\secp$ qubits. Apply a uniformly random permutation to the registers of $\bx$ and then measure the final $\secp$ qubits in the $T$-basis to obtain a bitstring $s$. Let $\widetilde{\bx}$ be the remaining $\secp p(\secp)$ registers. Then there exist negligible functions $\mu,\nu$ such that $$\Pr\left[(s = 0) \wedge \left(\left\|Q(\widetilde{\bx})- \Tstate^{p(\secp)}\right\|_1 > \mu(\secp)\right)\right] \leq \nu(\secp).$$

\end{lemma}
\begin{proof}
This follows from applying~\cite[Lemma I.1]{EC:DGJMS20} with parameters $n = \secp p(\secp)$, $k = \secp$, $\delta = 1/2$ followed by~\cite[Lemma 2.7]{EC:DGJMS20} with parameters $m = \secp p(\secp)$, $\ell = m/2, t=p(\secp)$.
\end{proof}


\begin{lemma}[\cite{EC:DGJMS20}]\label{lemma:linear-map}
For any $n \in \bbN$ and projector $\Pi$ on $2n$ qubits, define the quantum channel $\cL^\Pi$ by $$\cL^\Pi(\bx) \coloneqq \Pi\bx\Pi + \ket{\bot}\bra{\bot}\trace[(\bbI^{2n}-\Pi)\bx],$$ where $\ket{\bot}$ is a distinguished state on $2n$ qubits with $\Pi\ket{\bot} = 0$. For any $t \in \{0,1\}^n$, let $\Pi_{t,\Full} \coloneqq \ket{0^{2n}}\bra{0^{2n}}$ if $t = 0^n$ and $\Pi_{t,\Full} \coloneqq 0$ otherwise. Let $\Pi_{t,\Half} \coloneqq \bbI^n \otimes \ket{t}\bra{t}$. Then for any QRV $\bx$ on $2n$ registers and $t \in \{0,1\}^n$, $$\left\|\cL^{\Pi_{t,\Full}}(\bx) - \E_{U \gets \mathsf{GL}(2n,\bbF_2)}\left[\cL^{\Pi_{t,\Half}}(U(\bx))\right]\right\|_1 = \negl(n).$$
\end{lemma}

This implies the following lemma, which is stated in terms of an interactive game between adversary and challenger.

\begin{lemma}\label{lemma:Ztest}
Consider the following experiment. An adversary $\cA$ outputs a state $\bx^{\gray{M_1},\gray{M_2}}$ on $2n$ qubits. Then, the challenger samples $U \gets \mathsf{GL}(2n,\bbF_2), r,s \gets \{0,1\}^n$, computes $\left(\by_1^{\gray{M_1}},\by_2^{\gray{M_2}}\right) \coloneqq (\bbI \otimes X^rZ^s)U(\bx)$, and returns $\by_2$. Finally, the adversary outputs a string $r'$. The experiment accepts if $r' = r$ and otherwise rejects. Then there exist negligible functions $\mu,\nu$ such $$\Pr\left[(r' = r) \wedge \left(\left\|\by_1 - \mathbf{0}^n \right\|_1 > \mu(n)\right)\right] \leq \nu(n).$$

\end{lemma}

\begin{proof} 
We can specify any adversary by a starting state  $\left(\bx^{\gray{M_1},\gray{M_2}},\bz^{\gray{Z}}\right)$ and an attack unitary $A$, where $A$ is applied to $(\by_2,\bz)$ and then followed by a computational basis measurement to produce $r'$. Then we can write the experiment in the lemma as sampling $U,r,s$, computing $$\left(\bbI^{\gray{M_1}} \otimes X^rZ^s \otimes \bbI^{\gray{Z}}\right)\left(\bbI^{\gray{M_1}} \otimes A^{\gray{M_2},\gray{Z}}\right)\left(\bbI^{\gray{M_1}} \otimes X^rZ^s \otimes \bbI^{\gray{Z}}\right)\left(U \otimes \bbI^{\gray{Z}}\right)(\bx,\bz),$$ measuring the $\gray{M_2}$ register and rejecting if the result is not $0^n$. Now, one can apply the Pauli twirl (see for example \cite[Lemma A.2]{EC:DGJMS20}) to argue that $$\left(\bbI^{\gray{M_1}} \otimes X^rZ^s \otimes \bbI^{\gray{Z}}\right)\left(\bbI^{\gray{M_1}} \otimes A^{\gray{M_2},\gray{Z}}\right)\left(\bbI^{\gray{M_1}} \otimes X^rZ^s \otimes \bbI^{\gray{Z}}\right)$$ can be written as a classical mixture of Pauli attacks on register $\gray{M_2}$. Thus, we can write the state on registers $\gray{M_1},\gray{M_2}$ that results from this experiment as $$\sum_{t \in \{0,1\}^n} p_t \cL^{\Pi_{t,\mathsf{Half}}}(U(\bx)),$$ where $p_t$ is the probability the the $x$-value in the adversary's mixture of Pauli attacks is equal to $t$. Now, by \cref{lemma:linear-map}, this state is negligibly close to $$\sum_{t \in \{0,1\}^n} p_t \cL^{\Pi_{t,\mathsf{Full}}}(\bx).$$ Finally, observe that this state is either $\ket{\bot}\bra{\bot}$ (a reject), or $\ket{0^{2n}}\bra{0^{2n}}$ (an accept), so the event stated in the lemma will occur with probability 0.

\end{proof}

\subsection{Two-Message Two-Party Classical Computation}
\label{subsec:2pc}


As a building block, we will use post-quantum maliciously-secure two-message two-party classical computation in the CRS model where one party receives output. We will require that the simulator is straight-line and black-box. We will refer to such a protocol simply as $\twopc$.

$\twopc$ is defined by four algorithms $(\twopc.\gen,\twopc_1,\twopc_2,\twopc_\out)$. We will keep the convention that party $B$, with input $x_B$, first computes $\twopc_1$, then party $A$, with input $x_A$, computes $\twopc_2$, and finally party $B$ recovers its output $y$ with $\twopc_\out$. The syntax of these algorithms is as follows, where $C$ is the description of the circuit to be computed.

\begin{itemize}
    \item $\crs \gets \twopc.\gen(1^\secp)$.
    \item $(m_1,\st) \gets \twopc_1(1^\secp,C,\crs,x_B)$.
    \item $m_2 \gets \twopc_2(1^\secp,C,\crs,m_1,x_A)$.
    \item $y \gets \twopc_\out(1^\secp,\st,m_2)$.
\end{itemize}

Let $C = \{C_\secp\}_{\secp \in \bbN}$ be a (potentially randomized) family of classical circuits where $C_\secp$ takes as input $(x_A,x_B) \in \{0,1\}^{n_A(\secp) + n_B(\secp)}$ and outputs $y$. Consider the case of an adversary $\cA = \{\cA_\secp\}_{\secp \in \bbN}$ corrupting party $A$. For every $\lambda \in \bbN$, the the view of the environment in the real execution is denoted by a random variable $\Real_{\Pi,\C,A}(\cA_\secp,x_A,x_B,\baux_\cA)$, where $x_A$ is party $A$'s input, $x_B$ is party $B$'s input, and $\baux_\cA$ is some potentially quantum side information that may be entangled with the distinguisher's side information $\baux_\cD$. The random variable consists of $(\bz,y_B)$, where $\bz$ is $\cA_\secp$'s final output after interacting with an honest $B$ algorithm $B(1^\secp,x_B)$, and $y_B$ is $B$'s output.

We require the existence of a QPT simulator $\Sim_A = \left(\Sim_A^{(1)},\Sim_A^{(2)}\right)$ that interacts with any QPT adversary $\cA = \{\cA_\secp\}_{\secp \in \bbN}$ corrupting party $A$. $\Sim_A$ has the following syntax.

\begin{itemize}
    \item $\Sim_A^{(1)}(1^\secp)$ generates $(\crs,\tau,m_1)$, sends $(\crs,m_1)$ to $\cA_\secp(x_A,\baux_\cA)$, and receives back $m_2$.
    \item $\Sim_A^{(2)}(1^\secp,x_A,\tau,m_2)$ computes either $x_A'$ or $\bot$, which it forwards to an ideal functionality $\cI_A[x_B](\cdot)$.
\end{itemize}

$\cI[x_B](\cdot)$ operates as follows. It takes an input $x_A'$ or $\bot$, and in the non-$\bot$ case it computes and outputs $y \gets C_\secp(x_A',x_B)$ (note this is not returned to the simulator), and in the $\bot$ case it outputs $\bot$. The random variable $\Ideal_{\Pi,\C,A}(\Sim_A,x_A,x_B,\baux_\cA)$ consists of the output of $\cA_\secp(x_A,\baux_\cA)$ after interacting with the simulator, along with the output of $\cI[x_B](\cdot)$. 

We require an analogous security property in the case that $\cA$ corrupts party $B$. Here, the syntax of $\Sim_B = \left(\Sim_B^{(1)},\Sim_B^{(2)}\right)$ is as follows.

\begin{itemize}
    \item $\Sim_B^{(1)}(1^\secp)$ generates $(\crs,\tau)$, sends $\crs$ to $\cA_\secp(x_B,\baux_\cA)$, and receives back $m_1$.
    \item $\Sim_B^{(2)}(1^\secp,\tau,m_1)$ takes the adversary's message $m_1$ and either extracts an input $x_B'$ or $\bot$, which it forwards to an ideal functionality $\cI[x_A](\cdot)$. In the non-$\bot$ case, $\cI[x_A]$ computes and returns $y \gets C(x_A,x_B')$ to the simulator and outputs $\mathsf{ok}$. In the $\bot$ case it outputs $\abort$, and the simulator send $\bot$ to $\cA$ and simulation ends.
    \item $\Sim_B^{(3)}(1^\secp,\tau,y)$ receives an output $y$ from the ideal functionality and uses it to form a second round message $m_2$, which it sends to $\cA_\secp$.
\end{itemize}

\begin{definition}[Post-Quantum Two-Message Two-Party Computation]
\label{def:2pc}

A protocol $\Pi$ securely computes $C$ against malicious party $P \in \{A,B\}$ if for all QPT $\cA = \{\cA_\secp\}_{\secp \in \bbN}$ corrupting party $P$, there exists a QPT $\Sim_P = \{\Sim_{P,\secp}\}_{\secp \in \bbN}$ such that for all $\{x_{A,\secp},x_{B,\secp},\baux_{\cA,\secp},\baux_{\cD,\secp}\}_{\secp \in \bbN}$ and all QPT $\cD = \{\cD_\secp\}_{\secp \in \bbN}$, there exists a negligible function $\nu(\cdot)$ such that

\begin{align*}
    &\bigg| \Pr\left[\cD_\secp\left(\baux_{\cD,\secp},\Real_{\Pi,\C,P}(\cA_\secp,x_{A,\secp},x_{B,\secp},\baux_{\cA,\secp})\right) = 1\right]\\ &- \Pr\left[\cD_\secp\left(\baux_{\cD,\secp},\Ideal_{\Pi,\C,P}(\Sim_{P,\secp},x_{A,\secp},x_{B,\secp},\baux_{\cA,\secp})\right) = 1\right] \bigg| \leq \nu(\secp).
\end{align*}

\end{definition}

A secure two-party computation protocol satisfying Definition~\ref{def:2pc} can be obtained based on any post-quantum maliciously-secure two-message OT with straight-line simulation, via~\cite{C:IshPraSah08,EC:IKOPS11}. 
The non-interactive secure two-party protocol from~\cite[Appendix B]{C:IshPraSah08} is based on Yao's garbled circuit technique~\cite{FOCS:Yao86} along with a cut-and-choose mechanism for proving that a garbled circuit is computed correctly. The cut-and-choose is non-interactive in the OT-hybrid model. This can be cast in the simpler setting of two-party computation with a CRS, where we replace the ideal calls to the OT with a post-quantum secure two-message OT with straight-line simulation (that is auxiliary-input secure). The latter can be based on the quantum hardness of the learning with errors (QLWE) problem~\cite{C:PeiVaiWat08}.

In Section \ref{sec:mdv-nizk}, we will rely on reusable post-quantum two-party computation with straight-line simulation, in order to obtain reusable malicious designated-verifier NIZKs for QMA. We point out that~\cite{C:LQRWW19} build reusable (post-quantum) two-party computation (with straight-line simulation) assuming (post-quantum)  malicious MDV-NIZKs for NP, and (post-quantum) oblivious transfer. Both can be obtained from QLWE~\cite{C:LQRWW19,C:PeiVaiWat08}.

\subsection{Round-Optimal MPC for Classical Reactive Functionalities}\label{subsec:reactivempc}

We, define a $d$-level $n$-party randomized functionality $\cF = (\cF_1,\dots,\cF_d)$ as follows. Let $r$ be random coins. Then each $\cF_j$ can be defined as $$\left(y_1^{(j)},\dots,y_n^{(j)}\right) \coloneqq \cF_j\left(\left(x_1^{(1)},\dots,x_1^{(j)}\right),\dots,\left(x_n^{(1)},\dots,x_n^{(j)}\right),w^{(j)},r\right),$$ where $\left(x_i^{(1)},\dots,x_i^{(d)}\right)$ are the set of party $i$'s private inputs, $(w^{(1)},\dots,w^{(d)})$ are some public inputs, and $\left(y_i^{(1)},\dots,y_i^{(d)}\right)$ are the set of party $i$'s outputs. We allow $x_i^{(j)}$ to be an arbitrary function of $y_i^{(1)},\dots,y_i^{(j-2)}$.

Now, we define an ideal functionality $\cI_\cF$ for computing $\cF$ in $d+1$ rounds. Let $\cH \subset [n]$ be a subset of honest parties and $\cM \coloneqq [n] \setminus \cH$ be the corresponding subset of malicious parties. $\cI_\cF$ is initialized with honest party inputs $\left\{x_i^{(1)}\right\}_{i \in \cH}$. 

\begin{itemize}

    \item In round 1, accept and store private inputs $\left\{x_i^{(1)}\right\}_{i \in \cM}$.
    
    \item In round $j$, for $j \in \{2,\dots,d\}$, do the following. Accept public input $w^{(j-1)}$ and compute 
    
    $$\left(y_1^{(j-1)},\dots,y_n^{(j-1)}\right) \coloneqq \cF_j\left(\left(x_1^{(1)},\dots,x_1^{(j-1)}\right),\dots,\left(x_n^{(1)},\dots,x_n^{(j-1)}\right),w^{(j-1)},r\right).$$ Output $\left\{y_i^{(j-1)}\right\}_{i \in \cM}$. Accept either $\left(\mathsf{ok},\left\{x_i^{(j)}\right\}_{i \in \cM}\right)$ or $\abort$ as input. If $\mathsf{ok}$, set level $j-1$ honest party outputs to $\left\{y_i^{(j-1)}\right\}_{i \in \cH}$ and compute the next set of honest party inputs $\left\{x_i^{(j)}\right\}_{i \in \cH}$. If $\abort$, set honest party outputs to $\bot$ for each level $k \geq j-1$.
    
    \item In round $d+1$, accept public input $w^{(d)}$ and compute 
    
    $$\left(y_1^{(d)},\dots,y_n^{(d)}\right) \coloneqq \cF_j\left(\left(x_1^{(1)},\dots,x_1^{(d)}\right),\dots,\left(x_n^{(1)},\dots,x_n^{(d)}\right),w^{(d)},r\right).$$ Output $\left\{y_i^{(d)}\right\}_{i \in \cM}$. Accept either $\mathsf{ok}$ or $\abort$ as input. If $\mathsf{ok}$, set level $d$ honest party outputs to $\left\{y_i^{(d)}\right\}_{i \in \cH}$. If $\abort$, set level $d$ honest party outputs to $\bot$.
\end{itemize}

We observe that the protocol of \cite{EC:GarSri18a} can be used to implement this ideal functionality for any $d$-level $n$-party functionality. If \cite{EC:GarSri18a} is instantiated with post-quantum maliciously-secure two-message oblivious transfer in the CRS model with straight-line black-box simulation, the resulting $d+1$ round MPC protocol will be quantum-secure and admit a straight-line black-box simulator. 

To see this, recall that \cite{EC:GarSri18a} is a compiler that operates on any underlying MPC protocol. One can fix the underlying MPC protocol to support reactive functionalities. Such an MPC can be obtained from any non-reactive MPC by having the non-reactive MPC output a fresh set of secret shares of each party's input between each execution.
We will rely on a maliciously-secure reactive MPC from OT~\cite{C:CreVanTap95,C:IshPraSah08}. Then, the first round of the \cite{EC:GarSri18a}-compiled protocol will consist of first round OT messages committing each party to the randomness they will use throughout the execution of the underlying reactive MPC protocol, as in the original \cite{EC:GarSri18a} protocol. Each subsequent round will consist of the \cite{EC:GarSri18a} second-round messages for computing the appropriate portion of the reactive functionality, along with encryptions of each party's input for the next round. During the security proof, in each round the simulator can simply invoke the simulator for the appropriate portion of the underlying reactive MPC.

Thus, we will construct each of our multi-party quantum computation protocols with respect to a $d$-level $n$-party classical ideal functionality, that can be implemented by \cite{EC:GarSri18a}.

\subsection{Quantum Multi-Key Fully-Homomorphic Encryption}
\label{sec: qmfhe}

We use a quantum multi-key fully-homomorphic encryption scheme that supports classical encryption of classical ciphertexts. We do not require compactness or the classicality-preserving property as required by~\cite{Agarwal2020PostQuantumMC}, but we do require a form a circuit-privacy, presented below as ciphertext re-randomization.

\begin{definition}[Quantum Multi-Key Fully-Homomorphic Encryption~\cite{FOCS:Mahadev18b,Agarwal2020PostQuantumMC}]
\label{def:qmfhe}
A quantum multi-key fully-homomorphic encryption scheme is given by seven algorithms ($\qmfhe.\gen,\qmfhe.\KeyGen$, $\qmfhe.\CEnc$, $\qmfhe.\Enc$, $\qmfhe.\eval$, $\qmfhe.\Rerand$, $\qmfhe.\Dec$) with the following syntax.
\begin{itemize}
    \item $\crs \gets \qmfhe.\gen(1^\secp)$: A PPT algorithm that outputs a classical common reference string.
    \item $(\pk,\sk) \leftarrow \qmfhe.\KeyGen(1^\secp,\crs)$ : A PPT algorithm that given a security parameter, samples a classical public key and a classical secret key.
    \item $\ct \leftarrow \qmfhe.\CEnc(\pk,x)$ : A PPT algorithm that takes as input a bit $x$ and outputs a classical ciphertext.
    \item $\bct \leftarrow \qmfhe.\Enc(\pk,\bx)$ : A QPT algorithm that takes as input a qubit $\bx$ and outputs a quantum ciphertext.
    \item $\widehat{\bct} \gets \qmfhe.\eval((\pk_1,\dots,\pk_n),Q,(\bct_1,\dots,\bct_n))$: A QPT algorithm that takes as input a set of $n$ public keys, a quantum circuit $Q$, and a set of $n$ (classical or quantum) ciphertexts, and outputs an evaluated ciphertext $\widehat{\bct}$.
    \item $\widetilde{\bct} \gets \qmfhe.\Rerand((\pk_1,\dots,\pk_n),\bct)$: A QPT algorithm that re-randomizes a ciphertext $\bct$ encrypted under a set of $n$ public keys
    \item $\bx \gets \qmfhe.\Dec((\sk_1,\dots,\sk_n),\bct)$: A QPT algorithm that takes as input a set of $n$ secret keys and a quantum ciphertext $\bct$ and outputs a qubit.
\end{itemize}

The scheme satisfies the following.
\begin{enumerate}
\item
{\bf Quantum Semantic Security:} The encryption algorithm maintains quantum semantic security.
\item {\bf Quantum Homomorphism: } For any polynomial-size quantum circuit $Q$, input state $\bx_1,\dots,\bx_n$, crs $\crs \in \qmfhe.\gen(1^\secp)$, and key pairs $(\pk_1,\sk_1),\dots,(\pk_n,\sk_n)  \in \qmfhe.\KeyGen(1^\secp,\crs)$, it holds that $\by_0 \approx_s \by_1,$ where $\by_{0},\by_{1}$ are QRVs defined as follows:
    \begin{itemize}
        \item $\by_{0}$: For each $i \in [n]$, encrypt each classical bit of $\bx_i$ with $\qmfhe.\CEnc(\pk_i,\cdot)$ and the rest with $\qmfhe.\Enc(\pk_i,\cdot)$. Execute $\qmfhe.\eval((\pk_1,\dots,\pk_n),Q,\cdot)$ on the $n$ encryptions to obtain $\widehat{\bct}$. Then output $\qmfhe.\Dec((\sk_1,\dots,\sk_n),\qmfhe.\Rerand(\widehat{\bct}))$.
        \item $\by_{1}$: Output $Q(\bx_1,\dots,\bx_n)$.
    \end{itemize}


\item {\bf Ciphertext Re-randomization:} For any $\crs \in \qmfhe.\gen(1^\secp)$, key pairs $(\pk_1,\sk_1),\dots,(\pk_n,\sk_n) \in \qmfhe.\KeyGen(1^\secp,\crs)$, and ciphertexts $\bct_1,\bct_2$ such that $$\qmfhe.\Dec((\sk_1,\dots,\sk_n),\bct_1) = \qmfhe.\Dec((\sk_1,\dots,\sk_n),\bct_2),$$ it holds that $$\qmfhe.\Rerand((\pk_1,\dots,\pk_n),\bct_1) \approx_s \qmfhe.\Rerand((\pk_1,\dots,\pk_n),\bct_2).$$ 
\end{enumerate}
\end{definition}

We now sketch how to add the ciphertext re-randomization property to the QMFHE scheme constructed in~\cite{Agarwal2020PostQuantumMC} via ``noise-flooding''. An evaluated ciphertext encrypting the quantum state $\brho$ will have the form $$\mfhe.\Enc((\pk_1,\dots,\pk_n),(x,z)),X^xZ^z\brho,$$where $\mfhe$ is a \emph{classical} multi-key fully-homomorphic encryption scheme. Thus, it suffices to show how to add ciphertext re-randomization to the classical multi-key FHE scheme of~\cite{EC:MukWic16}. 

It is well-known that standard single-key FHE schemes from the literature (\cite{C:GenSahWat13}) are statistically re-randomizable. Now to construct MFHE with ciphertext re-randomization, we can append to each $\mfhe$ public key a freshly sampled GSW encryption of its corresponding secret key. To re-randomize a $\mfhe$ ciphertext encrypted under public keys $\pk_1,\dots,\pk_n$, one can compute the partial decryption under each corresponding GSW ciphertext, resulting in $n$ ciphertexts whose plaintexts sum to $\mu (q/2) + e$, where $\mu$ was the bit encrypted under $\mfhe$. Then, add a random additive secret sharing of $e'$ for a large enough $e'$ under the encryptions and re-randomize each. The result is an random additive sharing of $\mu (q/2) + e + e'$ under re-randomized GSW ciphertexts, where $e + e' \approx_s e''$ for some distribution $e''$ independent of the computation.

\subsection{Non-Interactive Equivocal Commitment}
\label{sec: equivocal commitments}

\begin{definition}[Equivocal Commitment]
A quantum-secure statistically-binding non-interactive equivocal commitment is given by three algorithms $(\Com.\Gen,\Com.\Enc,\Com.\Ver)$ with the following syntax.
\begin{itemize}
    \item $\crs \gets \Com.\Gen(1^\secp)$.
    \item $\cmt \coloneqq \Com.\Enc(1^\secp,\crs,m;r)$.
    \item $b \gets \Com.\Ver(1^\secp,\crs,\cmt,m,r)$.
\end{itemize}
It satisfies the following notion of correctness. For any $m \in \{0,1\}^*$, $$\Pr\left[b = 1 : \begin{array}{r}\crs \gets \Com.\Gen(1^\secp), r \gets \{0,1\}^\secp \\ \cmt \coloneqq \Com.\Enc(1^\secp,\crs,m;r),b \gets \Com.\Ver(1^\secp,\crs,\cmt,m,r)\end{array}\right] = 1-\negl(\secp).$$
It satisfies the statistical binding property.
$$\Pr_{\crs \gets \Com.\Gen(1^\secp)}\left[\begin{array}{l} \exists (\cmt,m_0,m_1,r_0,r_1), m_0 \neq m_1 \text{ s.t. } \\ \Com.\Ver(1^\secp,\crs,\cmt,m_0,r_0) = 1 = \Com.\Ver(1^\secp,\crs,\cmt,m_1,r_1) \end{array}\right] = \negl(\secp).$$

Finally, it satisfies the following notion of security (hiding). There exists algorithms $\Com.\Sim.\Gen,\Com.\Sim.\Open$ such that for any $m \in \{0,1\}^*,$

$$\Pr\left[b = 1 : \begin{array}{r}(\crs,\cmt,\tau) \gets \Com.\Sim.\Gen(1^\secp) \\ r_m \gets \Com.\Sim.\Open(1^\secp,\tau,m),b \gets \Com.\Ver(1^\secp,\crs,\cmt,m,r_m)\end{array}\right] = 1-\negl(\secp),$$ and

$$\left\{(\crs,\cmt) : \begin{array}{r}\crs \gets \Com.\Gen(1^\secp)\\ \cmt \gets \Com.\Enc(1^\secp,\crs,m)\end{array}\right\}_{\secp \in \bbN} \approx_c \left\{(\crs,\cmt) : (\crs,\cmt,\tau) \gets \Com.\Sim.\Gen(1^\secp)\right\}_{\secp \in \bbN}.$$
\end{definition}

A commitment scheme satisfying the above definition can be based on any quantum-secure one-way function~\cite{JC:Naor91}.

\subsection{Garbled Circuits}

\begin{definition}[Garbled Circuit]\label{def:GC}
A garbling scheme for circuits is a tuple of PPT algorithms $(\Garble, \GEval)$. $\Garble$ is the circuit garbling procedure and $\GEval$ is the corresponding evaluation procedure. More formally:
\begin{itemize}
  \item $(\widetilde{C},\{\lab_{i,b}\}_{i \in [n], b\in \{0,1\}})  \gets \Garble\left(1^\secp, C\right)$: $\Garble$ takes as input a security parameter $1^\secp$, a classical circuit $C$, and outputs a \emph{garbled circuit} $\widetilde{C}$ along with labels $\{\lab_{i,b}\}_{i \in [n], b \in \{0,1\}}$, where $n$ is the length of the input to $C$. 
  \item $y \gets \GEval\left(\widetilde{C}, \{\lab_{i,x_i}\}_{i \in [n]}\right)$: Given a garbled circuit $\widetilde{C}$ and a sequence of input labels $\{ \lab_{i,x_i}\}_{i \in [n]}$, $\GEval$ outputs a string $y$.
\end{itemize}

\paragraph{Correctness.} For correctness, we require that for any classical circuit $C$ and input $x \in \{0,1\}^n$ we have that:
\[\Pr\left[C(x) = \GEval\left(\widetilde{C}, \{\lab_{i,x_i}\}_{i \in [n]}\right)\right] = 1, \]
where $(\widetilde{C}, \{\lab_{i,b}\}_{i \in [n], b \in \{0,1\}}) \gets \Garble\left(1^\secp, C\right)$.

\paragraph{Security.} For security, we require that there exists a PPT simulator $\GSim$ such that for any classical circuit $C$ and input $x \in \{0,1\}^n$, we have that
\[\left(\widetilde{C}, \{\lab_{i,x_i}\}_{i \in [n]}\right) \approx_c  \GSim\left(1^\secp,1^{n},1^{|C|},C(x)\right),\]
where $(\widetilde{C}, \{\lab_{i,b}\}_{i\in [n], b\in \{0,1\}}) \gets \Garble\left(1^{\secp},C\right)$.
\end{definition}


\ifsubmission
\else
\section{A Garbling Scheme for Clifford + Measurement Circuits}

In this section, we formalize and prove the security of a method sketched in \cite[\S2.5]{ARXIV:BrakerskiYuen20} for garbling Clifford plus measurement circuits. Note that this is not the main garbling scheme analyzed in \cite{ARXIV:BrakerskiYuen20}, but it is a scheme that is sketched there informally. We begin by giving the formal definition of a Clifford + measurement circuit, as well as our definition of a garbling scheme for such circuits.

\begin{definition}[Clifford + Measurement ($\CM$) Circuit]
\label{def: Cliff plus Meas circ}A Clifford + Measurement ($\CM$) circuit with parameters $\{n_i,k_i\}_{i \in [d]}$ operates on $n_0 \coloneqq n_1 + k_1$ input qubits and applies $d$ alternating layers of Clifford unitary and computational basis measurements, during which a total of $k \coloneqq k_1 + \dots + k_d$ of the input qubits are measured. It is specified by $(F_0,f_1,\dots,f_d)$, where $F_0$ is a Clifford unitary, and each $f_i$ is a classical circuit which takes as input the result of computational basis measurements on the $i$th layer, and outputs a Clifford unitary $F_i$. In layer $i \in [d]$, $k_i$ qubits are measured and $n_i$ qubits are left over. The circuit is evaluated by first applying $F_0$ to the $n_0$ input qubits. Then the following steps are performed for $i = 1,\dots,d$:
\begin{itemize}
\item Measure the remaining $k_i$ qubits in the computational basis, resulting in outcomes $m_i \in \{0,1\}^{k_i}$.
\item Evaluate $f_i(m_i)$ to obtain a classical description of a Clifford $F_i \in \mathscr{C}_{n_i}$.
\item Apply $F_i$ to the first $n_i$ registers.

\end{itemize}

The output of the circuit is the result of applying $F_d$ to the final $n_d$ registers. 

\end{definition}


It is well-known (\cite{BravyiKitaev05}) that any polynomial-size quantum circuit can be written as a $\CM$ circuit with polynomial-size parameters $\{n_i,k_i\}_{i \in [d]}$. The transformation maintains correctness as along as sufficient $\Tstate$ states are appended to the input during evaluation.

\begin{definition}[Garbling Scheme for $\CM$ Circuits]\label{defn:QGC}
A Garbling Scheme for $\CM$ Circuits consists of three procedures $(\QGarble,\QGEval,\QGSim)$ with the following syntax.
\begin{itemize}
    \item $(E_0,\widetilde{Q}) \gets \QGarble(1^\secp,Q)$: A \emph{classical} PPT procedure that takes as input the security parameter and a $\CM$ circuit and outputs a Clifford ``input garbling'' matrix $E_0$ and a quantum garbled circuit $\widetilde{Q}$.
    \item $\bx_\out \gets \QGEval(\widetilde{\bx}_\inp,\widetilde{Q})$: A QPT procedure that takes as input a garbled input $\widetilde{\bx}_\inp$ and a garbled $\CM$ circuit $\widetilde{Q}$, and outputs a quantum state $\bx_\out$.\
    \item $(\widetilde{\bx}_\inp,\widetilde{Q}) \gets \QGSim(1^\secp,\{n_i,k_i\}_{i \in [d]},\bx_\out)$: A QPT procedure that takes as input the security parameter, parameters for a $\CM$ circuit, and an output state, and outputs a simulated garbled input and garbled circuit.
\end{itemize}
\paragraph{Correctness.} For any $\CM$ circuit $Q$ with parameters $\{n_i,k_i\}_{i \in [d]}$, and $n_0$-qubit input state $\bx_\inp$ along with (potentially entangled) auxiliary information $\bz$, 
\begin{align*}
    \left\{\left(\QGEval\left(E_0\left(\bx_\inp,\Zstate^{k\secp}\right), \widetilde{Q}\right),\bz\right) : \left(E_0,\widetilde{Q}\right) \gets \QGarble\left(1^\secp,Q\right)\right\} \approx_s \left(Q\left(\bx_\inp\right),\bz\right). 
\end{align*}
\paragraph{Security.} For any $\CM$ circuit $Q$ with parameters $\{n_i,k_i\}_{i \in [d]}$, and $n_0$-qubit input state $\bx_\inp$ along with (potentially entangled) auxiliary information $\bz$,
\begin{align*}
    \left\{\left(E_0\left(\bx_\inp,\Zstate^{k\secp}\right), \widetilde{Q},\bz\right) : \left(E_0,\widetilde{Q}\right) \gets \QGarble\left(1^\secp,Q\right)\right\} \approx_c \left(\QGSim\left(1^\secp,\{n_i,k_i\}_{i \in [d]},Q(\bx_\inp)\right),\bz\right). 
\end{align*}
\end{definition}

   


Before formally describing the concrete garbling scheme for $\mathsf{C} + \mathsf{M}$ circuits, we give a formal definition of a process $\labenc$ for sampling a ``label encoding'' unitary given a set of classical garbled circuit labels. For $\lambda$-bit strings $s_0$, $s_1$ and a bit $b$, let $C_b^{s_0,s_1}$ be the Clifford acting on $\lambda+1$ qubits, operating as follows:
\begin{itemize}
    \item Apply $Z_b$ to the first qubit. Looking ahead, $b$ will be chosen at random so that $Z_b$ will have the effect of a $Z$-twirl, which is equivalent to a measurement in the computational basis.
    \item Map $\ket{0, 0^{\lambda}}$ to $\ket{0, s_0}$, and $\ket{1,0^{\lambda}}$ to $\ket{1, s_1}$.
\end{itemize}

\paragraph{$\labenc(\labset)$:} Takes as input $\labset = \{\lab_{i,0},\lab_{i,1}\}_{i \in [n]}$, where the $\lab_{i,b}$ are $\lambda$-bit strings. Draws $n$ random bits $b_1,\dots,b_n \gets \{0,1\}$, and outputs $\bigotimes_{i \in [n]}C^{\lab_{0,i}, \lab_{1,i}}_{b_i}$,



\begin{lemma}
\label{lem: labels}
Let $m>n$. For any $m$-qubit state $\ket{\phi}$ and set of labels $\labset = \{\lab_{i,0}, \lab_{i,1}\}_{i \in [n]}$, where the $\lab_{i,b}$ are $\lambda$-bit strings, $$L\ket{\phi'}\bra{\phi'}L^{\dagger} = \mathbb{E}_{\mathsf{inp}} \ket{\phi'_{\mathsf{inp}}} \bra{\phi'_{\mathsf{inp}}} \otimes \ket{\lab_{1,\mathsf{inp}_1}, \dots, \lab_{n,\mathsf{inp}_n}}\bra{\lab_{1,\mathsf{inp}_1}, \dots, \lab_{n,\mathsf{inp}_n}} \,,$$ 
where $L \gets \labenc(\labset)$, $\ket{\phi'}$ is the $(m+n\secp)$-qubit state consisting of $\ket{\phi}$ and $n\secp$ ancillary 0 states, $\ket{\phi'_{\mathsf{inp}}}$ is the post-measurement state on the first $m-n$ qubits, conditioned on measuring the last $n$ qubits and obtaining outcome $\mathsf{inp}$, and the expectation is taken over $\mathsf{inp} \in \{0,1\}^n$ distributed according to the result of measuring the last $n$ qubits of $\ket{\phi}$ in the computational basis. 
\end{lemma}
\begin{proof}
We can write $\ket{\phi'}$ as follows:
$$\ket{\phi'} = \sum_{x \in \{0,1\}^n} \alpha_x \ket{\phi_x} \ket{x} \,.$$
for some $\alpha_x \in \mathbb{C}$. Then,
$$ \mathbb{E}_{L \gets \labenc(\labset)} L \ket{\phi'} \bra{\phi'} L^{\dagger} = \mathbb{E}_{z \gets \{0,1\}^n} \bbI \otimes Z^z\otimes \bbI \left( \sum_{x \in \{0,1\}^n} \alpha_x \ket{\phi_x} \ket{x} \ket{\lab_{x}}\right) \left(\sum_{x' \in \{0,1\}^n} \alpha_{x'} \bra{\phi_{x'}} \bra{x} \bra{\lab_{x'}}\right) \bbI \otimes Z^z \otimes \bbI \,,$$
where $\ket{\lab_x} = \ket{\lab_{1,x_1}, \ldots, \lab_{1,x_n}}$.


By a well-known property of the Pauli-Z twirl, the above is equal to:
$$ \sum_{x \in \{0,1\}^n} |\alpha_x|^2 \ket{\phi_x}\bra{\phi_x} \otimes \ket{x}\bra{x} \otimes \ket{\lab_{x}}\bra{\lab_{x}} \,,$$
which implies the desired statement.
\end{proof}

Now, we are ready to describe formally the garbling scheme $(\QGarble,\QGEval,\QGSim)$ for $\mathsf{C}+\mathsf{M}$ circuits sketched by \cite{ARXIV:BrakerskiYuen20}. Let $(\Garble,\GEval,\GSim)$ be a classical garbling scheme.

\paragraph{$\QGarble(1^\secp,Q)$:}
Takes as input a $\CM$ circuit $Q$ with parameters $\{n_i,k_i\}_{i \in [d]}$.
\begin{enumerate}
    \item For $i \in [0,\dots,d]$, define $h_i \coloneqq k - \sum_{j=1}^i k_j$, so that $h_0 = k, h_1 = k-k_1, h_2 = k-k_1-k_2$, and so on.
    \item For each $i \in [0,\dots,d]$, sample $E_i \gets \mathscr{C}_{n_i + h_i\secp}$.
    \item For each $i \in [d]$, let $f_i$ be the classical circuit (derived from the description of $Q$) that takes as input $k_i$ bits interpreted as the outcomes of computational basis measurements in layer $i$ and outputs a Clifford circuit $F_i \in \mathscr{C}_{n_{i}}$ to be applied on the remaining $n_{i}$ qubits.
    \item Let $g_d$ be a classical circuit outputting descriptions of Clifford circuits, defined so that $g_d(x) = f_d(x)E_d^\dagger$. Compute $(\labset_d, \widetilde{g}_d) \gets \Garble(1^\secp,g_d)$.  
    \item For each $i$ from $d-1$ to 1, sample $L_{i+1} \gets \labenc(\labset_{i+1})$ and compute $(\labset_i,\widetilde{g}_i) \gets \Garble(1^\secp,g_i)$, where $g_i$ is a classical circuit that outputs descriptions of Clifford circuits,  $$g_i(x) = (E_{i+1} \otimes L_{i+1}) (f_i(x) \otimes \bbI^{h_i\secp}) E_i^\dagger\,.$$ 
    \item Let $F_0$ be the initial Clifford to be applied to the input qubits, sample $L_1 \gets \labenc(\labset_1)$ and output $$E_0, \widetilde{Q} = \left((E_1 \otimes L_1)(F_0\otimes\bbI^{h_0\secp})E_0^\dagger, \widetilde{g}_1,\dots,\widetilde{g}_d\right).$$
\end{enumerate}

\paragraph{$\QGEval(\widetilde{\bx}_\inp,\widetilde{Q})$} Takes as input a garbled input $\widetilde{\bx}_\inp$ and a garbled $\CM$ circuit $\widetilde{Q}$.

\begin{enumerate}
    \item Write $\widetilde{Q} = (D_0,\widetilde{g}_1,\dots,\widetilde{g}_d)$ and set $\bx_0 \coloneqq \widetilde{\bx}_\inp$. For $i$ from 1 to $d$, compute $D_{i-1}(\bx_{i-1})$, measure the last $k_i\secp$ qubits to obtain a set of labels $\labsimset_i$, compute $D_i \gets \GEval(\tilde{g}_i,\labsimset_i)$, and set $\bx_i$ to be the remaining $n_{i} + h_i \lambda$ qubits of the state. 
    \item Output $D_d(\bx_d)$.
\end{enumerate}

\paragraph{$\QGSim(1^\secp,\{n_i,k_i\}_{i \in [d]},\bx_\out)$:} Takes as input parameters for a $\CM$ circuit and a state $\bx_\out$.

\begin{enumerate}
    \item For each $i \in [0,\dots,d]$, sample $D_i \gets \mathscr{C}_{n_i + h_i\secp}$, where recall that $h_i \coloneqq k - \sum_{j=1}^i k_j$.
    \item Let $\bx_d = D_d^\dagger(\bx_\out)$. For $i$ from $d$ to $1$, compute $\labsimset_i,\tilde{g}_i \gets \GSim(1^\secp,1^{k_i},1^{|g_i|},D_i^\dagger)$ and set $\bx_{i-1} \coloneqq D_{i-1}^\dagger(\bx_i,\ket{\widetilde{\lab}_i}\bra{\widetilde{\lab}_i})$.  
    \item Output $\bx_0,D_0,\tilde{g}_1,\dots,\tilde{g}_d$.
\end{enumerate}

\begin{theorem}
\label{thm: garbling}
The triple $(\QGarble, \QGEval, \QGSim)$ defined above satisfies \cref{defn:QGC}.
\end{theorem}
To prove Theorem \ref{thm: garbling}, we need the following additional lemma.

\begin{lemma}
\label{lem: clifford}
For any state $\bx$ on $n$ qubits and any Clifford $R$ on $n$ qubits. The following two states are identical:
\begin{itemize}
    \item $\mathbb{E}_{C \leftarrow \mathscr{C}_{n}} \left( C (\bx), RC^{\dagger} \right)$
    \item $\mathbb{E}_{D \leftarrow \mathscr{C}_{n}} \left( D^{\dagger} R (\bx), D \right)$
\end{itemize}
\end{lemma}

\begin{proof}
The proof is straightforward. Notice first that, because $\mathscr{C}_n$ is a group, we have that 
$$ \mathbb{E}_{C \leftarrow \mathscr{C}_{n}} \left( C (\bx), RC^{\dagger} \right) = \mathbb{E}_{C \leftarrow \mathscr{C}_{n} \cdot R} \left( C (\bx), RC^{\dagger} \right) \,,$$
where we denote by $\mathscr{C}_n \cdot R$ the group $\{C R : C \in \mathscr{C}_n\}$.
We can equivalently rewrite the RHS as
$$\mathbb{E}_{D \leftarrow \mathscr{C}_{n}} \left( DR (\bx), R(DR)^{\dagger} \right) \,,$$
which, upon simplification, is equal to
$$\mathbb{E}_{D \leftarrow \mathscr{C}_{n}} \left( DR (\bx), D^{\dagger} \right) \,.$$
Finally, using again that $\mathscr{C}_n$ is a group, the latter equals
$$\mathbb{E}_{D \leftarrow \mathscr{C}_{n}} \left( D^{\dagger}R (\bx), D \right) \,,$$
as desired.
\end{proof}

\begin{proof}[Proof of Theorem \ref{thm: garbling}]
We will show this by induction on the number of measurement layers $d$ in the circuit $Q$ (we use the same notation as above for the components of $Q$).

\paragraph{Base step ($d = 0$):} When $d=0$, the LHS of the equation defining security in \cref{defn:QGC} is:
\begin{equation}
\left\{E_0\left(\bx_\inp\right), \widetilde{Q} : \,\left(E_0, \widetilde{Q} = F_0 E_0^{\dagger} \right) \gets \QGarble(Q)\right\}    
\end{equation}
By definition of $\mathsf{QGarble}$, the latter is equivalent to:
\begin{equation}
\bigg\{E_0\left(\bx_\inp \right), \widetilde{Q}:
\,\,\,\,\,\,\widetilde{Q} = F_0 E_0^{\dagger}, \,\, E_0 \gets \mathscr{C}_{n_0} \bigg\}
\end{equation}

By Lemma \ref{lem: clifford}, the latter is identical to:
\begin{equation}
\label{eq: 5}
\bigg\{D_0^{\dagger}F_0 \left(\bx_\inp \right), D_0:
\,\,D_0 \gets \mathscr{C}_{n_0} \bigg\} \,.
\end{equation}

\paragraph{Inductive step ($d \Rightarrow d+1$):} Suppose that for some $d$ the following two distributions are identical for any $\mathsf{C} + \mathsf{M}$ circuit $Q$ with $d$ measurements (where we use the same notation as in definition \ref{def: Cliff plus Meas circ} for the components of $Q$), and any input state $\bx_\inp$.
\begin{itemize}
    \item $$
\left\{E_0\left(\bx_\inp,\Zstate^{k\secp}\right), \widetilde{Q} : \,\left(E_0, \widetilde{Q} = \left((E_1 \otimes L_1) (F_0 \otimes \bbI^{h_0\lambda}) E_0^{\dagger}, \widetilde{g}_1, \ldots, \widetilde{g}_d\right)\right) \gets \QGarble(Q)\right\}    
$$
    \item \begin{align*}
\bigg\{D_0^{\dagger} \bigg(D_1^{\dagger}\otimes \bbI\bigg) \bigg(\cdots \Big(D_{d-1}^{\dagger}\otimes \bbI\Big) \Big((D_d^{\dagger}\otimes \bbI) \left(Q\left(\bx_\inp \right) \otimes \lab_{d} \right) &\otimes \lab_{d-1} \Big) \otimes \cdots \otimes \lab_1 \bigg), D_0, \widetilde{g}_1, \ldots, \widetilde{g}_d: \\
\,\,\,\,\,\,& D_i \gets \mathscr{C}_{n_i + h_i \lambda}, \,i \in \{0,\ldots,d\},\\
&(\lab_{i}, \widetilde{g}_i) \gets \mathsf{GSim}( D_i), \text{ for $i \in [d]$} \bigg\}
\end{align*}
\end{itemize}

Let $Q$ be a $\mathsf{C} + \mathsf{M}$ circuit with $d+1$ measurements, and let $\bx_\inp$ be an input to the circuit. Consider the distribution of input encoding + garbled circuit:
$$
\left\{E_0\left(\bx_\inp,\Zstate^{k\secp}\right), \widetilde{Q} : \,\left(E_0, \widetilde{Q} = \left((E_1 \otimes L_1) (F_0 \otimes \bbI^{h_0\lambda}) E_0^{\dagger}, \widetilde{g}_1, \ldots, \widetilde{g}_{d+1}\right)\right) \gets \QGarble(Q)\right\}    
$$
Let $Q_d$ be the circuit that runs $Q$ up to (and including) the adaptive Clifford controlled on the $d$-th measurement outcome. For ease of notation, we simply write $\lab_{i,x}$ to denote the encoding label for measurement outcome $x$ at the $i$-th layer. More precisely, $\lab_{i,x} = (\lab_{i,x_1},\ldots,\lab_{i,x_n})$ for an appropriate $n$. Since $\widetilde{g}_{d+1}$ is independent of the random Clifford $E_d$, we can apply the inductive hypothesis to the $d$-measurement circuit $\left(E_{d+1} \otimes L_{d+1} \right) Q_d$ on input $\bx_\inp$). We deduce that the above distribution is computationally indistinguishable from: 
\begin{align*}
\bigg\{D_0^{\dagger} \bigg(D_1^{\dagger}\otimes \bbI\bigg) \bigg(\cdots \Big(D_{d}^{\dagger}\otimes \bbI\Big) \Big( &\left(E_{d+1} \otimes L_{d+1} \right) \left(Q_d\left(\bx_\inp\right)\right) \otimes \lab_{d} \Big) \otimes \cdots \otimes \lab_1 \bigg), D_0, \widetilde{g}_1, \ldots, \widetilde{g}_d, \widetilde{g}_{d+1}: \\
\,\,\,\,\,\,& D_i \gets \mathscr{C}_{n_i + h_i \lambda}, \,i\in\{0,\ldots,d\},\, E_{d+1} \gets \mathscr{C}_{n_{d+1}} \\
&(\lab_{i}, \widetilde{g}_i) \gets \mathsf{GSim}( D_i), \text{ for $i \in [d]$}, \\
&\left(\lab_{d+1} = \{\lab_{d+1,x}\}_{x \in \{0,1\}^{\lambda}}, \widetilde{g}_{d+1}\right) \gets \mathsf{Garble}\left(g_{d+1}: g_{d+1}(x) = f_{d+1}(x) E_{d+1}^{\dagger}\right), \\ 
& L_{d+1} \gets \mathsf{LabEnc}(\lab_{d+1}) \bigg\}
\end{align*}
Let $Q_d^{x}\left(\bx_\inp \right)$ be the post-measurement state upon executing circuit $Q$ up to the $d$-th measurement, conditioned on the $d$-th measurement outcome being $x$. By Lemma \ref{lem: labels}, the above distribution is identical to:
\begin{align*}
\bigg\{D_0^{\dagger} \bigg(D_1^{\dagger}\otimes \bbI\bigg) \bigg(\cdots \Big(D_{d}^{\dagger}\otimes \bbI\Big) \Big( &\mathbb{E}_{x \gets \mathsf{Meas}(Q_d\left(\bx_\inp\right))} \left[(E_{d+1} \otimes \bbI) \left(Q_d^{x} \left(\bx_\inp \right) \otimes \lab_{d+1,x}\right)\right] \otimes \lab_{d} \Big) \otimes \cdots \otimes \lab_1 \bigg), \\
&D_0, \widetilde{g}_1, .., \widetilde{g}_{d+1}: \\
\,\,\,\,\,\,&  D_i \gets \mathscr{C}_{n_i + h_i \lambda}, \,i\in\{0,\ldots,d\},\, E_{d+1} \gets \mathscr{C}_{n_{d+1}}, \\
&(\lab_{i}, \widetilde{g}_i) \gets \mathsf{GSim}( D_i), \text{ for $i \in [d]$}, \\
&\left(\lab_{d+1} = \{\lab_{d+1,x}\}_{x \in \{0,1\}^{\lambda}}, \widetilde{g}_{d+1}\right) \gets \mathsf{Garble}\left(g_{d+1}: g_{d+1}(x) = f_{d+1}(x) E_{d+1}^{\dagger}\right) \bigg\}
\end{align*}
We apply the simulation property of the classical garbling scheme (for each $x$), and deduce that the latter is computationally indistinguishable from:
\begin{align*}
\bigg\{\mathbb{E}_{x \gets \mathsf{Meas}(Q_d\left(\bx_\inp\right))} \bigg[\bigg\{D_0^{\dagger} \bigg(D_1^{\dagger}\otimes \bbI\bigg) \bigg(\cdot \cdot\Big(D_{d}^{\dagger}\otimes \bbI\Big) \Big( &(E_{d+1} \otimes \bbI) \left(Q_d^{x} \left(\bx_\inp\right) \otimes \lab_{d+1,x}\right) \otimes \lab_{d} \Big) \otimes \cdot \cdot \otimes \lab_1 \bigg), \\
&D_0, \widetilde{g}_1,\ldots, \widetilde{g}_{d},\widetilde{g}_{d+1,x} \bigg]: \\
\,\,\,\,\,\,&  D_i \gets \mathscr{C}_{n_i + h_i \lambda}, \,i\in\{0,\ldots,d\},\, E_{d+1} \gets \mathscr{C}_{n_{d+1}}, \\
&(\lab_{i}, \widetilde{g}_i) \gets \mathsf{GSim}( D_i), \text{ for $i \in [d]$}, \\
&\left(\lab_{d+1,x}, \widetilde{g}_{d+1,x}\right) \gets \mathsf{GSim}\left(f_{d+1}(x) E_{d+1}^{\dagger}\right), \text{ for $x \in \{0,1\}^{\lambda}$} \bigg\}
\end{align*}
We apply Lemma \ref{lem: clifford} (for each $x$) to deduce that the latter is identical to:
\begin{align*}
\bigg\{\mathbb{E}_{x \gets \mathsf{Meas}(Q_d\left(\bx_\inp\right))} \bigg[D_0^{\dagger} \bigg(D_1^{\dagger}\otimes \bbI\bigg) \bigg(\cdot \cdot\Big(D_{d}^{\dagger}\otimes \bbI\Big) \Big( &(D_{d+1,x}^{\dagger} \otimes \bbI) \left(f_{d+1}(x)Q_d^{x} \left(\bx_\inp\right) \otimes \lab_{d+1,x}\right) \otimes \lab_{d} \Big) \otimes \cdot \cdot \otimes \lab_1 \bigg), \\
&D_0, \widetilde{g}_1,\ldots, \widetilde{g}_{d},\widetilde{g}_{d+1,x} \bigg]: \\
\,\,\,\,\,\,&  D_i \gets \mathscr{C}_{n_i + h_i \lambda}, \,i\in\{0,\ldots,d\}, \\
& D_{d+1,x} \gets \mathscr{C}_{n_{d+1}}, \,x \in \{0,1\}^{\lambda}, \\
&(\lab_{i}, \widetilde{g}_i) \gets \mathsf{GSim}( D_i), \text{ for $i \in [d]$}, \\
&\left(\lab_{d+1,x}, \widetilde{g}_{d+1,x}\right) \gets \mathsf{GSim}\left(D_{d+1,x}\right), \text{ for $x \in \{0,1\}^{\lambda}$} \bigg\}
\end{align*}
It is straightforward to see that latter is the same distribution as:
\begin{align*}
\bigg\{D_0^{\dagger} \bigg(D_1^{\dagger}\otimes \bbI\bigg) \bigg(\cdot \cdot\Big(D_{d}^{\dagger}\otimes \bbI\Big) \Big( (D_{d+1}^{\dagger} \otimes \bbI) &\left( \mathbb{E}_{x \gets \mathsf{Meas}(Q_d\left(\bx_\inp\right))} \left[f_{d+1}(x)Q_d^{x} \left(\bx_\inp\right)\right] \otimes \lab_{d+1}\right) \otimes \lab_{d} \Big) \otimes \cdot \cdot \otimes \lab_1 \bigg), \\
&D_0, \widetilde{g}_1,\ldots, \widetilde{g}_{d},\widetilde{g}_{d+1} \bigg]: \\
\,\,\,\,\,\,& D_i \gets \mathscr{C}_{n_i + h_i \lambda}, \,i\in\{0,\ldots,d+1\},  \\
&(\lab_{i}, \widetilde{g}_i) \gets \mathsf{GSim}( D_i), \text{ for $i \in [d+1]$} \bigg\}
\end{align*}
i.e. sampling the same $D_{d+1}$ and simulated garbling output for all $x$ results in the same distribution. Finally, we can rewrite the latter as:
\begin{align*}
\bigg\{D_0^{\dagger} \bigg(D_1^{\dagger}\otimes \bbI\bigg) \bigg(\cdot \cdot\Big(D_{d}^{\dagger}\otimes \bbI\Big) \Big( (D_{d+1}^{\dagger} \otimes \bbI) &\left( Q\left(\bx_\inp\right) \otimes \lab_{d+1}\right) \otimes \lab_{d} \Big) \otimes \cdot \cdot \otimes \lab_1 \bigg), \\
&D_0, \widetilde{g}_1,\ldots, \widetilde{g}_{d},\widetilde{g}_{d+1} \bigg]: \\
\,\,\,\,\,\,& D_i \gets \mathscr{C}_{n_i + h_i \lambda}, \,i\in\{0,\ldots,d+1\}, \\
&(\lab_{i}, \widetilde{g}_i) \gets \mathsf{GSim}( D_i), \text{ for $i \in [d+1]$} \bigg\} \,,
\end{align*}
as desired.
\pagebreak


\end{proof}


\fi

\section{Quantum Non-Interactive Secure Computation}
\label{sec:three-message}

\subsection{The Protocol}\label{subsec:3-msg-protocol}

In what follows, we describe our protocol for two-party quantum computation in the setting of sequential messages. This protocol requires three messages of interaction when both players desire output, and two messages in a setting where only one party obtains an output, which can be seen as a Q-NISC (Quantum Non-interactive Secure Computation) protocol.

\paragraph{Ingredients.} Our protocol will make use of the following cryptographic primitives: (1) Quantum-secure two-message two-party classical computation in the CRS model $(\twopc.\gen,\twopc_1,\twopc_2,\twopc_\out)$ with a straight-line black-box simulator (\ifsubmission\else\cref{subsec:2pc}\fi), and (2) a garbling scheme for $\CM$ circuits $(\QGarble, \QGEval, \QGSim)$. 

\paragraph{Notation.}

The protocol below computes a two-party quantum functionality represented by a $\CM$ circuit $Q$ that takes $n_A + n_B$ input qubits, produces $m_A + m_B$ output qubits, and requires $n_Z$ auxiliary $\Zstate$ states and $n_T$ auxiliary $\Tstate$ states. Let $\secp$ be the security parameter. The total number of quantum registers used will be $s = n_A + (n_B + \secp) + (2n_Z + \secp) + (n_T + 1)\secp$, and we'll give a name to different groups of these registers.

\ifsubmission
\else
In round 1, $B$ operates on $n_B + \secp$ registers, partitioned as $(\gray{B},\gray{\Trap_B})$, and sends these registers to $A$. In round 2, $A$ operates on these registers, along with $\gray{A}$ of size $n_A$, $\gray{Z_A}$ of size $2n_Z$, $\gray{\Trap_A}$ of size $\secp$, and $\gray{T_A}$ of size $(n_T+1)\secp$. An honest party $A$ will return all registers to $B$ in the order $(\gray{A},\gray{B},\gray{\Trap_B},\gray{Z_A},\gray{\Trap_A},\gray{T_A})$. During party $B$'s subsequent computation, the register $\gray{Z_A}$ will be partitioned into two registers $(\gray{Z_\inp},\gray{Z_\chck})$, where each has size $n_Z$, and register $\gray{T_A}$ will be partitioned into two registers $(\gray{T_\inp},\gray{T_\chck})$, where $\gray{T_\inp}$ has size $n_T\secp$ and $\gray{T_\chck}$ has size $\secp$. 
\fi

Given a $\CM$ circuit $Q$ and a Clifford $C_\out \in \mathscr{C}_{m_A + \secp}$, we define another $\CM$ circuit $Q_\dst[C_\out]$. This circuit takes as input $n_A + n_B + n_Z + \secp + n_T\secp$ qubits $(\bx_A,\bx_B,\bz_\inp,\trap_A,\bt_\inp)$ on registers $(\gray{A},\gray{B},\gray{Z_\inp},\gray{\Trap_A},\gray{T_\inp})$. It will first apply the magic state distillation circuit from~\cref{lemma:distillation} with parameters $(n_T\secp,\secp)$ to $\bt_\inp$ to produce QRV $\bt$ of size $n_T$. It will then run $Q$ on $(\bx_A,\bx_B,\bz_\inp,\bt)$ to produce $(\by_A,\by_B)$. Finally, it will output $(C_\out(\by_A,\trap_A),\by_B)$.

\protocol
{\proref{fig:classical-f}: Classical Functionality $\mathcal{F}[Q]$}
{Classical functionality to be used in \proref{fig:three-message}.}
{fig:classical-f}
{
\textbf{Common Information:} Security parameter $\secp$, and $\CM$ circuit $Q$ to be computed with $n_A + n_B$ input qubits, $m_A + m_B$ output qubits, $n_Z$ auxiliary $\Zstate$ states, and $n_T$ auxiliary $\Tstate$ states. Let $s = n_A + (n_B + \secp) + (2n_Z + \secp) + (n_T + 1)\secp$.\\

\textbf{Party A Input:} Classical descriptions of $C_A \in \mathscr{C}_s$ and $C_\out \in \mathscr{C}_{m_A + \secp}$.\\
\textbf{Party B Input:} Classical description of $C_B \in \mathscr{C}_{n_B + \secp}$.\\

\textbf{The Functionality:}

\begin{enumerate}
\item Sample the unitary $U_{\decchck}$ as follows:
\begin{itemize}
    \item Sample a random permutation $\pi$ on $(n_T+1)\secp$ elements.
    \item Sample a random element $M \gets \mathsf{GL}(2n_T,\bbF_2)$.
    \item Compute a description of the Clifford $U_\chck$ that operates as follows on registers $(\gray{A},\gray{B},\gray{\Trap_B},\gray{Z_A},\gray{\Trap_A},\gray{T_A})$.
    \begin{enumerate}
        \item Rearrange the registers of $\gray{T_A}$ according to the permutation $\pi$ and then partition the registers into $(\gray{T_\inp},\gray{T_\chck})$.
        \item Apply the linear map $M$ to the registers $\gray{Z_A}$ and then partition the registers into $(\gray{Z_\inp},\gray{Z_\chck})$.
        \item Re-arrange the registers to $(\gray{A},\gray{B},\gray{Z_\inp},\gray{\Trap_A},\gray{T_\inp},\gray{Z_\chck},\gray{\Trap_B},\gray{T_\chck})$.
    \end{enumerate}
    \item Define $U_{\decchck}$ as: $$U_{\decchck} \coloneqq U_\chck\left(\bbI^{n_A} \otimes C_B^\dagger \otimes \bbI^{(2n_Z + \secp) + (n_T+1)\secp}\right)C_A^\dagger.$$
\end{itemize}
\item Sample $(E_0,D_0,\widetilde{g}_1,\dots,\widetilde{g}_d) \gets \QGarble(1^\secp,Q_\dst[C_\out])$.
\item Compute a description of $U_{\decchckenc} \coloneqq \left(E_0 \otimes \bbI^{(n_Z+\secp)+\secp}\right)U_{\decchck}^\dagger.$
\end{enumerate}

\textbf{Party B Output:} (1) A unitary $U_{\decchckenc}$ on $s$ qubits (to be applied on registers $(\gray{A},\gray{B},\gray{\Trap_B},\gray{Z_A},\gray{\Trap_A},\gray{T_A})$), and (2) A QGC $(D_0,\widetilde{g}_1,\dots,\widetilde{g}_d)$ (to be applied to registers $(\gray{A},\gray{B},\gray{Z_\inp},\gray{\Trap_A},\gray{T_\inp})$).\\
}
\clearpage

\protocol
{\proref{fig:three-message}: Three-message two-party quantum computation}
{Three-message two-party quantum computation.}
{fig:three-message}
{
\textbf{Common Information:} (1) Security parameter $\secp$, and (2) a $\CM$ circuit $Q$ over $n_A + n_B$ input qubits, $m_A + m_B$ output qubits, $n_Z$ auxiliary $\Zstate$ states, and $n_T$ auxiliary $\Tstate$ states. Let $s = n_A + (n_B + \secp) + (2n_Z + \secp) + (n_T + 1)\secp$.\\

\textbf{Party A Input:} $\bx_A$\\
\textbf{Party B Input:} $\bx_B$\\

\underline{\textbf{The Protocol:}}\\
\textbf{Setup.} Run classical 2PC setup: $\crs \gets \twopc.\gen(1^\secp)$.\\

\textbf{Round 1.} \emph{Party $B$:}
\begin{enumerate}
    \item Sample $C_B \gets \mathscr{C}_{n_B + \secp}$ and compute $\bm_{B,1} \coloneqq C_B(\bx_B, \Zstate^{\secp})$.
    \item Compute $(m_{B,1},\st) \gets \twopc_1(1^\secp,\cF[Q],\crs,C_B)$.
    \item Send to Party $A$: $(m_{B,1},\bm_{B,1})$.
\end{enumerate}

\textbf{Round 2.} \emph{Party $A$:}
\begin{enumerate}
    \item Sample $C_A \gets \mathscr{C}_s$ and $C_\out \gets \mathscr{C}_{m_A + \secp}$.
    \item Compute $\bm_{A,2} \coloneqq C_A(\bx_A,\bm_{B,1},\Zstate^{2n_Z},\Zstate^\secp,\Tstate^{(n_T+1)\secp})$.
    \item Compute
    $m_{A,2} \gets \twopc_2(1^\secp,\cF[Q],\crs,m_{B,1},(C_A,C_\out))$.
    \item Send to Party $B$: $(m_{A,2},\bm_{A,2})$.
\end{enumerate}

\textbf{Round 3.} \emph{Party $B$:}
\begin{enumerate}
    \item Compute $(U_{\decchckenc},D_0,\tilde{g}_1,\dots,\tilde{g}_d) \gets \twopc_\out(1^\secp,\st,m_{A,2})$.
    \item Compute $(\bm_\inp,\bz_\chck,\trap_B,\bt_\chck) \coloneqq U_{\decchckenc}(\bm_2)$, where 
    \begin{itemize}
        \item $\bm_\inp$ is on registers $(\gray{A},\gray{B},\gray{Z_\inp},\gray{\Trap_A},\gray{T_\inp})$,
        \item $(\bz_\chck,\trap_B,\bt_\chck)$ is on registers $(\gray{Z_\chck},\gray{\Trap_B},\gray{T_\chck})$.
    \end{itemize}
    \item Measure each qubit of $(\bz_\chck,\trap_B)$ in the standard basis and abort if any measurement is not zero.
    \item Measure each qubit of $\bt_\chck$ in the $T$-basis and abort if any measurement is not zero.
    \item Compute $(\widehat{\by}_A,\by_B) \gets \QGEval((D_0,\tilde{g}_1,\dots,\tilde{g}_d),\bm_\inp)$, where $\widehat{\by}_A$ consists of $m_A + \secp$ qubits and $\by_B$ consists of $m_B$ qubits.
    \item Send to Party $A$: $\widehat{\by}_A$.
\end{enumerate}

\textbf{Output Reconstruction.}
\begin{itemize}
\item \emph{Party $A$:} Compute $(\by_A,\trap_A) \coloneqq C_\out^\dagger(\widehat{\by}_A)$, where $\by_A$ consists of $m_A$ qubits and $\trap_A$ consists of $\secp$ qubits. Measure each qubit of $\trap_A$ in the standard basis and abort if any measurement is not zero. Otherwise, output $\by_A$.
\item \emph{Party $B$:} Output $\by_B$.
\end{itemize}
}

\clearpage

\subsection{Security} 

\begin{theorem}
Assuming post-quantum maliciously-secure two-message oblivious transfer, there exists maliciously-secure NISC for quantum computation and maliciously-secure three-message two-party quantum computation.
\end{theorem}

\begin{proof}
Let $\Pi$ be the protocol described in~\proref{fig:three-message} computing some quantum circuit $Q$. \ifsubmission\else We first show that $\Pi$ satisfies \cref{def:mpqc} for any $\cA$ corrupting party $A$.\fi

\paragraph{The simulator.} Consider any QPT adversary $\cA = \{\cA_\secp\}_{\secp \in \bbN}$ corrupting party A. The simulator $\Sim$ is defined as follows. Whenever we say that the simulator aborts, we mean that it sends $\bot$ to the ideal functionality and to the adversary.

\paragraph{$\Sim^{\cI[\bx_B](\cdot)}(\bx_A,\baux_\cA)$:}

\begin{itemize}
    \item Compute $(\crs,\tau,m_{B,1}) \gets \twopc.\Sim_A^{(1)}(1^\secp)$, sample $C_B \gets \mathscr{C}_{n_B + \secp}$, compute $\bm_{B,1} \coloneqq C_B(\Zstate^{n_B},\Zstate^{\secp})$, and send $(\crs,m_{B,1},\bm_{B,1})$ to $\cA_\secp(\bx_A,\baux_\cA)$.
    \item Receive $(m_{A,2},\bm_{A,2})$ from $\cA_\secp$ and compute $\out \gets \twopc.\Sim_A^{(1)}(1^\secp,\tau,m_{A,2})$. If $\out = \bot$ then abort. Otherwise, parse $\out$ as $(C_A,C_\out)$.
    \item Using $(C_A,C_B)$, sample $U_{\decchck}$ as in the description of $\cF[Q]$. Compute $$(\bx_A',\bx_B',\bz_\inp,\trap_A,\bt_\inp,\bz_\chck,\trap_B,\bt_\chck) \coloneqq U_{\decchck}(\bm_{A,2}).$$ Measure each qubit of $\bz_\chck$ and $\trap_B$ in the standard basis and each qubit of $\bt_\chck$ in the $T$-basis. If any measurement is non-zero, then abort.
    \item Forward $\bx_A'$ to $\cI[\bx_B](\cdot)$ and receive back $\by_A$. Compute $\widehat{\by}_A \coloneqq C_\out(\by_A,\trap_A)$, send $\widehat{\by}_A$ to $\cA_\secp$, send $\mathsf{ok}$ to $\cI[\bx_B]$, and output the output of $\cA_\secp$.
\end{itemize}

We consider a sequence of hybrid distributions, where the first hybrid $\cH_0$ is $\Real_{\Pi,\Q}(\cA_\secp,\bx_A,\bx_B,\baux_\cA)$, i.e. the real interaction between the adversary  $\cA_\secp(\bx_A,\baux_\cA)$ and an honest party $B(1^\secp,\bx_B)$. In each hybrid, we describe the differences from the previous hybrid.

\begin{itemize}
    \item $\cH_1$: Simulate $\twopc$ as described in $\Sim$, using $\twopc.\Sim_A^{(1)}$ to compute $m_{B,1}$ and $\twopc.\Sim_A^{(2)}$ to extract an input $(C_A,C_\out)$ (or abort). Use $(C_A,C_\out)$ to sample an output $(U_{\decchckenc},D_0,\widetilde{g}_1,\dots,\widetilde{g}_d)$ of the classical functionality. Use this output to run party $B$'s honest Message 3 algorithm.
    \item $\cH_2$: In this hybrid, we change how $B$'s third round message $\widehat{\by}_A$ is sampled. In particular, rather than evaluating the quantum garbled circuit on $\bm_\inp$, we will directly evaluate $Q_\dst[C_\out]$ on the input. In more detail, given $\bm_{A,2}$ returned by $\cA_\secp$, $(C_A,C_\out)$ extracted from $\cA_\secp$, and $C_B$ sampled in Message 1, $\widehat{\by}_A$ is sampled as follows. Sample $U_{\decchck}$ as in Step 1 of $\cF[Q]$. Compute $$(\bx_A',\bx_B',\bz_\inp,\trap_A,\bt_\inp,\bz_\chck,\trap_B,\bt_\chck) \coloneqq U_{\decchck}(\bm_{A,2})$$ and carry out the checks on $\bz_\chck,\trap_B,\bt_\chck$ as described in Steps 3.(c) and 3.(d) of~\proref{fig:three-message}, aborting if needed. Then, compute $$(\widehat{\by}_A,\by_B) \gets Q_\dst[C_\out](\bx_A',\bx_B',\bz_\inp,\trap_A,\bt_\inp)$$ and return $\widehat{\by}_A$ to $\cA_\secp$.
    \item $\cH_3$: Compute $\bm_{B,1}$ as $C_B(\Zstate^{n_B},\Zstate^{\secp})$, and substitute $\bx_B$ for $\bx_B'$ before applying $Q_\dst[C_\out]$ to the registers described above in $\cH_2$.
    \item $\cH_4$: Rather than directly computing $Q_\dst[C_\out]$, query the ideal functionality with $\bx_A'$, receive $\by_A$, and send $\widehat{\by}_A \coloneqq C_\out(\by_A,\trap_A)$ to $\cA_\secp$. This hybrid is $\Ideal_{\Pi,\Q,A}(\Sim,\brho_\secp,\bx_A,\bx_B,\baux)$.
\end{itemize}

\noindent We show indistinguishability between each pair of hybrids.

\begin{itemize}
    \item $\cH_0 \approx_c \cH_1$: This follows from the security against corrupted $A$ of $\twopc$.
    \item $\cH_1 \approx_s \cH_2$: This follows from the statistical correctness of QGC.
    \item $\cH_2 \approx_s \cH_3$: First, by the security of the Clifford authentication code, conditioned on all measurements of qubits in $\trap_B$ returning 0, we have that $\bx_B' \approx_s \bx_B$. Next, switching $\bx_B$ to $\Zstate^{n_B}$ in $B$'s first message is perfectly indistinguishable due to the perfect hiding of the Clifford authentication code. 
    \item $\cH_3 \approx_s \cH_4$: First, by~\cref{lemma:linear-map}, conditioned on all measurements of qubits in $\bz_\chck$ returning 0, we have that $\bz_\inp \approx_s \Zstate^{n_Z}$.
    
    Next, the above observation, along with~\cref{lemma:distillation}, implies that, conditioned on all $T$-basis measurements of qubits in $\bt_\chck$ returning 0, it holds that the output of $Q_\dst[C_\out](\bx_A',\bx_B,\bz_\inp,\trap_A,\bt_\inp)$ is statistically close to the result of computing $(\by_A,\by_B) \gets Q(\bx_A',\bx_B,\Zstate^{n_Z},\Tstate^{n_T})$ and returning $(C_\out(\by_A,\trap_A),\by_B)$. This is precisely what is being computed in $\cH_4$.

\end{itemize}

\end{proof}

\paragraph{On Reusable Security against Malicious A.}
We remark that the two-message special case of the above protocol, that is, our Quantum NISC protocol, can be lightly modified to also achieve {\em reusable} security. A reusable classical NISC protocol (see, eg.~\cite{C:CDIKLOV19}) retains security against malicious A in a setting where A and B execute many instances of secure computation that {\em reuse} the first message of B. 
A natural quantum analogue of this protocol enables computation of quantum circuits while guaranteeing security against malicious A, in a setting where A and B execute many instances of secure computation that {\em reuse} the first message of B. Here we assume that B's input is classical, and so functionality will hold over repeated executions.
We note that our protocol can be lightly modified to achieve reusable security against malicious A, by replacing the underlying classical 2PC with a reusable classical 2PC. 
The proof of security remains identical, except that the indistinguishability between hybrids 0 and 1 relies on the reusable security of the underlying classical two-party computation protocol. 
In \ifsubmission\else\cref{sec:mdv-nizk}\fi, we discuss how to achieve reusable MDV-NIZKs for NP, which can be viewed as a special case of reusable Q-NISC.

\ifsubmission
\else

Next, we show that $\Pi$ satisfies \cref{def:mpqc} for any $\cA$ corrupting party $B$.

\paragraph{The simulator.} Consider any QPT adversary $\cA = \{\cA_\secp\}_{\secp \in \bbN}$ corrupting party $B$. The simulator $\Sim$ is defined as follows.

\paragraph{$\Sim^{\cI[\bx_A](\cdot)}(\bx_B,\baux_\cA):$}
\begin{itemize}
    \item Simulate CRS and extract from adversary's round 1 message:
    \begin{itemize}
        \item Compute $(\crs,\tau) \gets \twopc.\Sim_B^{(1)}(1^\secp)$ and send $\crs$ to the adversary $\cA_\secp(\bx_B,\baux_\cA)$.
        \item Receive $(m_{B,1},\bm_{B,1})$ from $\cA_\secp$ and compute $\inp \gets \twopc.\Sim_B^{(2)}(1^\secp,\tau,m_{B,1})$. If $\inp = \bot$ then abort. Otherwise, parse $\inp$ as $C_B$ and compute $(\bx'_B,\trap_B) \coloneqq C_B^\dagger(\bm_{B,1})$.
    \end{itemize}
    \item Query ideal functionality and compute simulated round 2 message:
    \begin{itemize}
        \item Forward $\bx_B'$ to $\cI[\bx_A](\cdot)$ and receive back $\by_B$.
        \item Sample $C_\out \gets \mathscr{C}_{m_A + \secp}$ and compute $\widehat{\by}_A' \coloneqq C_\out(\Zstate^{m_A + \secp})$.
        \item Compute $(\widetilde{\bm}_\inp,D_0,\widetilde{g}_1,\dots,\widetilde{g}_d) \gets \QGSim\left(1^\secp,\{n_i,k_i\}_{i \in [d]}, \left(\widehat{\by}_A',\by_B\right)\right)$, where $\widetilde{\bm}_\inp$ is the simulated quantum garbled input on registers $(\gray{A},\gray{B},\gray{Z_\inp},\gray{\Trap_A},\gray{T_\inp})$, and $\{n_i,k_i\}_{i \in [d]}$ are the parameters of $\CM$ circuit $Q_\dst[C_\out]$.
        \item Sample $U_{\decchckenc} \gets \mathscr{C}_s$ and compute $\bm_{A,2} \coloneqq U_{\decchckenc}^\dagger(\widetilde{\bm}_\inp,\Zstate^{n_Z},\trap_B,\Tstate^{\secp})$. 
        \item Compute $m_{A,2} \gets \twopc.\Sim_B^{(3)}(1^\secp,\tau,(U_{\decchckenc},D_0,\widetilde{g}_1,\dots,\widetilde{g}_d))$.
        \item Send $(m_{A,2},\bm_{A,2})$ to $\cA_\secp$. 
    \end{itemize}
    \item Check for abort:
    \begin{itemize}
         \item Receive $\widehat{\by}_A$ from $\cA_\secp$ and measure the last $\secp$ qubits of $C_\out^\dagger(\widehat{\by}_A)$. If any measurement is not zero, send $\abort$ to the ideal functionality and otherwise send $\mathsf{ok}$.
    \item Output the output of $\cA_\secp$.
    \end{itemize}
\end{itemize}

\begin{theorem}
Let $\Pi$ be the protocol described in~\proref{fig:three-message} computing some quantum circuit $Q$. Then $\Pi$ satisfies \cref{def:mpqc} for any $\cA$ corrupting party $B$.

\end{theorem}

\begin{proof}
We consider a sequence of hybrid distributions, where $\cH_0$ is $\Real_{\Pi,\Q}(\cA_\secp,\bx_A,\bx_B,\baux_\cA)$, i.e. the real interaction between $\cA_\secp(\bx_B,\baux_\cA)$ and an honest party $A(1^\secp,\bx_A)$. In each hybrid, we describe the differences from the previous hybrids.

\begin{itemize}
    \item $\cH_1$: Simulate $\twopc$, using $\twopc.\Sim_B^{(1)}$ to sample $\twopc.\crs$, $\twopc.\Sim_B^{(2)}$ to extract the adversary's input $C_B$, and $\twopc.\Sim_B^{(3)}$ to sample party $A$'s message $m_{A,2}$. Use $C_B$ and freshly sampled $(C_A,C_\out)$ to sample the output of the classical functionality that is given to $\twopc.\Sim_B^{(3)}$.
    \item $\cH_2$: In this hybrid, we make a (perfectly indistinguishable) switch in how $\bm_{A,2}$ is computed and how $U_{\decchckenc}$ (part of the classical $\twopc$ output) is sampled. Define $(\bx'_B,\trap_B) \coloneqq C_B^\dagger(\bm_{B,1})$, where $C_B$ was extracted from $m_{B,1}$. Note that in $\cH_1$, by the definition of $\cF[Q]$,  $$U_{\decchckenc}(\bm_{A,2}) \coloneqq (E_0(\bx_A,\bx_B',\Zstate^{n_Z+\secp},\Tstate^{n_T\secp}),\Zstate^{n_Z},\trap_B,\Tstate^\secp).$$
    Moreover, there exists a Clifford unitary $U$ such that $U_{\decchckenc} = UC_A^\dagger$, where $C_A$ was sampled uniformly at random from $\mathscr{C}_s$. Thus, since the Clifford matrices form a group, an equivalent sampling procedure would be to sample $U_{\decchckenc} \gets \mathscr{C}_s$ and define $$\bm_{A,2} \coloneqq U_{\decchckenc}^\dagger(E_0(\bx_A,\bx_B',\Zstate^{n_Z+\secp},\Tstate^{n_T\secp}),\Zstate^{n_Z},\trap_B,\Tstate^\secp).$$ This is how $\cH_2$ is defined.
    \item $\cH_3$: In this hybrid, we simulate the quantum garbled circuit. In particular, compute $$(\widehat{\by}_A,\by_B) \gets Q_\dst[C_\out](\bx_A,\bx_B',\Zstate^{n_Z+\secp},\Tstate^{n_T\secp}),$$ followed by $$(\widetilde{\bm}_\inp,D_0,\widetilde{g}_1,\dots,\widetilde{g}_d) \gets \QGSim(1^{\secp},\{n_i,k_i\}_{i \in [d]},(\widehat{\by}_A,\by_B)).$$ Finally, substitute $\widetilde{\bm}_\inp$ for $E_0(\bx_A,\bx_B',\Zstate^{n_Z+\secp},\Tstate^{n_T\secp})$ in the computation of $\bm_{A,2}$ so that $$\bm_{A,2} \coloneqq U_{\decchckenc}^\dagger(\widetilde{\bm}_\inp,\Zstate^{n_Z},\trap_B,\Tstate^\secp).$$
    \item $\cH_4$: Note that $Q_\dst[C_\out](\bx_A,\bx_B',\Zstate^{n_Z+\secp},\Tstate^{n_T\secp})$ may be computed in two stages, where the first outputs $(\by_A,\by_B,\Zstate^\secp,C_\out)$ and the second outputs $(C_\out(\by_A,\Zstate^\secp),\by_B)$. In this hybrid, compute only the first stage, set $\by_A$ aside and re-define the final output to be $(\widehat{\by}_A',\by_B) \coloneqq (C_\out(\Zstate^{m_A+\secp}),\by_B)$.
    
    Now, during $A$'s output reconstruction step, if the check (step 4.(b) of~\proref{fig:three-message}) passes, output $\by_A$, and otherwise abort.
    \item $\cH_5$: Instead of directly computing $\by_B$ from the first stage of $Q_\dst[C_\out](\bx_A,\bx_B',\Zstate^{n_Z+\secp},\Tstate^{n_T\secp})$, forward $\bx_B'$ to $\cI[\bx_A](\cdot)$ and receive back $\by_B$. Now, during party $A$'s output reconstruction step, if the check passes, send $\mathsf{ok}$ to the ideal functionality, and otherwise send $\abort$ to the ideal functionality. This is $\Ideal_{\Pi,\Q,B}(\Sim,\brho_\secp,\bx_A,\bx_B,\baux)$.
\end{itemize}

\noindent We show indistinguishability between each pair of hybrids.

\begin{itemize}
    \item $\cH_0 \approx_c \cH_1$: This follows directly from the security against corrupted $B$ of $\twopc$.
    \item $\cH_1 \equiv \cH_2$: This is argued above.
    \item $\cH_2 \approx_c \cH_3$: This follows directly from the security of the QGC.
    \item $\cH_3 \approx_s \cH_4$: This follows directly from the (perfect) hiding and (statistical) authentication of the Clifford code.
    \item $\cH_4 \equiv \cH_5$: This follows from the definition of $\cI[\bx_A](\cdot)$.
\end{itemize}

\end{proof}

\fi

\section{Application: Reusable Malicious Designated Verifier NIZK for QMA}
\label{sec:mdv-nizk}

In this section, we show how a small tweak to the protocol from last section gives a reusable MDV-NIZK for QMA. Features of our construction differ from those of \cite{shmueli2020multitheorem} in the following ways.

\begin{itemize}
    \item It assumes post-quantum OT and reusable MDV-NIZK for NP, whereas \cite{shmueli2020multitheorem} assumed (levelled) fully-homomorphic encryption (note that both assumptions are known from QLWE).
    \item The prover only requires a single copy of the witness state, whereas \cite{shmueli2020multitheorem} required the prover to have access to polynomially-many identical copies of the witness.
\end{itemize}

\begin{definition}[MDV-NIZK Argument for QMA]\label{def:mdv-nizk}
A non-interactive computational zero-knowlege argument for a language $\cL = (\cL_\yes,\cL_\no) \in \QMA$ in the malicious designated-verifier model consists of $4$ algorithms $(\Setup,\VSetup,\Prove,\Verify)$ with the following syntax.

\begin{itemize}
    \item $\crs \gets \Setup(1^\secp)$: A classical PPT algorithm that on input the security parameter samples a common uniformly random string $\crs$.
    \item $(\pvk,\svk) \gets \VSetup(\crs)$: A classical PPT algorithm that on input $\crs$ samples a pair of public and secret verification keys.
    \item $\bpi \gets \Prove(\crs,\pvk,x,\bw)$: A QPT algorithm that on input $\crs$, the public verification key, an instance $x \in \cL_\yes$, and a quantum witness $\bw$, outputs a quantum state $\bpi$.
    \item $\Verify(\crs,\svk,x,\bpi)$: A QPT algorithm that on input $\crs$, secret verification key $\svk$, and instance $x \in \cL$, and a quantum proof $\bpi$, outputs a bit indicating acceptance or rejection.
\end{itemize}
The protocol satisfies the following properties.
\begin{itemize}
    \item \textbf{Statistical Completeness:} There exists a negligible function $\mu(\cdot)$ such that for any $\secp \in \bbN$, $x \in \cL_\yes \cap \{0,1\}^\secp$, $\bw \in \cR_\cL(x)$, $\crs \in \Setup(1^\secp)$, $(\pvk,\svk) \in \VSetup(\crs),$ $$\Pr_{\bpi \gets \Prove(\crs,\pvk,x,\bw)}\left[\Verify(\crs,\svk,x,\bpi)\right] \geq 1 -\mu(\secp).$$
    \item \textbf{Reusable (Non-Adaptive) Soundness:\footnote{A previous version of this paper defined and claimed to achieve adaptive soundness from polynomially-hard assumptions. However, we actually only achieve \emph{non-adaptive} soundness from polynomially-hard QLWE. Similar to \cite{shmueli2020multitheorem}, we could upgrade security to adaptive soundness if we use complexity leveraging and assume sub-exponentially secure QLWE.}} For every quantum polynomial-size adversarial prover $\cP^* = \{\cP^*_\secp,\bp_\secp\}_{\secp \in \bbN}$ and $\{x_\secp\}_{\secp \in \bbN}$ where for each $\secp \in \bbN, x_\secp \in \cL_\mathsf{no}$, there exists a negligible function $\mu(\cdot)$ such that for every $\secp \in \bbN$, $$\Pr_{\substack{\crs \gets \Setup(1^\secp) \\ (\pvk,\svk) \gets \VSetup(\crs) \\ \bpi \gets \cP^*_\secp(\bp_\secp,\crs,\pvk)^{\Verify(\crs,\svk,\cdot,\cdot)}}}\left[1 = \Verify(\crs,\svk,x,\bpi)\right] \leq \mu(\secp).$$
    \item \textbf{Malicious Zero-Knowledge:} There exists a QPT simulator $\Sim$ such that for every QPT distinguisher $\cD = \{\cD_\secp,\bd_\secp\}_{\secp \in \bbN}$, there exists a negligible function $\mu(\cdot)$ such that for every $\secp \in \bbN$, $$\left|\Pr_{\crs \gets \Setup(1^\secp)}\left[\cD_\secp(\bd_\secp,\crs)^{\Prove(\crs,\cdot,\cdot,\cdot)}\right] - \Pr_{(\crs,\tau) \gets \Sim(1^\secp)}\left[\cD_\secp(\bd_\secp,\crs)^{\Sim(\tau,\cdot,\cdot)}\right]\right| \leq \mu(\secp),$$ where,
    \begin{itemize}
        \item Every query $\cD_\secp$ makes to the oracle is of the form $(\pvk^*,x,\bw)$, where $\pvk^*$ is arbitrary, $x \in \cL_\yes \cup \{0,1\}^\secp$, and $\bw \in \cR_\cL(s)$.
        \item $\Prove(\crs,\cdot,\cdot,\cdot)$ is the honest prover algorithm and $\Sim(\tau,\cdot,\cdot)$ acts only on $\tau,\pvk^*$, and $x$.
    \end{itemize}
\end{itemize}
\end{definition}

\begin{theorem}
Assuming post-quantum maliciously-secure two-message oblivious transfer and post-quantum reusable MDV-NIZK for NP (see the discussion following \cref{def:2pc}), there exists a reusable MDV-NIZK satisfying \cref{def:mdv-nizk}.
\end{theorem}

\begin{proof}
For $x \in \cL$, let $\cV_\cL[x](\cdot)$ be the QMA verification circuit that takes as input a potential witness $\bw$ and outputs a bit indicating acceptance or rejection. For any $x$, we will use \proref{fig:three-message} to compute the functionality $\cV_\cL[x](\cdot)$ (where Alice has input $\bw$ and only Bob obtains output) in two messages. Note that Bob has no input, and thus his first message is entirely classical, only consisting of the first message of the classical $\twopc$. This already gives a one-time MDV-NIZK.

Now, we argue how to achieve reusability, while maintaining soundness and zero-knowledge. First, we will instantiate the classical $\twopc$ with one that is reusable and post-quantum secure~\cite{C:LQRWW19}. Given such a $\twopc$ protocol, Bob can compute his first message independently of the statement to be proven by Alice, and Alice can subsequently re-use this first message to prove any number of statements. This already satisfies zero-knowledge, as the MDV-NIZK simulator can always just query the 2PQC simulator with output 1.

To achieve reusable soundness, we alter the classical functionality $\cF[\cV_\cL[x]]$ computed by the $\twopc$. It now takes as input a PRF key $k$ from Bob and generates all of Bob's randomness via $\PRF(k,x)$, i.e., the PRF applied to the (classical) instance $x$. This includes the auxiliary state checking randomness (permutation $\pi$ and linear map $M$) along with Bob's contribution to the classical randomness used to generate the quantum garbled circuit. To prove reusable soundness, let $\cP^*$ be a cheating prover, $x^* \in \cL_\mathsf{no}$ be a no instance, and consider the following hybrids.

\begin{itemize}
    \item $\cH_1$: The $\crs$ and the verifier's classical $\twopc$ message are generated by the $\twopc$ simulator, and the prover's oracle $\Verify(\crs,\svk,\cdot,\cdot)$ is now answered with help from the $\twopc$ simulator. That is, the $\twopc$ simulator extracts from the classical part of the prover's proof and computes the functionality specified by the instance $x$, which outputs the classical part of the quantum garbled circuit. This classical part is then used to compute the verifier's output on the quantum part of the prover's proof. Indistinguishability from the real game follows by security of the reusable classical $\twopc$.
    \item $\cH_2$: The PRF calls made during computation of the classical functionality are replaced with calls to a uniformly random function. Indistinguishability from the real game follows by security of the PRF.
\end{itemize}

In $\cH_2$, whenever the prover submits a proof for $x^*$, the randomness used to generate the $\mathbf{0}$ and $\mathbf{T}$ state checks and the quantum garbled circuit will be a string that is uniform and independent of the prover's view. Thus, by soundness of these checks, along with statistical correctness of the quantum garbled circuit, the verification oracle will output 0 with overwhelming probability. Thus, $P^*$ could not be making the verifier output 1 on any proof for $x^*$, except with negligible probability.

\end{proof}

\section{Two-Round Two-Party Quantum Computation with Preprocessing}
\label{sec:two-round-preprocess}

\ifsubmission\else This section presents a three-round protocol that only requires two rounds of online communication. This protocol can be equivalently interpreted as a two-round protocol with (quantum) pre-processing. \fi

\subsection{The Protocol}

\paragraph{Ingredients.} Our protocol will make use of the following cryptographic primitives, which are all assumed to be sub-exponentially secure (i.e. there exists $\epsilon$ such that the primitive is $(2^{-\secp^{\epsilon}})$-secure).

\begin{itemize}
    \item Quantum-secure two-message two-party classical computation in the CRS model where one party receives output $(\twopc.\gen,\twopc_1,\twopc_2,\twopc_\out)$ and with a straight-line black-box simulator\ifsubmission\else~(\cref{subsec:2pc})\fi.
    \item A garbling scheme for $\CM$ circuits $(\QGarble, \QGEval, \QGSim)$.
    \item A quantum multi-key FHE scheme $\qmfhe=(\KeyGen,\CEnc,\Enc,\Eval,\Rerand,\Dec)$ with ciphertext re-randomization and classical encryption of classical messages.
    \item A quantum-secure equivocal commitment $\Com = (\Com.\Gen,\Com.\Enc,\Com.\Ver)$.
    \item A quantum-secure classical garbled circuit $(\Garble,\GEval,\GSim)$.
\end{itemize}

\paragraph{Notation.}

The circuit $Q_\dst[C_\out,x_\out,z_\out]$ is defined like $Q_\dst[C_\out]$ from~\cref{subsec:3-msg-protocol} except that $X^{x_\out}Z^{z_\out}$ is applied to $B$'s output $\by_B$. $f_{\inpcor}[E_0,U_{\rerand}]$ is a classical ``input correction'' circuit that takes as input $x_\inp,z_\inp \in \{0,1\}^{n_B}$ and outputs $U_{\rerandenc} \coloneqq E_0\left(\bbI^{n_A} \otimes X^{x_\inp}Z^{z_\inp} \otimes \bbI^{n_Z + \secp + n_T\secp}\right)U_\rerand^\dagger.$ 

For a $2 \times n$ set of elements $\{a_{i,b}\}_{i \in [n], b \in \{0,1\}}$, and a string $x \in \{0,1\}^n$, we let $a^{(x)} \coloneqq \{a_{i,x_i}\}_{i \in [n]}$. We will use this notation below to refer to sets of public keys $\pk^{(x_\out,z_\out)}$, secret keys $\sk^{(x_\out,z_\out)}$, labels $\lab^{(x_\out,z_\out)}$, and random strings $r^{(x_\out,z_\out)}$. Let $c_\mathsf{lev}$ be a constant satisfying $c_\mathsf{lev} > 1/\epsilon$.

\protocol
{\proref{fig:classical-g}: Classical Functionality $\cG[Q,\Com.\crs]$}
{Classical functionality to be used in \proref{fig:two-online}.}
{fig:classical-g}
{
\textbf{Common Information:} (1) Security parameter $\secp$, (2) a $\CM$ circuit $Q$ on $n_A + n_B$ input qubits, $m_A + m_B$ output qubits, $n_Z$ auxiliary $\Zstate$ states, and $n_T$ auxiliary $\Tstate$ states, and (3) a crs $\Com.\crs$ for an equivocal commitment. Let $s = n_A + (n_B + \secp) + (2n_Z + \secp) + (n_T + 1)\secp$. Let $\secplev = \max\{\secp,(2n_B)^{c_\mathsf{lev}}\}$.\\

\textbf{Party A Input:} Classical descriptions of $C_A \in \mathscr{C}_s$, $C_\out \in \mathscr{C}_{m_A + \secp}$, $\{r_{i,b}\}_{i \in [2n_B], b \in \{0,1\}} \in (\{0,1\}^\secplev)^{4n_B}$, $x_\out,z_\out \in \{0,1\}^{m_B},s \in \{0,1\}^\secplev$.\\
\textbf{Party B Input:} Classical description of $C_B \in \mathscr{C}_{n_B + \secp}$.\\

\textbf{The Functionality:}
\begin{enumerate}
\item Sample $U_{\decchck}$ as in $\cF[Q]$, using $C_A$ and $C_B$.
\item Sample $(E_0,D_0,\widetilde{g}_1,\dots,\widetilde{g}_d) \gets \QGarble(1^\secplev,Q_\dst[C_\out,x_\out,z_\out])$. 
\item Sample $U_\rerand \gets \mathscr{C}_{n_A + n_B + n_Z + \secp + n_T}$.
\item Compute a description of $U_{\decchckrerand} \coloneqq \left(U_\rerand \otimes \bbI^{(n_Z+\secp)+\secp}\right)U_{\decchck}.$
\item Compute $(\{\lab_{i,b}\}_{i \in [2n_B],b \in \{0,1\}},\widetilde{f}_{\inpcor}) \gets \Garble(1^\secplev,f_{\inpcor}[E_0,U_\rerand])$.
\item For each $i \in [2n_B], b \in \{0,1\}$, compute $(\pk_{i,b},\sk_{i,b}) \coloneqq \qmfhe.\Gen(1^\secplev;r_{i,b})$ and $\ct_{i,b} \gets \qmfhe.\CEnc(\pk_{i,b},\lab_{i,b})$.
\item Compute $\cmt \coloneqq \Com.\Enc(\Com.\crs,(x_\out,z_\out);s)$.
\end{enumerate}

\textbf{Party B Output:} (1) A unitary $U_{\decchckrerand}$ to be applied to $s$ qubits, partitioned as registers $(\gray{A},\gray{B},\gray{\Trap_B},\gray{Z_A},\gray{\Trap_A},\gray{T_A})$, (2) a classical garbled circuit along with encryptions of its labels $\{\pk_{i,b},\ct_{i,b}\}_{i \in [2n_B],b \in \{0,1\}},\widetilde{f}_{\inpcor}$, (3) a QGC $(D_0,\widetilde{g}_1,\dots,\widetilde{g}_d)$ to be applied to registers $(\gray{A},\gray{B},\gray{Z_\inp},\gray{\Trap_A},\gray{T_\inp})$, and (4) a commitment $\cmt$.
}

\clearpage

\protocol
{\proref{fig:two-online}: Two-party quantum computation with two online rounds}
{Two-party quantum computation with two online rounds.}
{fig:two-online}
{
\textbf{Common Information:} Security parameter $\secp$, and $\CM$ circuit $Q$ to be computed with $n_A + n_B$ input qubits, $m_A + m_B$ output qubits, $n_Z$ auxiliary $\Zstate$ states, and $n_T$ auxiliary $\Tstate$ states. Let $s = n_A + (n_B + \secp) + (2n_Z + \secp) + (n_T + 1)\secp$. Let $\secplev = \max\{\secp,(2n_B)^{c_\mathsf{lev}}\}$.\\

\textbf{Party $A$ input:} $\bx_A$\\
\textbf{Party $B$ input:} $\bx_B$\\

\underline{\textbf{The Protocol:}}

\textbf{Setup.} Run classsical 2PC setup: $\twopc.\crs \gets \twopc.\gen(1^\secplev),\Com.\crs \gets \Com.\Gen(1^\secplev)$.\\

\textbf{Round 0 (pre-processing).}

\emph{Party $B$:}
\begin{enumerate}
    \item Prepare $n_B$ EPR pairs $\left\{\left(\be_1^{(i)},\be_2^{(i)}\right)\right\}_{i \in [n_B]}$. Let $\be_1$ denote $(\be_1^{(i)})_{i \in [n_B]}$ and $\be_2$ denote $(\be_2^{(i)})_{i \in [n_B]}$. 
    \item Sample $C_B \gets \mathscr{C}_{n_B + \secp}$ and compute $\bm_{B,1} \coloneqq C_B(\be_1, \Zstate^{\secp})$.
    \item Compute $(m_{B,1},\st) \gets \twopc_1(1^\secplev,\cG[Q,\Com.\crs],\twopc.\crs,C_B)$.
    \item Send to Party $A$: $(m_{B,1},\bm_{B,1})$.
\end{enumerate}

\textbf{Round 1.}

\emph{Party $A$:}
\begin{enumerate}
    \item Sample the following:
    \begin{itemize}
    \item a random Clifford $C_A \gets \mathscr{C}_s$,
    \item a random Clifford $C_\out \gets \mathscr{C}_{m_A + \secp}$,
    \item $4n_B$ random length-$\secplev$ bitstrings $\{r_{i,b}\}_{i \in [2n_B], b \in \{0,1\}}$,
    \item one random length-$\secplev$ bitstring $s$,
    \item two random length-$m_B$ bitstrings $x_\out,z_\out$.
    \end{itemize}
    \item Compute $\bm_{A,2} \coloneqq C_A(\bx_A,\bm_{B,1},\Zstate^{2n_Z},\Zstate^\secp,\Tstate^{(n_T+1)\secp})$.
    \item Compute 
    \[m_{A,2} \gets \twopc_2(1^\secplev,\cG[Q,\Com.\crs],\twopc.\crs,m_{B,1},(C_A,C_\out,\{r_{i,b}\}_{i,b},x_\out,z_\out,s)).\]
    \item Send to Party $B$: $(m_{A,2},\bm_{A,2})$.
\end{enumerate}

\emph{Party $B$:}
\begin{enumerate}
    \item Perform Bell measurements on each pair of corresponding qubits in $(\bx_B,\be_2)$, obtaining measurement outcomes $(x_\inp,z_\inp)$. 
    \item Send to Party $A$: $(x_\inp,z_\inp)$.
\end{enumerate}
}
\clearpage

\protocol
{\proref{fig:two-online}: Two-party quantum computation with two online rounds}
{Two-party quantum computation with two online rounds (continued).}
{fig:two-online}
{
\textbf{Round 2.}

\emph{Party $A$:}
\begin{enumerate}
    \item Send to Party $B$: $\left(r^{(x_\inp,z_\inp)},x_\out,z_\out,s\right)$.
\end{enumerate}
\emph{Party $B$:}
\begin{enumerate}
    \item Compute 
    \[\left(\begin{array}{c}U_{\decchckrerand},\{\pk_{i,b},\ct_{i,b}\}_{i,b},\\\widetilde{f}_{\inpcor},D_0,\tilde{g}_1,\dots,\tilde{g}_d,\cmt\end{array}\right) \gets \twopc_\out(1^\secplev,\st,m_{A,2}).\]
    \item Compute
    \[(\bm_\inp,\bz_\chck,\trap_B,\bt_\chck) \coloneqq U_{\decchckrerand}(\bm_{A,2}),\]
    where 
    \begin{itemize}
        \item $\bm_\inp$ is on registers $(\gray{A},\gray{B},\gray{Z_\inp},\gray{\Trap_A},\gray{T_\inp})$,
        \item $(\bz_\chck,\trap_B,\bt_\chck)$ is on registers $(\gray{Z_\chck},\gray{\Trap_B},\gray{T_\chck})$.
    \end{itemize}
    \item Measure $(\bz_\chck,\trap_B)$ in the standard basis and abort if any measurement is not zero.
    \item Measure each qubit of $\bt_\chck$ in the $T$-basis and abort if any measurement is not zero.
    \item Compute a ciphertext $\qmfhe.\Enc(\pk^{(x_\inp,z_\inp)},U_{\rerandenc})$ via homomorphic evaluation, where $U_{\rerandenc} \gets \GEval(\widetilde{f}_{\inpcor},\lab^{(x_\inp,z_\inp)})$.
    \item Compute a ciphertext $\qmfhe.\Enc(\pk^{(x_\inp,z_\inp)},(\widehat{\by}_A,\overline{\by}_B))$ via homomorphic evaluation, where $(\widehat{\by}_A,\overline{\by}_B) \gets \QGEval((D_0,\tilde{g}_1,\dots,\tilde{g}_d),U_{\rerandenc}(\bm_\inp))$.
    \item Apply $\qmfhe.\Rerand$ to the encryption of $\widehat{\by}_A$ and send the result $\qmfhe.\Enc(\pk^{(x_\inp,z_\inp)},\widehat{\by}_A)$.
\end{enumerate}
    
\textbf{Output Reconstruction.}
\begin{itemize}
\item \emph{Party $A$}: Use $\sk^{(x_\inp,z_\inp)}$ to decrypt  $\qmfhe.\Enc(\pk^{(x_\inp,z_\inp)},\widehat{\by}_A)$. If decryption fails, then abort. Compute $(\by_A,\trap_A) \coloneqq C_\out^\dagger(\widehat{\by}_A)$, where $\by_A$ consists of $m_A$ qubits and $\trap_A$ consists of $\secp$ qubits. Measure each qubit of $\trap_A$ in the standard basis and abort if any measurement is not zero. Otherwise, output $\by_A$.
\item \emph{Party $B$}: Use $r^{(x_\inp,z_\inp)}$ to generate $\pk^{(x_\inp,z_\inp)}, \sk^{(x_\inp,z_\inp)}$ and check that these public keys match the public keys obtained from the output of $\twopc$ in Round 2. If not, then abort. Use $\sk^{(x_\inp,z_\inp)}$ to decrypt $\qmfhe.\Enc(\pk^{(x_\inp,z_\inp)},\overline{\by}_B)$. If $\Com.\Ver(1^\secplev,\Com.\crs,\cmt,(x_\out,z_\out),s) = 1$, then compute and output $\by_B \coloneqq X^{x_\out}Z^{z_\out}\overline{\by}_B$, and otherwise abort.
\end{itemize}
}
\clearpage

\ifsubmission
\else
\subsection{Security}

\begin{theorem}
Assuming post-quantum maliciously-secure two-message oblivious transfer and (levelled) multi-key quantum fully homomorphic encryption with sub-exponential security, there exists maliciously-secure three-round two-party quantum computation. Both of the above assumptions are known from the sub-exponential hardness of QLWE.
\end{theorem}

\begin{proof}

Let $\Pi$ be the protocol described in~\proref{fig:two-online} computing some quantum circuit $Q$. First, we show that $\Pi$ satisfies \cref{def:mpqc} for any $\cA$ corrupting party $A$.

\paragraph{The simulator.} Consider any QPT adversary $\{\cA_\secp\}_{\secp \in \bbN}$ corrupting party $A$. The simulator $\Sim$ is defined as follows.

\paragraph{$\Sim^{\cI[\bx_B](\cdot)}(\bx_A,\baux_\cA)$:}

\begin{itemize}
    \item Compute $(\crs,\tau,m_{B,1}) \gets \twopc.\Sim_A^{(1)}(1^{\secplev})$, sample $C_B \gets \mathscr{C}_{n_B + \secp}$, compute $\bm_{B,1} \coloneqq C_B(\Zstate^{n_B},\Zstate^{\secp})$, sample $x_\inp,z_\inp \gets \{0,1\}^{n_B}$, and send $(m_{B,1},\bm_{B,1}), (x_\inp,z_\inp)$ to the adversary $\cA_\secp(\bx_A,\baux_\cA)$.
    \item Receive $(m_{A,2},\bm_{A,2})$ from $\cA_\secp$ and compute $\out \gets \twopc.\Sim_A^{(1)}(1^\secp,\tau,m_{A,2})$. If $\out = \bot$ then abort. Otherwise, parse $\out$ as $(C_A,C_\out,\{r_{i,b}\}_{i,b},x_\out,z_\out,s)$.
    \item Using $(C_A,C_B)$, sample $U_{\decchck}$ as in the description of $\cF[Q]$. Compute $$(\bx_A',\bx_B',\bz_\inp,\trap_A,\bt_\inp,\bz_\chck,\trap_B,\bt_\chck) \coloneqq U_{\decchck}(\bm_{A,2}).$$ Measure each qubit of $\bz_\chck$ and $\trap_B$ in the standard basis and each qubit of $\bt_\chck$ in the $T$-basis. If any measurement is non-zero, then abort.
    \item Forward $\bx_A'$ to $\cI[\bx_B](\cdot)$ and receive back $\by_A$. Compute $\widehat{\by}_A \coloneqq C_\out(\by_A,\trap_A)$, and send a re-randomized $\qmfhe.\Enc(\pk^{(x_\inp,z_\inp)},\widehat{\by}_A)$ to $\cA_\secp$, where $\pk^{(x_\inp,z_\inp)}$ are generated from $r^{(x_\inp,z_\inp)}$.
    \item Receive $\left(\{r'_i\}_{i \in [2n_B]},x'_\out,z'_\out,s'\right)$ from $\cA$ and check that:
    \begin{itemize}
        \item For all $i \in [2n_B], \pk_i'$ is equal to $\pk_i$, where $(\pk_i',\sk_i') \coloneqq \qmfhe.\Gen(1^{\secplev};r_i')$ and $(\pk_i,\sk_i) \coloneqq \qmfhe.\Gen(1^{\secplev};r_{i,(x_\inp,z_\inp)_i})$.
        \item $\Com.\Ver(1^{\secplev},\Com.\crs,\cmt,(x'_\out,z'_\out),s') = 1$, where $\cmt \coloneqq \Com.\Enc(1^{\secplev},\Com.\crs,(x_\out,z_\out);s)$.
    \end{itemize}
    If the checks pass send $\mathsf{ok}$ to $\cI[\bx_B]$ and otherwise send $\abort$.
\end{itemize}

We consider a sequence of hybrid distributions, where $\cH_0$ is $\Real_{\Pi,\Q}(\cA_\secp,\bx_A,\bx_B,\baux_\cA)$, i.e. the real interaction between $\cA_\secp(\bx_A,\baux_\cA)$ and an honest party $B(1^\secp,\bx_B)$. In each hybrid, we describe the differences from the previous hybrid.

\begin{itemize}
    \item $\cH_1$: Simulate $\twopc$ as described in $\Sim$, using $\twopc.\Sim_A^{(1)}$ to compute $m_{B,1}$ and $\twopc.\Sim_A^{(2)}$ to extract an input $(C_A,C_\out,\{r_{i,b}\}_{i,b},x_\out,z_\out,s)$ (or abort). Use $(C_A,C_\out,\{r_{i,b}\}_{i,b},x_\out,z_\out,s)$ to sample an output $(U_{\decchckrerand},D_0,\widetilde{g}_1,\dots,\widetilde{g}_d)$ of the classical functionality. Use this output to run party $B$'s honest Round 2 algorithm.
    \item $\cH_2$: In this hybrid, we change how $B$'s second round message $\qmfhe.\Enc(\pk^{(x_\inp,z_\inp)},\widehat{\by}_A)$ is sampled. In particular, rather than evaluating the classical garbled circuit and quantum garbled circuit under $\qmfhe$, we will directly evaluate $Q_\dst[C_\out,x_\out,z_\out]$ on the input. In more detail, given $\bm_{A,2}$ returned by $\cA_\secp$, $(C_A,C_\out,\{r_{i,b}\}_{i,b},x_\out,z_\out,s)$ extracted from $\cA_\secp$, and $C_B$ sampled in Message 0, $\widehat{\by}_A$ is sampled as follows. Sample $U_{\decchck}$ as in Step 1 of $\cF[Q]$. Compute $$(\bx_A',\overline{\bx}_B',\bz_\inp,\trap_A,\bt_\inp,\bz_\chck,\trap_B,\bt_\chck) \coloneqq U_{\decchck}(\bm_{A,2})$$ and carry out the checks on $\bz_\chck,\trap_B,\bt_\chck$ as described in Steps 3 and 4 of~\proref{fig:two-online}, aborting if needed. Then, set $\bx_B' \coloneqq X^{x_\inp}Z^{z_\inp}\overline{\bx}_B'$ and compute $$(\widehat{\by}_A,\overline{\by}_B) \gets Q_\dst[C_\out,x_\out,z_\out](\bx_A',\bx_B',\bz_\inp,\trap_A,\bt_\inp)$$ and return a re-randomized $\qmfhe.\Enc(\pk^{(x_\inp,z_\inp)},\widehat{\by}_A)$ to $\cA_\secp$.
    \item $\cH_3$: Compute $\bm_{B,1}$ as $C_B(\Zstate^{n_B},\Zstate^{\secp})$, and sample $x_\inp,z_\inp \gets \{0,1\}^{n_B}$ rather than computing them based on Bell measurement outcomes. Furthermore, substitute $\bx_B$ for $\bx_B'$ before applying $Q_\dst[C_\out,x_\out,z_\out]$ to the registers described above in $\cH_2$.
    \item $\cH_4$: Do not compute $\overline{\by}_B$ followed by $\by_B \coloneqq X^{x_\inp}Z^{z_\inp}\overline{\by}_B$ (in party $B$'s reconstruction). Rather, compute $$(\widehat{\by}_A,\by_B) \gets Q_\dst[C_\out,0^{m_B},0^{m_B}](\bx_A',\bx_B,\bz_\inp,\trap_A,\bt_\inp).$$ 
    \item $\cH_5$: Instead of directly computing $Q_\dst[C_\out,0^{m_B},0^{m_B}]$, query the ideal functionality with $\bx_A'$, receive $\by_A$, and send $\qmfhe.\Enc(\pk^{(x_\inp,z_\inp)},C_\out(\by_A,\trap_A))$ to $\cA_\secp$. After receiving $\left(\{r'_i\}_{i \in [2n_B]},x'_\out,z'_\out,s'\right)$ from $\cA$, carry out the checks described in $\Sim$ and send $\mathsf{ok}$ or $\abort$ to $\cI[\bx_B]$. This hybrid is $\Ideal_{\Pi,\Q,A}(\Sim,\brho_\secp,\bx_A,\bx_B,\baux)$.
\end{itemize}

\noindent We show indistinguishability between each pair of hybrids.

\begin{itemize}
    \item $\cH_0 \approx_c \cH_1$: This follows directly from the security against corrupted $A$ of $\twopc$.
    \item $\cH_1 \approx_s \cH_2$: This follows directly from the statistical correctness of the classical garbled circuit, the statistical correctness of the quantum garbled circuit, and the statistical ciphrerext re-randomization of $\qmfhe$.
    \item $\cH_2 \approx_s \cH_3$: First, by the correctness of teleportation, and by the security of the Clifford authentication code, conditioned on all measurements of qubits in $\trap_B$ returning 0, we have that $\bx_B' \approx_s \bx_B$. Next, switching $\be_{1}$ to $\Zstate^{n_B}$ in $B$'s first message is perfectly indistinguishable due to the perfect hiding of the Clifford authentication code. 
    \item $\cH_3 \approx_s \cH_4$: This follows from the statistical binding of $\Com$.
    \item $\cH_4 \approx_s \cH_5$: First, by~\cref{lemma:linear-map}, conditioned on all measurements of qubits in $\bz_\chck$ returning 0, we have that $\bz_\inp \approx_s \Zstate^{n_Z}$.
    
    Next, the above observation, along with~\cref{lemma:distillation}, implies that, conditioned on all $T$-basis measurements of qubits in $\bt_\chck$ returning 0, it holds that the output of $Q_\dst[C_\out](\bx_A',\bx_B,\bz_\inp,\trap_A,\bt_\inp)$ is statistically close to the result of computing $(\by_A,\by_B) \gets Q(\bx_A',\bx_B,\Zstate^{n_Z},\Tstate^{n_T})$ and returning $(C_\out(\by_A,\trap_A),\by_B)$. This is precisely what is being computed in $\cH_4$.

\end{itemize}

\end{proof}

\fi


Next, we show that $\Pi$ satisfies \cref{def:mpqc} for any $\cA$ corrupting party $B$.

\paragraph{The simulator.} Consider any QPT adversary $\{\cA_\secp\}_{\secp \in \bbN}$ corrupting party $B$. The simulator $\Sim$ is defined as follows.

\paragraph{$\Sim^{\cI[\bx_A](\cdot)}(\bx_B,\baux_\cA)$:}

\begin{itemize}
    \item Simulate CRS and extract from adversary's round 0 message:
    \begin{itemize}
        \item Compute $(\twopc.\crs,\twopc.\tau) \gets \twopc.\Sim_B^{(1)}(1^\secplev)$, $(\Com.\crs,\Sim.\cmt,\Com.\tau) \gets \Com.\Sim.\Gen(1^\secplev)$, and send $(\twopc.\crs,\Com.\crs)$ to  $\cA_\secp(\bx_B,\baux_\cA)$.
        \item Receive $(m_{B,1},\bm_{B,1})$ and then compute $\inp \gets \twopc.\Sim_B^{(2)}(1^\secplev,\twopc.\tau,m_1).$ If $\inp = \bot$, then abort, if not parse $\inp$ as $C_B$ and compute $(\overline{\bx}_B,\trap_B) \coloneqq C_B^\dagger(\bm_{B,1})$.
    \end{itemize}
    \item Compute quantum part of simulated round 1 message:
    \begin{itemize}
        \item Sample $C_\out \gets \mathscr{C}_{m_A + \secp}$ and compute $\widehat{\by}'_A \coloneqq C_\out(\Zstate^{m_A + \secp})$. 
        \item Prepare $m_B$ EPR pairs $\left\{\left(\be_{\Sim,1}^{(i)},\be_{\Sim,1}^{(i)}\right)\right\}_{i \in [m_B]}$, and let $$\be_{\Sim,1} \coloneqq \left(\be_{\Sim,1}^{(1)},\dots,\be_{\Sim,1}^{(m_B)}\right), \be_{\Sim,2} \coloneqq \left(\be_{\Sim,2}^{(1)},\dots,\be_{\Sim,1}^{(m_B)}\right).$$
        \item Compute $(\widetilde{\bm}_\inp,D_0,\widetilde{g}_1,\dots,\widetilde{g}_d) \gets \QGSim\left(1^\secplev,\{n_i,k_i\}_{i \in [d]}, \left(\widehat{\by}'_A,\be_{\Sim,1}\right)\right)$, where $\widetilde{\bm}_\inp$ is the simulated quantum garbled input \ifsubmission\else on registers $(\gray{A},\gray{B},\gray{Z_\inp},\gray{\Trap_A},\gray{T_\inp})$\fi, and $\{n_i,k_i\}_{i \in [d]}$ are the parameters of $\CM$ circuit $Q_\dst[\cdot,\cdot,\cdot]$.
        \item Sample $U_{\rerandenc} \gets \mathscr{C}_{n_A+n_B+n_Z+\secp+n_T\secp}$.
        \item Sample $U_{\decchckrerand} \gets \mathscr{C}_S$ and compute $$\bm_{A,2} \coloneqq U_{\decchckrerand}^\dagger(U_{\rerandenc}^\dagger(\widetilde{\bm}_\inp),\Zstate^{n_Z},\trap_B,\Tstate^{\secp}).$$
    \end{itemize}
    \item Compute classical part of simulated round 1 message:
    \begin{itemize}
        \item Compute $(\{\widetilde{\lab}_i\}_{i \in [2n_B]},\widetilde{f}_{\inpcor}) \gets \GSim(1^\secplev,1^{2n_B},1^{|f_{\inpcor}|},U_{\rerandenc})$.
        \item Sample $\{r_{i,b}\}_{i \in [2n_B], b \in \{0,1\}} \gets (\{0,1\}^\secplev)^{4n_B}$.
        \item For $i \in [2n_B], b \in \{0,1\}$, compute $(\pk_{i,b},\sk_{i,b}) \coloneqq \qmfhe.\Gen(1^\secplev;r_{i,b})$ and $\ct_{i,b} \gets \qmfhe.\CEnc(\pk_{i,b},\widetilde{\lab}_i)$.
        \item Compute 
        \[m_{A,2} \gets \twopc.\Sim_B^{(3)}\left(1^\secplev,\twopc.\tau,\left(\begin{array}{c}U_{\decchckrerand},\{\pk_{i,b},\ct_{i,b}\}_{i,b},\\\widetilde{f}_{\inpcor},D_0,\widetilde{g}_1,\dots,\widetilde{g}_d,\Sim.\cmt\end{array}\right)\right).\]
    \end{itemize}
    \item Send round 1 message and extract adversary's input:
    \begin{itemize}
        \item Send $(m_{A,2},\bm_{A,2})$ to $\cA_\secp$.
        \item Receive $(x_\inp,z_\inp)$ from $\cA_\secp$ and compute $\bx_B' \coloneqq X^{x_\inp}Z^{z_\inp}\overline{\bx}_B$.
    \end{itemize}
    \item Query ideal functionality and send simulated round 2 message:
    \begin{itemize}
        \item Forward $\bx_B'$ to $\cI[\bx_A](\cdot)$ and receive back $\by_B$.
        \item Perform Bell measurements on each pair of corresponding qubits in $(\by_B,\be_{\Sim,2})$ and let $x_\out,z_\out \in \{0,1\}^{m_B}$ be the measurement outcomes.
        \item Compute $s \gets \Com.\Sim.\Open(1^\secplev,\Com.\tau,(x_\out,z_\out))$.
        \item Send $\left(r^{(x_\inp,z_\inp)},x_\out,z_\out,s
        \right)$ to $\cA_\secp$.
    \end{itemize}
    \item Check for abort:
    \begin{itemize}
        \item Receive $\qmfhe.\Dec(\sk^{(x_\inp,z_\inp)},\widehat{\by}_A)$ from $\cA_\secp$ and use $\sk^{(x_\inp,z_\inp)}$ to decrypt the ciphertext. If decryption fails, then abort.
        \item Measure the last $\secp$ qubits of $C_\out^\dagger(\widehat{\by}_A)$ in the standard basis. If any measurement is not zero, send $\abort$ to the ideal functionality and otherwise send $\mathsf{continue}$.
        \item Output the output of $\cA_\secp$.
    \end{itemize}
\end{itemize}

\paragraph{Notation.} For any adversary $\{\cA_\secp\}_{\secp \in \bbN}$ and set of inputs $(\bx_A,\bx_B,\baux_\cA)$, we partition the distributions $\Real_{\Pi,\Q}(\cA_\secp,\bx_A,\bx_B,\baux_\cA)$ and $\Ideal_{\Pi,\Q}(\Sim,\bx_A,\bx_B,\baux_\cA)$ by the first round message $(x_\inp,z_\inp)$ sent by the adversary. That is, we define the distribution $\Real_{\Pi,\Q}^{(x_\inp,z_\inp)}(\cA_\secp,\bx_A,\bx_B,\baux_\cA)$ to be the distribution $\Real_{\Pi,\Q}(\cA_\secp,\bx_A,\bx_B,\baux_\cA)$ except that the output of the distribution is replaced with $\bot$ if the adversary did \emph{not} send $(x_\inp,z_\inp)$ in round 1. We define $\Ideal_{\Pi,\Q}^{(x_\inp,z_\inp)}(\Sim,\bx_A,\bx_B,\baux_\cA)$ analogously.


We now prove the following lemma, which is the main part of the proof of security against malicious $B$. For notational convenience, we drop the indexing of inputs and teleportation errors by $\secp$.

\begin{lemma}
\label{lemma:distribution-bot}

There exists a negligible function $\mu$ such that for any QPT $\cA = \{\cA_\secp\}_{\secp \in \bbN}$, QPT distinguisher $\cD = \{\cD_\secp\}_{\secp \in \bbN}$, inputs $(\bx_A,\bx_B,\baux_\cA,\baux_\cD)$, and teleportation errors $x_\inp,z_\inp$,

\begin{align*}
&\bigg|\Pr\left[\cD_\secp\left(\baux_\cD,\Real_{\Pi,\Q}^{(x_\inp,z_\inp)}\left(\cA_\secp,\bx_A,\bx_B,\baux_\cA\right)\right) = 1\right]\\ &- \Pr\left[\cD_\secp\left(\baux_\cD,\Ideal_{\Pi,\Q}^{(x_\inp,z_\inp)}\left(\Sim_\secp,\bx_A,\bx_B,\baux_\cA\right) \right) = 1  \right] \bigg| \leq \frac{\mu(\secp)}{2^{2n_B}}.
\end{align*}


\end{lemma}

\begin{proof}
First note that by the definition of $\secplev$, a $\cD$ violating the lemma distinguishes with probability at least $\left(\frac{1}{\poly(\secp)}\right)2^{-\secplev^{(1/c_\mathsf{lev})}} \geq \frac{1}{2^{\secplev^{\epsilon}}}.$

Now fix any collection $\cD,\cA,\bx_A,\bx_B,\baux_\cA,\baux_\cD,x_\inp,z_\inp$. We show the indistinguishability via a sequence of hybrids, where $\cH_0$ is the distribution $\Real_{\Pi,\Q}^{(x_\inp,z_\inp)}(\cA_\secp,\bx_A,\bx_B,\baux_\cA)$. In each hybrid, we describe the differences from the previous hybrid.

\begin{itemize}
    \item $\cH_1$: Simulate $\twopc$, using $\twopc.\Sim_B^{(1)}$ to sample $\twopc.\crs$, $\twopc.\Sim_B^{(2)}$ to extract the adversary's input $C_B$, and $\twopc.\Sim_B^{(3)}$ to sample party $A$'s message $m_{A,2}$. Use $C_B$ and freshly sampled $C_A,C_\out,\{r_{i,b}\}_{i,b},x_\out,z_\out,s$ to sample the output of the classical functionality that is given to $\twopc.\Sim_B^{(3)}$.
    \item $\cH_2$: Simulate $\Com$, using $\Com.\Sim.\Gen$ to sample $\Com.\crs$ and the commitment $\Sim.\cmt$. Note that $\Sim.\cmt$ is now used directly in computing the output of $\twopc$, and $s$ is no longer sampled by party $A$. Open the commitment in the second round to $(x_\out,z_\out)$ using $\Com.\Sim.\Open$.
    \item $\cH_3$: In this hybrid, we make a (perfectly indistinguishable) switch in how $\bm_{A,2}$ is computed and how $U_{\decchckrerand}$ (part of the $\twopc$ output) is sampled. Define $(\overline{\bx}'_B,\trap_B) \coloneqq C_B^\dagger(\bm_{B,1})$, where $C_B$ was extracted from $m_{B,1}$. Note that in $\cH_2$, by the definitions of $\cF[Q]$ and $\cG[Q]$,
    
    $$U_{\decchckrerand}(\bm_{A,2}) \coloneqq (U_{\rerand}(\bx_A,\overline{\bx}'_B,\Zstate^{n_Z},\Tstate^{n_T\secp}),\Zstate^{n_Z},\trap_B,\Tstate^{\secp}).$$ 
    
    Moreover, there exists a Clifford unitary $U$ such that $U_{\decchckrerand} = UC_A^\dagger$, where $C_A$ was sampled uniformly at random from $\mathscr{C}_s$. Thus, since the Clifford matrices form a group, an equivalent sampling procedure would be to sample $U_{\decchckrerand} \gets \mathscr{C}_s$ and define $$\bm_{A,2} \coloneqq U_{\decchckrerand}^\dagger(U_{\rerand}(\bx_A,\overline{\bx}_B',\Zstate^{n_Z+\secp},\Tstate^{n_T\secp}),\Zstate^{n_Z},\trap_B,\Tstate^\secp).$$ This is how $\cH_3$ is defined.
    \item $\cH_4^{(1)},\dots,\cH_4^{(2n_B)}$: In $\cH_4^{(i)}$, let $\ct_{i,1-(x_\inp,z_\inp)_i} \gets \qmfhe.\CEnc(\pk_{i,1-(x_\inp,z_\inp)_i},0)$.
    \item $\cH_5$: Simulate the classical garbled circuit. In particular, let $$U_{\rerandenc} \coloneqq E_0\left(\bbI^{n_A} \otimes X^{x_\inp}Z^{z_\inp} \otimes \bbI^{n_Z + \secp + n_T\secp}\right)U_\rerand^\dagger,$$ and compute $(\{\widetilde{\lab}_i\}_{i \in [2n_B]},\widetilde{f}_{\inpcor}) \gets \GSim(1^\secplev,1^{2n_B},1^{|f_{\inpcor}|},U_{\rerandenc})$. Now, each $\ct_{i,(x_\inp,z_\inp)_i}$ be will an encryption of $\widetilde{\lab}_i$.
    \item $\cH_6^{(1)},\dots,\cH_6^{(2n_B)}$: In $\cH_6^{(i)}$, let $\ct_{i,1-(x_\inp,z_\inp)_i} \gets \qmfhe.\CEnc(\pk_{i,1-(x_\inp,z_\inp)_i},\widetilde{\lab}_i)$.
    \item $\cH_7$: In this hybrid, we make another perfectly indistinguishable switch in how $\bm_{A,2}$ is computed. Let $\bx_B' \coloneqq X^{x_\inp}Z^{z_\inp}\overline{\bx}'_B$, and compute $U_{\rerandenc} \coloneqq E_0U_{\rerand}^\dagger$ and $$\bm_{A,2} \coloneqq U_{\decchckrerand}^\dagger(U_{\rerand}(\bx_A,\bx_B',\Zstate^{n_Z+\secp},\Tstate^{n_T\secp}),\Zstate^{n_Z},\trap_B,\Tstate^{\secp}).$$
    \item $\cH_8$: Simulate the quantum garbled circuit. In particular, compute $$(\widehat{\by}_A,\overline{\by}_B) \gets Q_\dst[C_\out,x_\out,z_\out](\bx_A,\bx_B',\Zstate^{n_Z},\Tstate^{n_T\secp}),$$ followed by $$(\widetilde{\bm}_\inp,D_0,\widetilde{g}_1,\dots,\widetilde{g}_d) \gets \QGSim(1^{\secplev},\{n_i,k_i\}_{i \in [d]},(\widehat{\by}_A,\overline{\by}_B)),$$ where $\{n_i,k_i\}_{i \in [d]}$ are the parameters of the $\CM$ circuit $Q_\dst[C_\out,x_\out,z_\out]$. 
    
    Sample $U_{\rerandenc} \gets \mathscr{C}_{n_A+n_B+n_Z+\secp+n_T\secp}$ and compute $$\bm_{A,2} \coloneqq U_{\decchckrerand}^\dagger(U_{\rerandenc}^\dagger(\bx_A,\bx'_B,\Zstate^{n_Z+\secp},\Tstate^{n_T\secp}),\Zstate^{n_Z},\trap_B,\Tstate^{\secp}).$$ 
    
    \item $\cH_{10}$: Note that $Q_\dst[C_\out,x_\out,z_\out](\bx_A,\bx_B',\Zstate^{n_Z+\secp},\Tstate^{n_T\secp})$ may be computed in two stages, where the first outputs $(\by_A,\by_B,\Zstate^\secp,C_\out,x_\out,z_\out)$ and the second outputs $(\widehat{\by}_A,\overline{\by}_B) \coloneqq (C_\out(\by_A,\Zstate^\secp),X^{x_\out}Z^{z_\out}\by_B)$. In this hybrid, we make the following perfectly indistinguishable switch to the second part of this computation. Prepare $m_B$ EPR pairs $\left\{\left(\be_{\Sim,1}^{(i)},\be_{\Sim,1}^{(i)}\right)\right\}_{i \in [m_B]}$, and let $\be_{\Sim,1} \coloneqq \left(\be_{\Sim,1}^{(1)},\dots,\be_{\Sim,1}^{(m_B)}\right)$ and $\be_{\Sim,2} \coloneqq \left(\be_{\Sim,2}^{(1)},\dots,\be_{\Sim,1}^{(m_B)}\right)$. Then set $(\widehat{\by}_A,\overline{\by}_B) = (C_\out(\by_A,\Zstate^\secp),\be_{\Sim,1})$ and let $x_\out,z_\out$ be the result of Bell measurements applied to corresponding pairs of qubits of $(\by_B,\be_{\Sim,2})$. Note that these Bell measurements do not have to be performed until the simulator sends its simulated round 2 message.
    \item $\cH_{11}$: After computing the first stage of $Q_\dst[C_\out,x_\out,z_\out](\bx_A,\bx_B',\Zstate^{n_Z+\secp},\Tstate^{n_T\secp})$, set $\by_A$ aside and re-define the final output to be $(\widehat{\by}'_A,\overline{\by}_B) = (C_\out(\Zstate^{m_A+\secp}),\be_{\Sim,1})$. Now, during $A$'s output reconstruction step, if the check (step 3) passes, output $\by_A$, and otherwise abort.
    \item $\cH_{12}$: Rather than directly computing $\by_A$ from the first stage of $Q_\dst[C_\out,x_\out,z_\out](\bx_A,\bx_B',\Zstate^{n_Z+\secp},\Tstate^{n_T\secp})$, forward $\bx_B'$ to $\cI[\bx_A](\cdot)$ and receive back $\by_B$, which gives the same distribution as $\cH_{11}$. Now, during $A$'s reconstruction step, if the check passes, send $\mathsf{ok}$ to the ideal functionality, and otherwise send $\abort$. This is $\Ideal_{\Pi,\Q}^{(x_\inp,z_\inp)}(\Sim,\bx_A,\bx_B,\baux_\cA)$.
\end{itemize}

\end{proof}

\begin{theorem}
Let $\Pi$ be the protocol described in~\proref{fig:two-online} computing some quantum circuit $Q$. Then $\Pi$ satisfies \cref{def:mpqc} for any $\cA$ corrupting party $B$.

\end{theorem}

\begin{proof}

Assume towards contradiction the existence of a QPT $\cD = \{\cD_\secp\}_{\secp \in \bbN}$, a QPT $\cA = \{\cA_\secp\}_{\secp \in \bbN}$, and $(\bx_A,\bx_B,\baux_\cA,\baux_\cD)$ such that 
\begin{align*}
&\bigg|\Pr\left[\cD_\secp\left(\baux_\cD,\Real_{\Pi,\Q}\left(\cA_\secp,\bx_A,\bx_B,\baux_\cA\right)\right) = 1\right]\\ &- \Pr\left[\cD_\secp\left(\baux_\cD,\Ideal_{\Pi,\Q}\left(\Sim_\secp,\bx_A,\bx_B,\baux_\cA\right)\right) = 1 \right] \bigg| \geq 1/\poly(\secp).
\end{align*}

Define $\Real \coloneqq \Real_{\Pi,\Q}(\cA_\secp,\bx_A,\bx_B,\baux_\cA)$ and $\Ideal \coloneqq \Ideal_{\Pi,\Q}(\Sim,\bx_A,\bx_B,\baux_\cA)$. Furthermore, let $\mathbf{E}^{(x_\inp,z_\inp)}_\Real$ be the event that $\cA$ sends $(x_\inp,z_\inp)$ as its first round message in $\Real$ and define $\mathbf{E}^{(x_\inp,z_\inp)}_\Ideal$, $\mathbf{E}^{(\abort)}_\Real$. Let $\mathbf{E}^{(\abort)}_\Real$ and $\mathbf{E}^{(\abort)}_\Ideal$ be the event that the adversary fails to report some of its teleporation errors, causing the honest party to abort. The above implies that either there exists some $(x_\inp,z_\inp) \in (\{0,1\}^{n_B})^2$ such that 
\begin{align*}
&\bigg|\Pr\left[\cD_\secp(\baux_\cD,\Real) = 1 \big| \mathsf{E}_\Real^{(x_\inp,z_\inp)}\right]\Pr\left[\mathsf{E}_\Real^{(x_\inp,z_\inp)}\right] \\&- \Pr\left[\cD_\secp(\baux_\cD,\Ideal) = 1 \big| \mathsf{E}_\Ideal^{(x_\inp,z_\inp)} \right]\Pr\left[\mathsf{E}_\Ideal^{(x_\inp,z_\inp)}\right] \bigg| \geq \frac{1}{\poly(\secp)(2^{2n_B}+1)}
\end{align*}

or that 

\begin{align*}
&\bigg|\Pr\left[\cD_\secp(\baux_\cD,\Real) = 1 \big| \mathsf{E}_\Real^{(\abort)}\right]\Pr\left[\mathsf{E}_\Real^{(\abort)}\right] \\&- \Pr\left[\cD_\secp(\baux_\cD,\Ideal) = 1 \big| \mathsf{E}_\Ideal^{(\abort)} \right]\Pr\left[\mathsf{E}_\Ideal^{(\abort)}\right] \bigg| \geq \frac{1}{\poly(\secp)(2^{2n_B}+1)}.
\end{align*}

Simulating the distribution conditioned on an abort is trivial, so the second case cannot occur, and the first case immediately contradicts~\cref{lemma:distribution-bot}, completing the proof.
\end{proof}

\section{Multi-Party Quantum Computation in Five Rounds}\label{sec:five-round}

In this section, we show the existence of a five-round protocol for multi-party quantum computation, assuming quantum-secure two-message oblivious transfer in the CRS model with a straight-line black-box simulator. The protocol we present satisfies security with abort, and only requires three rounds of online communication (that is, three rounds of communication once the parties receive their inputs). Thus, this implies the existence of a three-round protocol for multi-party quantum computation given some input-independent quantum pre-processing. 

We also note that the protocol can be adjusted to give security with \emph{unanimous} abort with four rounds of online communication (while keeping the total number of rounds at five), though we do not provide a formal description of this protocol. Roughly, this follows because if parties receive their inputs one round earlier, they will be able to receive and check the authenticity of their (encrypted) outputs at the end of round four, rather than checking the authenticity of their (unencrypted) outputs at the end of round five.

\subsection{The Protocol}\label{subsec:5-round-protocol}

\paragraph{Ingredients.}

\begin{itemize}
    \item Round-optimal quantum-secure multi-party computation for classical reactive functionalities in the CRS model, to be treated as an oracle called $\MPC$ ~(see \cref{subsec:reactivempc}).
    \item A garbling scheme for $\CM$ circuits $(\QGarble, \QGEval, \QGSim)$.
\end{itemize}

\paragraph{Notation.} We use the following parameters and notation throughout:
\begin{itemize}
\item Let $n$ be the number of parties.
\item Let $Q$ be a $\CM$ circuit with $m = m_1 + \dots + m_n$ input qubits and $\ell = \ell_1 + \dots + \ell_n$ output qubits. Let $Q_\dst\left[\left\{C_i^{\inp},C_i^{\out}\right\}_{i \in [n]}\right]$ be the $\CM$ circuit that 
\begin{itemize}
    \item first applies $T$-state distillation (taking as input $\secp$ times as many $T$ states as $Q$),
    \item then Clifford decodes each input using the $C_i^{\inp}$, outputting $\bot$ if any of the decodings fail,
    \item then applies $Q$,
    \item then Clifford encodes each part of the output using $C_i^{\out}$.
\end{itemize}
Let $k_0$ be the total number of 0 states necessary to garble $Q_\dst\left[\left\{C_i^{\inp},C_i^{\out}\right\}_{i \in [n]}\right]$ (which includes auxiliary 0 states for the computation itself, as well as extra 0states used for the garbling operation). Let $k_T$ be the total number of $T$ states that $Q_\dst[\{C_i^{\inp},C_i^{\out}\}_{i \in [n]}]$ takes as input.
\item Let $v = (m + \secp n) + 2(k_0 + \secp n) + (k_T + \secp n)$ be the number of registers teleported around the circle of parties, which includes all Clifford-encoded inputs, 0 states, and T states. We will refer to the first $m + \secp n$ registers as $\gray{N}$, the next $2(k_0 + \secp n)$ registers as $\gray{Z}$, and the final $k_T + \secp n$ registers as $\gray{T}$.
\item During the protocol, these $v$ registers will be manipulated. At one point (before the application of the quantum garbled circuit), we will rename the registers to $\gray{N},\gray{Z_\inp},\gray{T_\inp},\gray{Z_\test},\gray{T_{\test,1}},\dots,\gray{T_{\test,n}}$, where 
\begin{itemize}
    \item $\gray{N}$ has size $m + \secp n$ and holds each party's Clifford-encoded input.
    \item $\gray{Z_\inp}$ has size $k_0$ and holds the auxiliary 0 states that will be used in the computation of the quantum garbled circuit.
    \item $\gray{T_\inp}$ has size $k_T$ and holds the auxiliary T states that will be used in the computation of the quantum garbled circuit.
    \item $\gray{Z_\test}$ has size $k_0 + \secp n$ and holds the result of the 0-state check (will be $\mathbf{0}^{k_0 + \secp}$ if all parties are honest).
    \item Each $\gray{T_{\test,1}},\dots,\gray{T_{\test,n}}$ has size $2\secp$ and holds the result of the T-state check (each will be a Clifford-encoding of $\mathbf{T}^{\otimes \secp}$ if all parties are honest).
\end{itemize}

\end{itemize}

\protocol
{\proref{fig:classical}: Classical Functionality for Five-Round Quantum MPC}
{Classical Functionality for Five-Round Quantum MPC.}
{fig:classical}
{
\textbf{Public Parameters:} Security parameter $\secp$, number of parties $n$, $\CM$ circuit $Q$, and parameters $(m,\ell,k_0,k_T,v)$ defined above.\\

\textbf{Shared Randomness:} Random strings for $0$ state check $r,s \gets \{0,1\}^{k_0 + \secp n}$, re-randomization matrix $U_{\mathsf{enc}} \gets \mathscr{C}_{m + \secp n + k_0 + k_T}$, Cliffords for $T$-state checks $\{C_i^{T} \gets \mathscr{C}_{2\secp}\}_{i \in [n]}$, Cliffords for outputs $\{C_i^{\out} \gets \mathscr{C}_{\ell_i + \secp}\}_{i \in [n]}$.\\

\textbf{Offline Round 1:}
   Obtain input $C_i^\Circle$ from each party $i$.

\textbf{Offline Round 2:} 
   Obtain input $(x_i^{\Circle},z_i^{\Circle})$ from each party $i$.

\textbf{Online Round 1:} Obtain input $(x_i^{\inp},z_i^{\inp},C_i^{\inp})$ from each party $i$.
\begin{itemize}
    \item Compute the unitary $U_{\mathsf{dec}} \coloneqq {C_1^\Circle}^\dagger X^{x_1^{\Circle}}Z^{z_1^{\Circle}} \dots {C_n^\Circle}^\dagger X^{x_n^{\Circle}}Z^{z_n^{\Circle}}$, which will operate on registers $\gray{N},\gray{Z},\gray{T}$, where $\gray{N}$ has size $m+\secp n$, $\gray{Z}$ has size $2(k_0 + \secp n)$, and $\gray{T}$ has size $(k_T + \secp n)$.
    \item Compute the unitary $U_{\mathsf{check}}$ that operates on registers $\gray{Z},\gray{T}$ as follows.
    \begin{itemize}
        \item Sample a random linear map $M \gets \mathsf{GL}(2(k_0 + \secp n),\bbF_2)$, and apply it to the registers $\gray{Z}$. Now refer to the first $k_0$ qubits of $\gray{Z}$ as register $\gray{Z_{\inp}}$, the following $n$ groups of $\secp$ qubits as $\gray{T_{\test,Z,1}},\dots,\gray{T_{\test,Z,n}}$, and the final group of $k_0 + \secp n$ qubits as $\gray{Z_{\mathsf{test}}}$. 
        \item Sample a random permutation $\pi$ on $k_T + \secp n$ elements and rearrange the registers of $\gray{T}$ according to the permutation $\pi$. Now refer to the first $k_T$ qubits of $\gray{T}$ as register $\gray{T_{\inp}}$ and the following $n$ groups of $\secp$ qubits as $\gray{T_{\mathsf{test},T,1}},\dots,\gray{T_{\mathsf{test},T,n}}$.
        \item Rearrange the registers in the order $\gray{N},\gray{Z_\inp}, \gray{T_\inp}, \gray{Z_{\mathsf{test}}}, \gray{T_{\mathsf{test},T,1}}, \gray{T_{\test,Z,1}},\dots, \gray{T_{\mathsf{test},T,n}},\gray{T_{\test,Z,n}}.$
        \item Apply $X^rZ^s$ to $\gray{Z_{\mathsf{test}}}$.
        \item For each $i \in [n]$, apply $C_i^T$ to $(\gray{T_{\mathsf{test},T,i}},\gray{T_{\test,Z,i}})$ and re-name the combined registers $\gray{T_{\test,i}}$.
    \end{itemize}
    \item Output to party 1 the unitary $U_{\mathsf{test}} \coloneqq \left(U_{\mathsf{enc}}^{\gray{N},\gray{Z_\inp},\gray{T_\inp}} \otimes \bbI^{\gray{Z_\test},\gray{T_{\test,1},\dots,\gray{T_{\test,n}}}}\right)\left(\bbI \otimes U_{\mathsf{check}}^{\gray{Z},\gray{T}}\right)U_{\mathsf{dec}}^{\gray{N},\gray{Z},\gray{T}}.$
   
\end{itemize}

\textbf{Online Round 2:} Obtain input $r'$ from party 1.
\begin{itemize}
    \item Compute $(E_0,D_0,\widetilde{g}_1,\dots,\widetilde{g}_d) \gets \QGarble(1^\secp,Q_\dst[\{C_i^\inp,C_i^\out\}_{i \in [n]}])$. 
    \item Compute $U_{\mathsf{garble}} \coloneqq E_0\left(X^{x_1^\inp}Z^{z_1^\inp} \otimes \dots \otimes X^{x_n^\inp}Z^{z_n^\inp} \otimes \bbI^{\gray{Z_\inp},\gray{T_\inp}}\right)U_{\mathsf{enc}}^\dagger$.
    \item Output $U_{\mathsf{garble}}, D_0,\widetilde{g}_1,\dots,\widetilde{g}_d$ to party 1 and for each $i \in [n]$ output $C_i^T$ to party $i$.
\end{itemize}

\textbf{Online Round 3:}
\begin{itemize}
    \item If $r' = r$ then output $C_i^{\out}$ to party $i$ for each $i \in [n]$ and otherwise output $\bot$ to each party.
\end{itemize}
}


\paragraph{Offline Round 1.}

In the first offline round of communication, the parties send EPR pair halves to each other as follows.

\begin{itemize}
    \item \underline{Party $P_1$.} For each $i \in [n]$, party $P_1$ generates $m_i+ \secp$ EPR pairs
    \[\left(\be^{(1\leftrightarrow i)}_R,\be^{(1\leftrightarrow i)}_S\right)^{\otimes m_i + \secp},\]
    where the $(1\leftrightarrow i)$ superscript indicates that these EPR pairs will be shared between party $1$ and party $i$, and the $R$ and $S$ subscripts designate which halves of the EPR pairs will be for receiving teleported qubits and which halves will be for sending teleported qubits. It also generates $v$ EPR pairs
    \[\left(\be^{(n\leftrightarrow 1)}_R,\be^{(n\leftrightarrow 1)}_S\right)^{\otimes v}.\]
    
    For each $i \in [n] \setminus \{1\}$, party $P_1$ sends $\left(\be_S^{(1\leftrightarrow i)}\right)^{\otimes m_i + \secp}$ to party $P_i$. $P_1$ also sends $\left(\be^{(n\leftrightarrow 1)}_S\right)^{\otimes v}$ to party $P_n$. Finally, $P_1$ samples $C_1^\Circle \gets \mathscr{C}_v$ and sends input $C_1^\Circle$ to $\MPC$.
    
    \item \underline{Party $P_i$ for $i \in \{2,3,\dots,n\}$.} Every other party $P_i$ will generate $v$ EPR pairs
    \[\left(\be^{((i-1)\leftrightarrow i)}_R,\be^{((i-1)\leftrightarrow i)}_S\right)^{\otimes v}\]
    and send
    \[\left(\be^{((i-1)\leftrightarrow i)}_S\right)^{\otimes v}\]
    to party $P_{i-1}$. In addition, $P_i$ will sample $C_i^\Circle \gets \mathscr{C}_v$ and send input $C_i^\Circle$ to $\MPC$.
\end{itemize}

\paragraph{Offline Round 2.} In the second offline round, the parties perform local operations and then query the classical $\MPC$. 

\begin{itemize}
    \item \underline{Party $P_1$.} Party $P_1$ prepares the states $\mathbf{0}^{2(k_0 + \secp n)}$ and $\mathbf{T}^{k_T + \secp n}$. It generates the $v$ qubit quantum state
    \[ C_1^{\Circle}
    \left(
     \left( \be_R^{(1\leftrightarrow 1)} \right)^{\otimes m_1 + \lambda}, 
     \left(\be_R^{(1\leftrightarrow 2)} \right)^{\otimes m_2 + \lambda},\ldots,\left(\be_R^{(1 \leftrightarrow n)} \right)^{\otimes m_n + \lambda},\mathbf{0}^{2(k_0 + \secp n)},\mathbf{T}^{k_T + \secp n}\right),\] 
    and begins the process of teleporting this state to party $P_2$. That is, for each $j \in [v]$, it measures positions $(j,v+j)$ of 
    \[ \left(C_1^{\Circle}\left(\left( \be_R^{(1\leftrightarrow 1)} \right)^{\otimes m_1 + \lambda}, 
     \left(\be_R^{(1\leftrightarrow 2)} \right)^{\otimes m_2 + \lambda},\ldots,\left(\be_R^{(1 \leftrightarrow n)} \right)^{\otimes m_n + \lambda},\mathbf{0}^{2(k_0 + \secp n)},\mathbf{T}^{k_T + \secp n}\right),\left(\be^{(1\leftrightarrow 2)}_S\right)^{\otimes v}\right)\]
    in the Bell basis. The result of these measurements is two classical $v$-length bitstrings $x_1^{\Circle},z_1^{\Circle}$.
    
    
    Finally, it sends input $(x_1^{\Circle},z_1^{\Circle})$ to $\MPC$.

    
    \item \underline{Party $P_i$ for $i \in \{2,3,\dots,n\}$.} During this round, party $P_i$ will apply its own masking Clifford $C_i^{\Circle}$ to a state that has already been masked by parties $P_1,P_2,\dots,P_{i-1}$, and then teleport the resulting state along to $P_{i+1}$. Note that this is all happening \emph{simultaneously}, and party $P_i$ does not ``wait'' until it has received the state from party $P_{i-1}$. Precisely, party $P_i$ applies $C_i^{\Circle}$ to the state $\left(\be^{(i-1) \leftrightarrow i}_R\right)^{\otimes v}$, and then for each $j \in [v]$, it measures positions $(j,j+v)$ of 
    \[ \left(C_i^{\Circle}\left(\left(\be^{(i-1) \leftrightarrow i}_R\right)^{\otimes v}\right),\left(\be^{(i\leftrightarrow (i+1))}_S\right)^{\otimes v}\right)\]
    in the Bell basis. The result of these measurements is two classical $v$-length bitstrings $x_i^{\Circle},z_i^{\Circle}$.
    
    
    Finally, it sends input $(x_i^{\Circle},z_i^{\Circle})$ to $\MPC$.
    
\end{itemize}

\paragraph{Parties Receive Inputs.} After the offline rounds, each party $P_i$ receives a $m_i$-qubit input $\bx_i$. 

\paragraph{Online Round 1.} 


\begin{itemize}
\item \underline{Party $P_i$ for $i \in [n]$.} Each party $P_i$ samples a random Clifford $C_i^{\inp} \leftarrow \mathscr{C}_{m_i+\secp}$ and applies it to their input $\bx_i$ along with $\secp$-many $\mathbf{0}$ qubits. They then teleport this state into registers held by party $P_1$. That is, for each $j \in [m_i + \secp]$, it measures positions $(j,j+ m_i + \secp)$ of
\[ \left( C_i^{\inp}\left(\bx_i, \mathbf{0}^{\secp}\right),\left(\be_S^{(1 \leftrightarrow i)}\right)\right)\]
in the Bell basis. The result is two classical $(m_i+\secp)$-length bitstrings $x_i^{\inp},z_i^{\inp}$.

Each party $P_i$ inputs $(x_i^{\inp},z_i^{\inp},C_i^{\inp})$ to $\MPC$. $P_1$ receives an output $U_\mathsf{test}$ from $\MPC$.
\end{itemize}


\paragraph{Online Round 2.}
\begin{itemize}
\item \underline{Party $P_1$.} Party $P_1$ first computes \[\left(\by^{\gray{N},\gray{Z_\inp},\gray{T_\inp}},\by_Z^{\gray{Z_\test}},\by_{T,1}^{\gray{T_{\test,1}}},\dots,\by_{T,n}^{\gray{T_{\test,n}}}\right) \coloneqq U_{\mathsf{test}}\left(\left(\be^{(n\leftrightarrow 1)}_R\right)^{\otimes v}\right).\]
For each $i \in \{2,\dots,n\}$, it sends $\by_{T,i}$ to party $P_i$. Then, measure $\by_Z$ in the computational basis to obtain a classical string $r'$ of length $k_0 + \secp n$, and input $r'$ to $\MPC$.

$P_1$ receives output $(U_{\mathsf{garble}},D_0,\tilde{g}_1,\dots,\tilde{g}_d)$ from $\MPC$.
\item \underline{Party $P_i$ for $i \in [n]$.} Every party receives output $C_i^T$ from $\MPC$.
\end{itemize}

\paragraph{Online Round 3.}
\begin{itemize}
\item \underline{Party $P_1$.} Compute garbled input $\by_\inp = U_{\mathsf{garble}}(\by)$ and run the quantum garbled circuit $(D_0,\tilde{g}_1,\dots,\tilde{g}_d)$ on the resulting $(m+n\secp+k_0+k_T)$-qubit state $\by_\inp$. The result of running the quantum garbled circuit is an $\ell + n\secp$-state. Party $P_1$ partitions this state into $n$ different encrypted output states $\by_{\out,1},\dots,\by_{\out,n}$ where each $\by_{\out,i}$ is an $(\ell_i + \secp)$-qubit state. For each $i \in \{2,\dots,n\}$, party $P_1$ sends $\by_{\out,i}$ to party $P_i$. 
\item \underline{Party $P_i$ for $i \in [n]$.} After the conclusion of Online Round $2$, every  party $P_i$ for $i \in \{1,2,\dots,n\}$ has a ``$T$-check Clifford'' $C_i^{T}$ from $\MPC$, as well as a $2\secp$-qubit state $\by_{T,i}$ from party $P_1$. It computes ${C_i^{T}}^\dagger\left(\by_{T,i}\right)$, and then performs the binary projection that projects the first $\secp$ qubits onto $\mathbf{T}^{\secp}$ and the last $\secp$ qubits onto $\mathbf{0}^{\secp}$. If this projection fails, party $P_i$ asks the MPC to abort. Otherwise, receive output $C_i^\out$ from $\MPC$.

\end{itemize}

\paragraph{Output Reconstruction.}
\begin{itemize}
\item \underline{Party $P_i$ for $i \in [n]$.} After the conclusion of Online Round 3, each party $P_i$ has obtained an output decryption Clifford $C_i^{\out}$ from the classical $\MPC$, along with an $(\ell_i + \secp)$-qubit state $\by_{\out,i}$. It computes ${C_i^{\out}}^\dagger(\by_{\out,i})$ and measures whether the last $\secp$ trap qubits are all $0$. If not, it outputs $\bot$. Otherwise, its output is the remaining $\ell_i$ qubits. 
\end{itemize}

\subsection{Security}

\begin{theorem}
Assuming post-quantum maliciously-secure two-message oblivious transfer, there exists five-round maliciously-secure multi-party quantum computation.
\end{theorem}

We split the security proof into two cases based on the set of honest parties $\cH \subset [n]$. In the first case, $P_1$ is corrupted and in the second case, the only honest party is $P_1$.

\subsubsection{Case 1: $P_1$ is corrupted}

\paragraph{Simulator.}

Let $\cH \subset [n]$ be the set of honest parties and $k \neq 1$ be arbitrary such that $k \in \cH$. Let $\cM = [n] \setminus \cH$ be the set of corrupted parties. The simulator will act as party $k$, altering its actions as described below. All actions of parties $i \in \cH \setminus \{k\}$ will be honest except for those explicitly mentioned in the simulation (essentially just switching out their inputs for 0). Also note that the simulator will be implementing the ideal functionality for the classical MPC, so we allow the simulator to intercept the adversary's inputs to $\MPC$ and compute the outputs.




\begin{itemize}
    \item {\bf Offline Round 1.} 
    \begin{itemize}
        \item Following honest party $k$'s behavior, prepare $v$ EPR pairs  \[\left(\be^{((k-1)\leftrightarrow k)}_R,\be^{((k-1)\leftrightarrow k)}_S\right)^{\otimes v}\] and send
    \[\be_k \coloneqq \left(\be^{((k-1)\leftrightarrow k)}_S\right)^{\otimes v}\]
    to party $P_{k-1}$. 
    \item Receive state $\be_{k+1}$ of $v$ qubits from $P_{k+1}$.
    \item For each $i \in \cH$, receive a state $\be_{\inp,i}$ of $m_i + \secp$ qubits from $P_1$.
    \item  Obtain $\{C_i^\Circle\}_{i \in \cM}$ from the adversary's query to $\MPC$, and let $\{C_i^\Circle\}_{i \in \cH}$ be the values sampled by honest parties.
    \end{itemize}

    \item {\bf Offline Round 2.}
    \begin{itemize}
    \item 
    Sample $r,s \gets \{0,1\}^{k_0+\secp n}$, $\{C_i^T \gets \mathscr{C}_{2\secp}\}_{i \in [n]}$, $M \gets \mathsf{GL}(2(k_0 + \secp n),\bbF_2)$, and a permutation $\pi$ on $k_T + \secp n$ elements. Use these values to compute $U_\chck$ as in the computation in Online Round 1 of \proref{fig:classical}.
    
    \item Compute $\left(\bn^{\gray{N}},\bz^{\gray{Z_\inp}},\bt^{\gray{T_\inp}},\bz_\test^{\gray{Z_\test}},\bt_{\test,1}^{\gray{T_{\test,1}}},\dots,\bt_{\test,n}^{\gray{T_{\test,n}}}\right) \coloneqq (\bbI \ \otimes \ U_\chck){C_1^\Circle}^\dagger \cdots {C_{k-1}^\Circle}^\dagger(\be_k)$, and discard $\bz^{\gray{Z_\inp}},\bt^{\gray{T_\inp}}$.
    
    \item Sample $U_\test \gets \mathscr{C}_v$ and $U_{\mathsf{garble}} \gets \mathscr{C}_{m + n \secp + k_0 + k_T}$.
    
    \item Prepare $m + \secp n + k_0 + k_T$ EPR pairs $\left\{\left(\be^{(i)}_{\Sim,1},\be^{(i)}_{\Sim,2}\right)\right\}_{i \in [m + \secp n + k_0 + k_T]}$, and let $$\be_{\Sim,1} \coloneqq \left(\be_{\Sim,1}^{(1)},\dots,\be_{\Sim,1}^{(m + \secp n + k_0 + k_T)}\right), \be_{\Sim,2} \coloneqq \left(\be_{\Sim,2}^{(1)},\dots,\be_{\Sim,2}^{(m + \secp n + k_0 + k_T)}\right).$$ 
    
    
    \item Compute $$\bw \coloneqq {C_{k+1}^\Circle}^\dagger(\cdots {C_{n}^\Circle}^\dagger(U_\test^\dagger(U_{\mathsf{garble}}^\dagger(\be_{\Sim,2}),\bz_\test,\bt_{\test,1},\dots,\bt_{\test,n}))\cdots).$$
    
    \item Teleport $\bw$ into $\be_{k+1}$, and let $x_k^\Circle,z_k^\Circle$ be the teleportation errors. Let $\{x_i^\Circle,z_i^\Circle\}_{i \in \cH \setminus \{k\}}$ be the teleportation errors obtained by the other honest parties.
    
    \item Send $\{x_i^\Circle,z_i^\Circle\}_{i \in \cH}$ to the adversary.
    
    \item Receive $\{x_i^\Circle,z_i^\Circle\}_{i \in \cM}$ from the adversary.
    
    
    
    \end{itemize}

    \item {\bf Online Round 1.}
    \begin{itemize}

        \item Let $\widehat{x}_1, \widehat{z}_1$ be such that $$(\bbI \ \otimes \ U_\chck){C_1^\Circle}^\dagger X^{x_1^\Circle}Z^{z_1^\Circle} \cdots {C_{k-1}^\Circle}^\dagger X^{x_{k-1}^\Circle}Z^{z_{k-1}^\Circle} = X^{\widehat{x}_1}Z^{\widehat{z}_1}(\bbI \ \otimes \ U_\chck){C_1^\Circle}^\dagger \cdots {C_{k-1}^\Circle}^\dagger.$$ Write $\widehat{x}_1$ as $\widehat{x}_{\inp,1},\dots,\widehat{x}_{\inp,n},\widehat{x}_Z,\widehat{x}_T,\widehat{x}_{\test,Z},\widehat{x}_{\test,T,1},\dots,\widehat{x}_{\test,T,n}$, and same for $\widehat{z}_1$. Here, each $\widehat{x}_{\inp,i} \in \{0,1\}^{m_i + \secp}$, $\widehat{x}_Z \in \{0,1\}^{k_0}$, $\widehat{x}_T \in \{0,1\}^{k_T}$, $\widehat{x}_{\test,Z} \in \{0,1\}^{k_0 + \secp n}$, and each $\widehat{x}_{\test,T,i} \in \{0,1\}^{2\secp}$.

        \item Let  $\widehat{x}_2, \widehat{z}_2$ be such that
        $$X^{x_{k}^\Circle}Z^{z_{k}^\Circle}{C_{k+1}^\Circle}^\dagger X^{x_{k+1}^\Circle}Z^{z_{k+1}^\Circle}\dots {C_{n}^\Circle}^\dagger X^{x_{n}^\Circle}Z^{z_{n}^\Circle}  = X^{\widehat{x}_2}Z^{\widehat{z}_2}{C_{k+1}^\Circle}^\dagger\dots {C_{n}^\Circle}^\dagger.$$

        \item Let $\widehat{U}_\test \coloneqq U_\test X^{\widehat{x}_2}Z^{\widehat{z}_2}$ and use $\widehat{U}_\test$ in place of $U_\test$.
        

        \item For each $i \in [n]$, let $\widehat{C}_i^T \coloneqq C_i^T X^{\widehat{x}_{\test,T,i}}Z^{\widehat{z}_{\test,T,i}}$, and use $\widehat{C}_i^T$ in place of $C_i^T$ (in Online Round 2).

        \item Let $\widehat{r}:= r \oplus \widehat{x}_{\test,z}$ and use $\widehat{r}$ in place of $r$ (in Online Round 3).

        \item Sample $C_i^\inp \leftarrow  \mathscr{C}_{m_i+\secp}$.
        For $i \in \cH$, teleport $C_i^\inp(\mathbf{0}^{m_i + \secp})$ into $\be_{\inp,i}$. Let $x_i^\inp,z_i^\inp$ be the teleportation errors.
        \item Send $\{C_i^\inp,x_i^\inp,z_i^\inp\}_{i \in \cH}$ and $\widehat{U}_\test$ to the adversary. 
        \item Receive $\{C_i^\inp,x_i^\inp,z_i^\inp\}_{i \in \cM}$ from the adversary. 
    \end{itemize}
    
    \item {\bf Online Round 2.}
    \begin{itemize}
        \item Parse $\bn \coloneqq (\bn_1,\dots,\bn_n)$. 
        For each $i \in [n]$, compute $(\bx_i,\bz_i) \coloneqq {C_i^\inp}^\dagger X^{x_i^\inp} Z^{z_i^\inp}  X^{\widehat{x}_{\inp,i}}Z^{\widehat{z}_{\inp,i}}(\bn_i)$ and measure $\bz_i$. If any measurements are not all 0, then set $(\by_1,\dots,\by_n) \coloneqq (\bot,\cdots,\bot)$. Otherwise, query the ideal functionality with $\{\bx_i\}_{i \in \cM}$ and let $\{\by_i\}_{i \in \cM}$ be the output. Set $\by_i \coloneqq \mathbf{0}^{\ell_i}$ for each $i \in \cH$. 
        \item Sample $\{C_i^\out \gets \mathscr{C}_{\ell_i + \secp}\}_{i \in [n]}$ and set $$\by \coloneqq (C_1^\out(\by_1,\mathbf{0}^\secp),\cdots,C_n^\out(\by_n,\mathbf{0}^\secp)).$$
        \item Compute $$(\widetilde{\bx}^{\gray{N},\gray{Z_\inp},\gray{T_\inp}},D_0,\widetilde{g}_1,\dots,\widetilde{g}_d) \gets \QGSim\left(1^\secplev,Q_\dst,\by\right).$$
        \item Perform Bell measurements between $\widetilde{\bx}$ and $\be_{\Sim,1}$ and let $x_\Sim,z_\Sim$ be the teleportation errors.

        \item Send $\{\widehat{C}_i^T\}_{i \in \cM},X^{x_{\Sim}}Z^{z_{\Sim}}U_{\mathsf{garble}},D_0,\widetilde{g}_1,\dots,\widetilde{g}_d$ to the adversary. 
        \item Receive $\{\by_{T,i}\}_{i \in \cH}$ from the adversary.
    \end{itemize}

    \item {\bf Online Round 3.}
    \begin{itemize}
        \item For each $i \in \cH$, compute $\widehat{C}_i^{T}(\by_{T,i})$, and then perform the binary projection that projects the first $\secp$ qubits onto $\mathbf{T}^{\secp}$ and the last $\secp$ qubits onto $\mathbf{0}^{\secp}$. If this projection fails, then send abort to the ideal functionality (i.e. send $\{\abort_i\}_{i \in \cH}$ to the ideal functionality, and don't send a final round message to the adversary).
        \item Send $\{C_i^\out\}_{i \in \cM}$ to the adversary.
        \item Receive $\{\by_{\out,i}\}_{i \in \cH}$ from the adversary.
    \end{itemize}
    
    \item {\bf Output Reconstruction.}
    On behalf of each $i \in \cH$, compute $C_i^\out(\by_{\out,i})$ and measure whether the last $\secp$ trap qubits are all 0. If not, then send $\abort_i$ to the ideal functionality and otherwise send $\mathsf{ok}_i$.

\end{itemize}

\begin{lemma}
Let $\Pi$ be the protocol described in \cref{subsec:5-round-protocol} computing some quantum circuit $Q$. Then $\Pi$ satisfies \cref{def:mpqc} for any $\cA$ corrupting parties $M \subset [n]$ where $1 \in M$.
\end{lemma}

\begin{proof}

Fix any collection $\cD,\cA,\{\bx_i\}_{i \in [n]},\baux_\cA,\baux_\cD$.
We show the indistinguishability via a sequence of hybrids, where $\cH_0$ is the distribution $\Real_{\Pi,\Q}(\cA_\secp,\{\bx_i\}_{i \in [n]},\baux_\cA)$. In each hybrid, we describe the differences from the previous hybrid, and why each is indistinguishable.

\begin{itemize}
    \item $\cH_1$: \underline{Re-define $C^{\Circle}_k$}. \\ \\
    During Offline Round 2, 
    
    \begin{itemize}
    
    \item Sample $U_\enc,U_\chck$ honestly and sample $U_\test \gets \mathscr{C}_v$. Define
    
    $$C_k^\Circle \coloneqq {C_{k+1}^\Circle}^\dagger\dots {C_{n}^\Circle}^\dagger U_\test^\dagger (U_\enc \otimes \bbI)(\bbI \otimes U_\chck){C_{1}^\Circle}^\dagger\dots {C_{k-1}^\Circle}^\dagger,$$ and apply $C_k^\Circle$ to $\be_k \coloneqq \left(\be^{((k-1)\leftrightarrow k)}_S\right)^{\otimes v}$.
    \end{itemize}
    
    During Online Round 1,
    
    \begin{itemize}
    
    \item Let $\widehat{x}_1, \widehat{z}_1$ be such that $$(\bbI \ \otimes \ U_\chck){C_1^\Circle}^\dagger X^{x_1^\Circle}Z^{z_1^\Circle} \cdots {C_{k-1}^\Circle}^\dagger X^{x_{k-1}^\Circle}Z^{z_{k-1}^\Circle} = X^{\widehat{x}_1}Z^{\widehat{z}_1}(\bbI \ \otimes \ U_\chck){C_1^\Circle}^\dagger \cdots {C_{k-1}^\Circle}^\dagger.$$ Write $\widehat{x}_1$ as $\widehat{x}_{\enc},\widehat{x}_{\test,Z},\widehat{x}_{\test,T,1},\dots,\widehat{x}_{\test,T,n}$, and same for $\widehat{z}_1$. Here, $\widehat{x}_{\enc} \in \{0,1\}^{m + \secp n + k_0 + k_T}$, $\widehat{x}_{\test,Z} \in \{0,1\}^{k_0 + \secp n}$, and each $\widehat{x}_{\test,T,i} \in \{0,1\}^{2\secp}$.

    \item Let  $\widehat{x}_2, \widehat{z}_2$ be such that
    $$X^{x_{k}^\Circle}Z^{z_{k}^\Circle}{C_{k+1}^\Circle}^\dagger X^{x_{k+1}^\Circle}Z^{z_{k+1}^\Circle}\dots {C_{n}^\Circle}^\dagger X^{x_{n}^\Circle}Z^{z_{n}^\Circle}  = X^{\widehat{x}_2}Z^{\widehat{z}_2}{C_{k+1}^\Circle}^\dagger\dots {C_{n}^\Circle}^\dagger.$$
    
    \item Let $\widehat{U}_\test \coloneqq U_\test X^{\widehat{x}_2}Z^{\widehat{z}_2}$ and use $\widehat{U}_\test$ in place of $U_\test$ (in Online Round 1).
    
    \item Let $\widehat{U}_\enc:= U_\enc X^{\widehat{x}_\enc} Z^{\widehat{z}_\enc}$ and use $\widehat{U}_\enc$ in place of $U_\enc$ (in Online Round 2).
    
    \item For each $i \in [n]$, let $\widehat{C}_i^T \coloneqq C_i^T X^{\widehat{x}_{\test,T,i}}Z^{\widehat{z}_{\test,T,i}}$, and use $\widehat{C}_i^T$ in place of $C_i^T$ (in Online Round 2).
    
    \item Let $\widehat{r}:= r \oplus \widehat{x}_{\test,z}$ and use $\widehat{r}$ in place of $r$ (in Online Round 3).
        
    \end{itemize}
    
    This switch is perfectly indistinguishable from $\cH_0$ due to the fact that in $\cH_0$, $C_k^\Circle$ is a uniformly random Clifford, and in $\cH_1$, $U_\test$ is a uniformly random Clifford, and that $U_\enc,\{C_i^T\}_{i \in [n]}$ and $r$ are uniformly random.

    \item $\cH_2$: \underline{Re-define $U_\enc$}. \\ 
    
    During Offline Round 2,
    
    \begin{itemize}
        \item Sample $U_{\mathsf{garble}} \gets \mathscr{C}_{m+\secp n + k_0 + k_T}$, compute $(E_0,D_0,\widetilde{g}_1,\dots,\widetilde{g}_d) \gets \QGarble(1^\secp,Q_\dst[\{C_i^\inp,C_i^\out\}_{i \in [n]}])$, and use $U_{\mathsf{garble}}^\dagger E_0$ in place of $U_\enc$. Thus, we now have 
        $$C_k \coloneqq {C_{k+1}^\Circle}^\dagger\dots {C_{n}^\Circle}^\dagger U_\test^\dagger (U_{\mathsf{garble}}^\dagger \otimes \bbI)(E_0 \otimes \bbI)(\bbI \otimes U_\chck){C_{1}^\Circle}^\dagger\dots {C_{k-1}^\Circle}^\dagger.$$
        
    \end{itemize}
    
    During Online Round 2,
    
    \begin{itemize}
    
        \item Let $\widehat{x}_\inp,\widehat{z}_\inp$ be such that $$E_0 X^{\widehat{x}_\enc}Z^{\widehat{z}_\enc}\left(X^{x_1^\inp}Z^{z_1^\inp} \otimes \dots \otimes X^{x_n^\inp}Z^{z_n^\inp} \otimes \bbI^{\gray{Z_\inp},\gray{T_\inp}}\right) \coloneqq X^{\widehat{x}_\inp}Z^{\widehat{z}_\inp}E_0.$$
        \item Let $\widehat{U}_{\mathsf{garble}} \coloneqq X^{\widehat{x}_\inp}Z^{\widehat{z}_\inp}U_{\mathsf{garble}}$, and use $\widehat{U}_{\mathsf{garble}}$ in place of $U_{\mathsf{garble}}$.
    \end{itemize}

    This switch is perfectly indistinguishable from $\cH_2$ due to the fact that in $\cH_1$, $U_\enc$ is a uniformly random Clifford, and in $\cH_2$, $U_{\mathsf{garble}}$ is a uniformly random Clifford.


    \item $\cH_3$: \underline{Introduce new Pauli errors.} \\
    
    During Offline Round 2,
    \begin{itemize}
        \item Sample $x_{\Sim},z_{\Sim} \gets \{0,1\}^{m + \secp n + k_0 + k_T}$, and define $$C_k \coloneqq {C_{k+1}^\Circle}^\dagger\dots {C_{n}^\Circle}^\dagger U_\test^\dagger (U_{\mathsf{garble}}^\dagger \otimes \bbI)(X^{x_{\Sim}}Z^{z_{\Sim}} \otimes \bbI)(E_0 \otimes \bbI)(\bbI \otimes U_\chck){C_{1}^\Circle}^\dagger\dots {C_{k-1}^\Circle}^\dagger.$$
    \end{itemize}
    
    During Online Round 2,
    
    \begin{itemize}
        \item Let $\widehat{U}_{\mathsf{garble}} \coloneqq X^{x_{\Sim,}}Z^{z_{\Sim}}X^{\widehat{x}_\inp}Z^{\widehat{z}_\inp}U_{\mathsf{garble}}$.
    \end{itemize}
    
    This switch is perfectly indistinguishable since $U_{\mathsf{garble}}$ is a uniformly random Clifford.
    
    \item $\cH_4$: \underline{Introduce new EPR pairs.} 
    
    During Offline Round 2, prepare $m + \secp n + k_0 + k_T$ EPR pairs $\left\{\left(\be^{(i)}_{\Sim,1},\be^{(i)}_{\Sim,2}\right)\right\}_{i \in [m + \secp n + k_0 + k_T]}$, and let $$\be_{\Sim,1} \coloneqq \left(\be_{\Sim,1}^{(1)},\dots,\be_{\Sim,1}^{(m + \secp n + k_0 + k_T)}\right), \be_{\Sim,2} \coloneqq \left(\be_{\Sim,2}^{(1)},\dots,\be_{\Sim,2}^{(m + \secp n + k_0 + k_T)}\right).$$ 
    
    This hybrid will now compute $C_k^\Circle$ in three parts, as follows.

    \begin{itemize}
        \item First, compute 
        $$\left(\bn^{\gray{N}},\bz^{\gray{Z_\inp}},\bt^{\gray{T_\inp}},\bz_\test^{\gray{Z_\test}},\bt_{\test,1}^{\gray{T_{\test,1}}},\dots,\bt_{\test,n}^{\gray{T_{\test,n}}}\right) \coloneqq (\bbI \otimes U_\chck){C_{1}^\Circle}^\dagger\dots {C_{k-1}^\Circle}^\dagger(\be_{k}).$$ 
        
        \item Second, compute
        $$\left(\bx^{\gray{N},\gray{Z_\inp},\gray{T_\inp}}\right) \coloneqq E_0 \left(\bn,\bz,\bt\right).$$ 
         
        \item Third, compute
        $${C_{k+1}^\Circle}^\dagger\dots {C_{n}^\Circle}^\dagger U_\test^\dagger (U_{\mathsf{garble}}^\dagger \otimes \bbI)(\be_{\Sim,2},\bz_\test,\bt_{\test,1},\dots,\bt_{\test,n}).$$ In addition, perform Bell measurements between $\bx$ and $\be_{\Sim,1}$ to obtain teleportation errors, which will be defined to be $x_{\Sim},z_{\Sim}$. 
        
    \end{itemize}
    
    The second part of the computation (on registers $\gray{N},\gray{Z_\inp},\gray{T_\inp}$) will now be delayed until Online Round 2, since the teleportation errors are not used until the definition of $\widehat{U}_{\mathsf{garble}}$, which is given to the adversary in Online Round 2. Thus, this is identically distributed to $\cH_3$.

    \item $\cH_5$: \underline{Move around Pauli errors.} During Online Round 2,
    \begin{itemize}
        \item Change the computation to
        $$\bx \coloneqq E_0\left(X^{x_1^\inp}Z^{z_1^\inp} \otimes \dots \otimes X^{x_n^\inp}Z^{z_n^\inp} \otimes \bbI^{\gray{Z_\inp},\gray{T_\inp}}\right)\left(X^{\widehat{x}_\enc}Z^{\widehat{z}_\enc}\right)(\bn,\bz,\bt).$$
        \item Let $\widehat{U}_{\mathsf{garble}} \coloneqq X^{x_{\Sim}}Z^{z_{\Sim}}U_{\mathsf{garble}}$.
    \end{itemize}
    
    This is again identically distributed to $\cH_4$ since $U_{\mathsf{garble}}$ is a uniformly random Clifford.

    \item $\cH_6$: \underline{Simulate the quantum garbled circuit}. During Online Round 2, rather than applying $E_0$, do the following.
     \begin{itemize}
         \item Compute $$(\bn_\inp,\bz_\inp,\bz_\inp) \coloneqq \left(X^{x_1^\inp}Z^{z_1^\inp} \otimes \dots \otimes X^{x_n^\inp}Z^{z_n^\inp} \otimes \bbI^{\gray{Z_\inp},\gray{T_\inp}}\right)\left(X^{\widehat{x}_\enc}Z^{\widehat{z}_\enc}\right)(\bn,\bz,\bt).$$
         \item Compute $$\by \coloneqq Q_\dst[\{C_i^\inp,C_i^\out\}_{i \in [n]}](\bn_\inp,\bz_\inp,\bz_\inp).$$
         \item Compute $$(\widetilde{\bx}^{\gray{N},\gray{Z_\inp},\gray{T_\inp}},D_0,\widetilde{g}_1,\dots,\widetilde{g}_d) \gets \QGSim\left(1^\secplev,Q_\dst,\by\right).$$
         \item Perform Bell measurements between $\widetilde{\bx}$ and $\be_{\Sim,1}$ as before.
     \end{itemize}
     
     This is indistinguishable from $\cH_5$ due to security of the quantum garbled circuit.


    \item $\cH_{7}$: \underline{Switch honest party inputs to $\mathbf{0}$}. 
    \begin{itemize}
        \item During Online Round 1, for each $i \in \cH$, teleport $C_i^\inp(\mathbf{0}^{m_i + \secp})$ to Party 1 rather than $C_i^\inp(\bx_i,\mathbf{0}^{\secp})$.
        \item During Online Round 2, instead of directly applying $Q_\dst[\{C_i^\inp,C_i^\out\}_{i \in [n]}]$ to $(\bn_\inp,\bz_\inp,\bt_\inp)$, do the following. First apply ${C_1^\inp}^\dagger \otimes \dots \otimes {C_n^\inp}^\dagger$ to $\bn$. Then, swap out the honest party input registers for $\{\bx_i\}_{i \in \cH}$, and continue with the computation of $Q_\dst[\{C_i^\inp,C_i^\out\}_{i \in [n]}]$.
    \end{itemize}
    
    $\cH_7$ is statistically indistinguishable from $\cH_6$ due to properties of the Clifford authentication code. In particular, since the code is perfectly hiding, the adversary cannot tell that the inputs where switched to $\mathbf{0}$. Thus the adversary can only distinguish if the output of $Q_\dst[\{C_i^\inp,C_i^\out\}_{i \in [n]}]$ differs between $\cH_7$ and $\cH_6$. However, if any Clifford authentication test that happens within $Q_\dst[\{C_i^\inp,C_i^\out\}_{i \in [n]}]$ fails, then the output is $(\bot \dots \bot)$. In both $\cH_7$ and $\cH_6$, conditioned on these tests passing, the honest party inputs to $Q_\dst$ are statistically close to $\{\bx_i\}_{i \in \cH}$, due to the authentication property of the Clifford code.
    
    \item $\cH_8$: \underline{Query ideal functionality}. In Online Round 2, rather than computing $Q_\dst[\{C_i^\inp,C_i^\out\}_{i \in [n]}]$ on $(\bn_\inp,\bz_\inp,\bt_\inp)$, do the following.
    
    \begin{itemize}
        \item Parse $\bn_\inp \coloneqq (\bn_1,\dots,\bn_n)$. For each $i \in [n]$, compute $(\bx_i,\bz_i) \coloneqq C_i^\inp 
        (\bn_i)$ and measure $\bz_i$. If any measurements are not all 0, then set $(\by_1,\dots,\by_n) \coloneqq (\bot,\cdots,\bot)$.
        \item Otherwise, query the ideal functionality with $\{\bx_i\}_{i \in \cM}$ and let $\{\by_i\}_{i \in \cM}$ be the output. Set $\by_i \coloneqq \mathbf{0}^{\ell_i}$ for each $i \in \cH$.
        \item Sample $\{C_i^\out \gets \mathscr{C}_{\ell_i + \secp}\}_{i \in [n]}$ and set $$\by \coloneqq (C_1^\out(\by_1,\mathbf{0}^\secp),\cdots,C_n^\out(\by_n,\mathbf{0}^\secp)).$$
        \item Apply $\QGSim$ to $\by$ as before.
    \end{itemize}
    
    In the Output Reconstruction step, do the following.
    \begin{itemize}
        \item For each $i \in \cH$, 
        compute $C_i^\out(\by_i^\out)$ and measure whether the last $\secp$ trap qubits are all 0. If not, then send $\abort_i$ to the ideal functionality and otherwise send $\mathsf{ok}_i$.
    \end{itemize}

    Observe that one difference between $\cH_7$ and $\cH_8$ is that $\bz_\inp$ and $\bt_\inp$ are not used in the computation of $\CM$ circuit $Q_\dst$. Instead, the ideal functionality directly computes $Q$ on the inputs. This will result in statistically close outputs if i) the QGC satisfies statistical correctness, ii) $\bz_\inp$ is statistically close to $\mathbf{0}^{\otimes k_0}$, and iii) the result of applying the distillation circuit to $\bt_\inp$ is statistically close to $\mathbf{T}^{\otimes k_T / \secp}$. \cref{lemma:Ztest} implies that the second requirement holds conditioned on the adversary submitting the correct $r$ to the classical $\MPC$ in Online Round 1 (and otherwise, all honest parties abort). \cref{lemma:distillation} plus Clifford authentication implies that the third requirement holds conditioned on the honest party $T$-state checks in Online Round 3 all passing (and if any one of them fails, the honest parties abort).
    
    The other difference is that honest party outputs are determined by the ideal functionality's computation. First, the adversary cannot tell that $\{\by_i\}_{i \in \cH}$ are switched to $\mathbf{0}$ within the quantum garbled circuit, by perfect hiding of the Clifford authentication code (using Cliffords $\{C_i^\out\}_{i \in \cH}$). Next, in $\cH_7$, the adversary cannot make an honest party accept a state noticeably far from their real output $\by_i$, by authentication of the Clifford code. Thus, $\cH_7$ is statistically close to $\cH_8$. This completes the proof, as $\cH_8$ is the simulator described above.

\end{itemize}

\end{proof}

\subsubsection{Case 2: $P_1$ is the only honest party}

\paragraph{Simulator.} The simulator will act as party 1 and maintain the classical $\MPC$ oracle. It will compute $\MPC$ honestly throughout, and will compute honest party 1 actions throughout except for what is described below.

\begin{itemize}
    \item \textbf{Online Round 1.} Rather than Clifford-encoding and teleporting in $P_1$'s input $\bx_1$, the simulator will teleport $C_1^\inp(\mathbf{0}^{\ell_1 + \secp})$.
    \item \textbf{Online Round 3.} Rather than evaluating the quantum garbled circuit $(U_{\mathsf{garble}},D_0,\widetilde{g}_1,\dots,\widetilde{g}_d)$ on $\by$, the simulator will run the following computation on $\by$.
    \begin{itemize}
        \item Compute $(\bn_1,\dots,\bn_n,\bz,\bt) \coloneqq E_0^\dagger U_{\mathsf{garble}}(\by)$, where $(\bn_1,\dots,\bn_n)$ are the parties' Clifford-encoded inputs. 
        \item For each $i \in [n]$, compute $(\bx_i,\bz_i) \coloneqq {C_i^\inp}^\dagger (\bn_i)$ and measure $\bz_i$. If any measurements are not all 0, then set $(\by_1,\dots,\by_n) \coloneqq (\bot,\cdots,\bot)$. Otherwise, query the ideal functionality with $\{\bx_i\}_{i \in [2,\dots,n]}$ and let $\{\by_i\}_{i \in [2,\dots,n]}$ be the output.
        \item For each $i \in [2,\dots,n]$ set $\by_i^\out \coloneqq C_i^\out(\by_i,\mathbf{0}^\secp)$, set $\by_1^\out \coloneqq C_1^\out(\mathbf{0}^{\ell_1+\secp})$, and continue as the honest party 1.
    \end{itemize}
\end{itemize}

\begin{lemma}
Let $\Pi$ be the protocol described in \cref{subsec:5-round-protocol} computing some quantum circuit $Q$. Then $\Pi$ satisfies \cref{def:mpqc} for any $\cA$ corrupting parties $[2,\dots,n]$.

\end{lemma}

\begin{proof}
We consider a sequence of hybrid distributions, where $\cH_0$ is $\Real_{\Pi,\Q}(\cA_\secp,\{\bx_i\}_{i \in [n]},\baux_\cA)$, i.e. the real interaction between $\cA_\secp(\{\bx_i\}_{i \in [2,\dots,n]},\baux_\cA)$ and an honest party $P_1(1^\secp,\bx_1)$. In each hybrid, we describe the differences from the previous hybrids.

\begin{itemize}
    \item $\cH_1:$ \underline{Directly compute $Q_\dst[\{C_i^\inp,C_i^\out\}]$ in place of garbled circuit evaluation} During Online Round 3, this hybrid computes $(\bn_1,\dots,\bn_n,\bz,\bt) \coloneqq E_0^\dagger U_{\mathsf{garble}}(\by)$ and then applies $Q_\dst[\{C_i^\inp,C_i^\out\}]$ to produce outputs $(\by_1^\out,\dots,\by_n^\out)$. Statistical indistinguishability follows from the statistical correctness of the QGC.
    \item $\cH_2:$ \underline{Replace $P_1$'s input with $\mathbf{0}$} During Online Round 1, this hybrid teleports in $C_1^\inp(\mathbf{0}^{m_i + \secp})$. Then, during Online Round 2, this hybrid inserts $P_1$'s input $\bx_1$ before the computation of $Q$. This switch is perfectly indistinguishable due to the perfect hiding of the Clifford code.
    \item $\cH_3:$ \underline{Query ideal functionality} During Online Round 3, this hybrid computes $Q_\dst[\{C_i^\inp,C_i^\out\}]$ as described in the simulator, by using the ideal functionality to compute $Q$. This switch is statistically indistinguishable as long as i) $\bz$ is statistically close to $\mathbf{0}^{\otimes k_0}$ (which follows from \cref{lemma:Ztest}), and ii) the result of applying the distillation circuit to $\bt$ to statistically close to $\mathbf{T}^{\otimes k_T / \secp}$ (which follows from \cref{lemma:distillation}, as $P_1$ checks its own subset of T states). This hybrid is the simulator, completing the proof.
\end{itemize}

\end{proof}

\section{Multi-Party Quantum Computation in Four Rounds}

In this section, we show the existence of a four-round protocol for multi-party quantum computation, assuming the sub-exponential hardness of LWE. The protocol we present satisfies security with abort, and only requires two rounds of online communication (that is, two rounds of communication once the parties receive their inputs). Thus, this implies the existence of a two-round protocol for multi-party quantum computation given some input-independent quantum pre-processing. 

We also note that the protocol can be adjusted to give security with \emph{unanimous} abort with three rounds of online communication (while keeping the total number of rounds at four), though we do not provide a formal description of this protocol. Roughly, this follows because if parties receive their inputs one round earlier, they will be able to receive and check the authenticity of their (encrypted) outputs at the end of round three, rather than checking the authenticity of their (unencrypted) outputs at the end of round four. 

\subsection{The Protocol}\label{subsec:4-round-protocol}

\paragraph{Ingredients.} Our protocol will make use of the following cryptographic primitives, which are all assumed to be sub-exponentially secure (i.e. there exists $\epsilon$ such that the primitive is $(2^{-\secp^{\epsilon}})$-secure).

\begin{itemize}
    \item Round-optimal quantum-secure multi-party computation for classical reactive functionalities in the CRS model, to be treated as an oracle called $\MPC$ ~(see \cref{subsec:reactivempc}).
    \item A garbling scheme for $\CM$ circuits $(\QGarble, \QGEval, \QGSim)$.
    \item A quantum multi-key FHE scheme $\qmfhe=(\KeyGen,\CEnc,\Enc,\Eval,\Rerand,\Dec)$ with ciphertext rerandomization and classical encryption of classical messages.
    \item A quantum-secure classical garbled circuit $(\Garble,\GEval,\GSim)$.
\end{itemize}

\paragraph{Notation.} We use the same notation as described in \cref{subsec:5-round-protocol}.

\protocol
{\proref{fig:classical-4}: Classical Functionality for Four-Round Quantum MPC}
{Classical Functionality for Four-Round Quantum MPC.}
{fig:classical-4}
{
\textbf{Public Parameters:} Security parameter $\secp$, number of parties $n$, $\CM$ circuit $Q$, and parameters $(m,\ell,k_0,k_T,v,\epsilon)$ defined above. Let $\secplev = (2nv + 2m)^{c}$ for some $c > 1/\epsilon$.\\

\textbf{Shared Randomness:} Random strings for $0$ state check $r,s \gets \{0,1\}^{k_0 + \secp n}$, re-randomization matrix $U_{\mathsf{enc}} \gets \mathscr{C}_{m+\secp n + k_0 + k_T}$, Cliffords for $T$-state checks $\{C_i^{T} \gets \mathscr{C}_{2\secp}\}_{i \in [n]}$, Cliffords for outputs $\{C_i^{\out} \gets \mathscr{C}_{\ell_i + \secp}\}_{i \in [n]}$.\\ 

\textbf{Offline Round 1:} Obtain input $\left(C_i^{\Circle},C_i^{\inp}\right)$ from each party $i$.

\textbf{Offline Round 2:} Obtain input $\left(x_i^\Circle,z_i^\Circle\right)$ from each party $i$.
\begin{itemize}
    \item Sample a random linear map $M \gets \mathsf{GL}(2(k_0 + \secp n),\bbF_2)$, and a random permutation $\pi$ on $k_T + \secp n$ elements. Let $f_\test[\{C_i^\Circle\}_{i \in [n]},r,s,U_\enc,\{C_i^T\}_{i \in [n]},M,\pi]$ be the classical circuit that takes as input $\{x_i^\Circle,z_i^\Circle\}_{i \in [n]}$ and outputs $U_\test$, following the computation in Online Round 1 of \proref{fig:classical}.
    \item Compute $(\{\lab^\test_{i,b}\}_{i \in [2nv],b \in \{0,1\}},\widetilde{f}_\test) \gets \Garble(1^\secplev,f_{\test}[\{C_i^\Circle\}_{i \in [n]},r,s,U_\enc,\{C_i^T\}_{i \in [n]},M,\pi])$.
    \item For each $i \in [2nv], b \in \{0,1\}$, compute $(\pk^\test_{i,b},\sk^\test_{i,b}) \gets \qmfhe.\Gen(1^\secplev)$ and $\ct^\test_{i,b} \gets \qmfhe.\CEnc(\pk^\test_{i,b},\lab^\test_{i,b})$.
    \item Output $\left(\{\pk^\test_{i,b},\ct^\test_{i,b}\}_{i \in [2nv], b \in \{0,1\}}, \widetilde{f}_\test\right)$ to party 1.
   
\end{itemize}

\textbf{Online Round 1:} Obtain input $\left(x_i^\inp,z_i^\inp\right)$ from each party $i$ and additional input $\ct_r$ from party 1.
\begin{itemize}
    \item Let $f_\QGC[U_\enc,\{C_i^\inp,C_i^\out\}_{i \in [n]}]$ be the circuit that takes as input $\{x_i^\inp,z_i^\inp\}_{i \in [n]}$ and outputs $(U_\mathsf{garble},D_0,\widetilde{g}_1,\dots,\widetilde{g}_d)$, following the computation in Online Round 2 of \proref{fig:classical}.
    \item Compute $(\{\lab^\QGC_{i,b}\}_{i \in [2m],b \in \{0,1\}},\widetilde{f}_\QGC]) \gets \Garble(1^\secplev,f_{\QGC}[U_\enc,\{C_i^\inp,C_i^\out\}_{i \in [n]}])$.
    \item For each $i \in [2m], b \in \{0,1\}$, compute $(\pk^\QGC_{i,b},\sk^\QGC_{i,b}) \gets \qmfhe.\Gen(1^\secplev)$ and $\ct^\QGC_{i,b} \gets \qmfhe.\CEnc(\pk^\QGC_{i,b},\lab^\QGC_{i,b})$.
    \item Let $t^\Circle \coloneqq (x_1^\Circle,z_1^\Circle,\dots,x_n^\Circle,z_n^\Circle)$, and let $\sk_\test = \{\sk^\test_{i,t^\Circle_i}\}_{i \in [2nv]}$.
    \item Output $\left(\{\pk^\QGC_{i,b},\ct^\QGC_{i,b}\}_{i \in [2m], b \in \{0,1\}}, \widetilde{f}_\QGC\right)$ to party 1 and $(C_i^T,\sk_\test)$ to party $i$ for each $i \in [n]$.
    
\end{itemize}

\textbf{Online Round 2:}
\begin{itemize}
    \item Let $t^\inp \coloneqq (x_1^\inp,z_i^\inp,\dots,x_n^\inp,z_n^\inp)$, and let $\sk_\QGC = \{\sk^\QGC_{i,t^\inp_i}\}_{i \in 2m}$.
    \item Decrypt $\ct_r$ with $\sk_\test$ to obtain a value $r'$. If $r' = r$ then output $(C_i^\out,\sk_\QGC)$ to party $i$ for each $i \in [n]$, and otherwise output $\bot$ to each party.
\end{itemize}
}

\paragraph{Offline Round 1.} The parties send EPR halves to each other exactly as described in Offline Round 1 of the five-round protocol from \cref{sec:five-round}. In addition, each party $i$ samples $C_i^\Circle \gets \mathscr{C}_v$ and $C_i^\inp \gets \mathscr{C}_{m_i + \secp}$ and inputs $(C_i^\Circle,C_i^\inp)$ to the classical $\MPC$.

\paragraph{Offline Round 2.} 
\begin{itemize}
    \item \underline{Party $P_i$ for $i \in [n]$.} The parties apply $C_i^\Circle$ to their registers and teleport the results around the circle, exactly as described in Offline Round 2 of the five-round protocol from \cref{sec:five-round}. This results in teleportation errors $x_i^\Circle,z_i^\Circle$, which party $i$ broadcasts to all parties.
    \item  \underline{Party $P_1$.} $P_1$ receives output $\left(\left\{\pk^\test_{i,b},\ct^\test_{i,b}\right\}_{i \in [2nv], b \in \{0,1\}}, \widetilde{f}_\test\right)$ from $\MPC$.
\end{itemize}

\paragraph{Parties Receive Inputs.} After the offline rounds, each party $P_i$ receives an $m_i$-qubit input $\bx_i$. 

\paragraph{Online Round 1.} 
\begin{itemize}
    \item \underline{Party $P_i$ for $i \in [n]$.} The parties Clifford encode their inputs and teleport the encodings exactly as described in Online Round 1 of the five-round protocol from 
    \cref{sec:five-round}. This results in teleportation errors $x_i^\inp,z_i^\inp$, which party $i$ broadcasts to all parties. Each party $i$ receives output $\left(C_i^T,\sk_\test\right)$ from $\MPC$.
    \item \underline{Party $P_1$.} $P_1$ first homomorphically evaluates $\widetilde{f}_\test$ using encryptions of the labels corresponding to the $\left\{x_i^\Circle,z_i^\Circle\right\}_{i \in [n]}$ that were broadcast last round. This results in $\qmfhe.\Enc(\sk_\test,U_\test)$, where recall from \proref{fig:classical-4} that $\sk_\test$ is a set of $2nv$ secret keys. Now, $P_1$ homomorphically evaluates $U_\test$ on its registers, exactly as in Online Round 2 of the five-round protocol from \cref{sec:five-round} (except here the computation is performed homormophically under $\sk_\test$). This results in the following encryptions: $$\qmfhe.\Enc(\sk_\test, \by), \qmfhe.\Enc(\sk_\test, r'), \qmfhe.\Enc\left(\sk_\test, \by_{T,1}\right),\dots,\qmfhe.\Enc\left(\sk_\test,\by_{T,n}\right).$$   
    $P_1$ sends $\qmfhe.\Enc\left(\sk_\test, \by_{T,i}\right)$ to party $P_i$ and inputs $\qmfhe.\Enc(\sk_\test, r')$ to $\MPC$. Finally, it receives $\left(\left\{\pk^\QGC_{i,b},\ct^\QGC_{i,b}\right\}_{i \in [2m], b \in \{0,1\}}, \widetilde{f}_\QGC\right)$ from $\MPC$.
\end{itemize}

\paragraph{Online Round 2.}
\begin{itemize}
    \item \underline{Party $P_i$ for $i \in [n]$.} Each party $P_i$ decrypts the state $\qmfhe.\Enc\left(\sk_\test,\by_{T,i}\right)$ received from $P_1$ using $\sk_\test$ received from $\MPC$. It computes $C_i^{T}\left(\by_{T,i}\right)$, and then performs the binary projective measurement that projects the first $\secp$ qubits onto $\mathbf{T}^{\secp}$ and the last $\secp$ qubits onto $\mathbf{0}^{\secp}$. If this measurement returns $0$, party $P_i$ sends $\abort$ to $\MPC$. If no parties send $\abort$, then they each receive $\left(C_i^\out,\sk_\QGC\right)$ from $\MPC$.
    \item \underline{Party $P_1$.} $P_1$ first decrypts $\qmfhe.\Enc(\sk_\test, \by)$ using $\sk_\test$. Then, it homomorphically evaluates $\widetilde{f}_\QGC$ using encryption of the labels corresponding to the $\left\{x_i^\inp,z_i^\inp\right\}_{i \in [n]}$ that were broadcast last round. This results in $\qmfhe.\Enc(\sk_\QGC,(U_\mathsf{garble},D_0,\widetilde{g}_1,\dots,\widetilde{g}_d))$, where recall from \proref{fig:classical-4} that $\sk_\QGC$ is a set of $2m$ secret keys. Now, $P_1$ homomorphically evaluates the quantum garbled circuit on $\by$ to obtain encryptions $\qmfhe.\Enc(\sk_\QGC,\by_{\out,1}),\dots,\qmfhe.\Enc(\sk_\QGC,\by_{\out,n})$. For each $i \in [2,\dots,n]$, apply $\qmfhe.\Rerand$ to $\qmfhe.\Enc(\sk_\QGC,\by_{\out,i})$ and send the resulting ciphertext to $P_i$.
\end{itemize}

\paragraph{Output Reconstruction}
\begin{itemize}
    \item \underline{Party $P_i$ for $i \in [n]$.} Each party $P_i$ decrypts their state $\qmfhe.\Enc(\sk_\QGC,\by_{\out,i})$ using $\sk_\QGC$. Then, it computes $C_i^{\out}(\by_{\out,i})$ and measures whether the last $\secp$ trap qubits are all $0$. If not, it outputs $\bot$. Otherwise, its output is the remaining $\ell_i$ qubits. 
\end{itemize}

\subsection{Security}

\begin{theorem}
Assuming post-quantum maliciously-secure two-message oblivious transfer and (levelled) multi-key quantum fully-homomorphic encryption with sub-exponential security, there exists four-round multi-party quantum computation. Both of the above assumptions are known from the sub-exponential hardness of QLWE.
\end{theorem}

We split the security proof into two cases based on the set of honest parties $\cH \subset [n]$. In the first case, $P_1$ is corrupted and in the second case, the only honest party is $P_1$.

\subsubsection{Case 1: $P_1$ is corrupted}

\paragraph{Simulator.}

Let $\cH \subset [n]$ be the set of honest parties and $k \neq 1$ be arbitrary such that $k \in \cH$. The simulator will act as party $k$, altering its actions as described below. All actions of parties $i \in \cH \setminus \{k\}$ will be honest except for those explicitly mentioned in the simulation (essentially just switching out their inputs for 0).

\begin{itemize}
    \item \textbf{Offline Round 1.} 
    \begin{itemize}
        \item Following honest party $k$'s behavior, prepare $v$ EPR pairs $\left\{\left(\be^{(i)}_{R},\be^{(i)}_{S}\right)\right\}_{i \in [v]}$, let $$\be_{R} \coloneqq \left(\be_{R}^{(1)},\dots,\be_{R}^{(v)}\right), \be_{S} \coloneqq \left(\be_{S}^{(1)},\dots,\be_{S}^{(v)}\right),$$ and send $\be_S$ to party $k-1$. Receive a state $\be_{k+1}$ of $v$ qubits from $P_{k+1}$
        \item For each $i \in \cH$, receive a state $\be_{\inp,i}$ of $m_i + \secp$ qubits from $P_1$.
        \item  Obtain $\{C_i^\Circle,C_i^\inp\}_{i \in [n] \setminus \cH}$ from the adversary's query to $\MPC$, and let $\{C_i^\Circle,C_i^\inp\}_{i \in \cH}$ be the values sampled by honest parties.
    \end{itemize}
    \item \textbf{Offline Round 2.} 
    \begin{itemize}
    
        \item Sample $r,s \gets \{0,1\}^{k_0+\secp n}$, $\{C_i^T \gets \mathscr{C}_{2\secp}\}_{i \in [n]}$, $M \gets \mathsf{GL}(2(k_0 + \secp n),\bbF_2)$, and a permutation $\pi$ on $k_T + \secp n$ elements. Use these values to compute $U_\chck$ as in the computation in Online Round 1 of \proref{fig:classical}.

        \item Compute $\left(\bn^{\gray{N}},\bz^{\gray{Z_\inp}},\bt^{\gray{T_\inp}},\bz_\test^{\gray{Z_\test}},\bt_{\test,1}^{\gray{T_{\test,1}}},\dots,\bt_{\test,n}^{\gray{T_{\test,n}}}\right) \coloneqq (\bbI \ \otimes \ U_\chck){C_1^\Circle}^\dagger \cdots {C_{k-1}^\Circle}^\dagger(\be_S)$, and discard $\bz^{\gray{Z_\inp}},\bt^{\gray{T_\inp}}$.

        \item Prepare $\ell + \secp n$ EPR pairs $\left\{\left(\be^{(i)}_{\Sim,1},\be^{(i)}_{\Sim,2}\right)\right\}_{i \in [\ell + \secp n]}$, and let $$\be_{\Sim,1} \coloneqq \left(\be_{\Sim,1}^{(1)},\dots,\be_{\Sim,1}^{(\ell + \secp n)}\right), \be_{\Sim,2} \coloneqq \left(\be_{\Sim,2}^{(1)},\dots,\be_{\Sim,2}^{(\ell + \secp n)}\right).$$ 
        
        \item Sample $U_\test \gets \mathscr{C}_v$ and $U_{\mathsf{garble}} \gets \mathscr{C}_{m + n \secp + k_0 + k_T}$.
        \item Compute $(\widetilde{\bx}_\inp,D_0,\widetilde{g}_1,\dots,\widetilde{g}_d) \gets \QGSim\left(1^\secplev,Q_\dst,\be_{\Sim,2}\right)$, and let $$\bw \coloneqq {C_{k+1}^\Circle}^\dagger(\cdots {C_{n}^\Circle}^\dagger(U_\test^\dagger(U_{\mathsf{garble}}^\dagger(\widetilde{\bx}_\inp),\bz_\test,\bt_{\test,1},\dots,\bt_{\test,n}))\cdots).$$
        \item Compute $\left(\{\widetilde{\lab^\test_i}\}_{i \in [2nv]},\widetilde{f}_{\test}\right) \gets \GSim\left(1^\secplev,1^{2nv},1^{|f_{\test}|},U_{\test}\right)$.
        \item For $i \in [2nv], b \in \{0,1\}$, sample $$(\pk_{i,b}^\test,\sk_{i,b}^\test) \gets \qmfhe.\Gen(1^\secplev), \ct_{i,b}^\test \gets \qmfhe.\CEnc(\pk_{i,b}^\test,\widetilde{\lab^\test_i}).$$
        \item Teleport $\bw$ into $\be_{k+1}$, and let $x_k^\Circle,z_k^\Circle$ be the teleportation errors. Let $\{x_i^\Circle,z_i^\Circle\}_{i \in \cH \setminus \{k\}}$ be the teleportation errors obtained by the other honest parties.
        \item Send $\{x_i^\Circle,z_i^\Circle\}_{i \in \cH},\{\pk_{i,b}^\test,\ct_{i,b}^\test\}_{i \in [2nv], b \in \{0,1\}},\widetilde{f}_{\test}$ to the adversary.
        \item Receive $\{x_i^\Circle,z_i^\Circle\}_{i \in [n] \setminus \cH}$ from the adversary.

    \end{itemize}
    \item \textbf{Online Round 1.}
    \begin{itemize}
    
        \item Let $\widehat{x}_1, \widehat{z}_1$ be such that $$(\bbI \ \otimes \ U_\chck){C_1^\Circle}^\dagger X^{x_1^\Circle}Z^{z_1^\Circle} \cdots {C_{k-1}^\Circle}^\dagger X^{x_{k-1}^\Circle}Z^{z_{k-1}^\Circle} = X^{\widehat{x}_1}Z^{\widehat{z}_1}(\bbI \ \otimes \ U_\chck){C_1^\Circle}^\dagger \cdots {C_{k-1}^\Circle}^\dagger.$$ Write $\widehat{x}_1$ as $\widehat{x}_{\inp,1},\dots,\widehat{x}_{\inp,n},\widehat{x}_z,\widehat{x}_t,\widehat{x}_{\test}$, and same for $\widehat{z}_1$. Here, each $\widehat{x}_{\inp,i} \in \{0,1\}^{m_i + \secp}$, $\widehat{x}_z \in \{0,1\}^{k_0}$, $\widehat{x}_t \in \{0,1\}^{k_T}$, and $\widehat{x}_{\test} \in \{0,1\}^{k_0 + 3 \secp n}$.

        \item Let $\widehat{x}_2, \widehat{z}_2$ be such that $$\left(\bbI \otimes X^{\widehat{x}_\test} Z^{\widehat{z}_\test}\right)U_\test X^{x_{n}^\Circle}Z^{z_{n}^\Circle}C^{\Circle}_{n}\cdots X^{x_{k+1}^\Circle}Z^{z_{k+1}^\Circle}C^{\Circle}_{k+1}X^{x_k^\Circle}Z^{z_k^\Circle} =  X^{\widehat{x}_2}Z^{\widehat{z}_2} U_\test C^{\Circle}_n\cdots C^{\Circle}_{k+1}.$$ Write $\widehat{x}_2$ as $\widehat{x}_{\mathsf{garble}},\widehat{x}_{\test,z},\widehat{x}_{\test,T,1},\dots,\widehat{x}_{\test,T,n}$, and same for $\widehat{z}_2$. Here, $\widehat{x}_{\mathsf{garble}} \in \{0,1\}^{m + \secp n + k_0 + k_T}$, $\widehat{x}_{\test,Z} \in \{0,1\}^{k_0 + \secp n}$, and each $\widehat{x}_{\test,T,i} \in \{0,1\}^{2\secp}$.



        
        \item Let $\widehat{r} \coloneqq r \oplus \widehat{x}_{\test,Z}$.
        \item For each $i \in [n]$, let $\widehat{C}_i^T \coloneqq C_i^T X^{\widehat{x}_{\test,T,i}}Z^{\widehat{z}_{\test,T,i}}$.
        \item Compute $(\{\widetilde{\lab^\QGC_i}\}_{i \in [2nv]},\widetilde{f}_{\QGC}) \gets \GSim(1^\secplev,1^{2m},1^{|f_{\QGC}|},(U_{\mathsf{garble}}X^{\widehat{x}_{\mathsf{garble}}}Z^{\widehat{z}_{\mathsf{garble}}},D_0,\widetilde{g}_1,\dots,\widetilde{g}_d))$.
        \item For $i \in [2m],b \in \{0,1\}$, sample $$(\pk_{i,b}^\QGC,\sk_{i,b}^\QGC) \gets \qmfhe.\Gen(1^{\secplev}), \ct_{i,b}^\QGC \gets \qmfhe.\CEnc(\pk_{i,b}^\QGC,\widetilde{\lab^\QGC_i}).$$
        \item For $i \in \cH$, teleport $C_i^\inp(\mathbf{0}^{m_i + \secp})$ into $\be_{\inp,i}$. Let $x_i^\inp,z_i^\inp$ be the teleportation errors.
        \item Send $\{x_i^\inp,z_i^\inp\}_{i \in \cH},\{\widehat{C}_i^T\}_{i \in [n] \setminus \cH},\sk_\test,\{\pk_{i,b}^\QGC,\ct_{i,b}^\QGC\}_{i \in [2m], b \in \{0,1\}},\widetilde{f}_\QGC$ to the adversary (where $\sk_\test$, defined in \proref{fig:classical-4}, is the set of secret keys corresponding to $\{x_i^\Circle,z_i^\Circle\}_{i \in [n]}$).
        \item Receive $\ct_Z,\{\ct_{T,i}\}_{i \in \cH}, \{x_i^\inp,z_i^\inp\}_{i \in [n] \setminus \cH}$ from the adversary (where $\sk_\QGC$, defined in \proref{fig:classical-4}, is the set of secret keys corresponding to $\{x_i^\inp,z_i^\inp\}_{i \in [n]}$).
    \end{itemize}
    \item \textbf{Online Round 2.}
    \begin{itemize}
        \item For each $i \in \cH$, decrypt $\ct_{T,i}$ with $\sk_\test$ to obtain $\by_{T,i}$. Compute $\widehat{C}_i^{T}(\by_{T,i})$, and then perform the binary projection that projects the first $\secp$ qubits onto $\mathbf{T}^{\secp}$ and the last $\secp$ qubits onto $\mathbf{0}^{\secp}$. If this projection fails, then abort (i.e. send $\{\abort_i\}_{i \in \cH}$ to the ideal functionality, and don't send a final round message to the adversary).
        \item Decrypt $\ct_Z$ with $\sk_\test$ to obtain $r'$. Check if $r' = \widehat{r}$. If not, then send $\bot$ to the adversary, send $\{\abort_i\}_{i \in \cH}$ to the ideal functionality, and exit. Otherwise, continue.
        \item Parse $\bn \coloneqq (\bn_1,\dots,\bn_n)$. For each $i \in [n]$, compute $(\bx_i,\bz_i) \coloneqq C_i^\inp X^{x_i^\inp} Z^{z_i^\inp} X^{\widehat{x}_{\inp,i}} Z^{\widehat{z}_{\inp,i}} (\bn_i)$ and measure $\bz_i$. If any measurements are not all 0, then set $(\by_1,\dots,\by_n) \coloneqq (\bot,\cdots,\bot)$. Otherwise, query the ideal functionality with $\{\bx_i\}_{i \in [n] \setminus \{\cH\}}$ and let $\{\by_i\}_{i \in [n] \setminus \{\cH\}}$ be the output. Set $\by_i \coloneqq \mathbf{0}^{\ell_i}$ for each $i \in \cH$.
        \item Sample $\{C_i^\out \gets \mathscr{C}_{\ell_i + \secp}\}_{i \in [n]}$.
        \item Teleport $(C_1^\out(\by_1,\mathbf{0}^\secp),\cdots,C_n^\out(\by_n,\mathbf{0}^\secp))$ into $\be_{\Sim,1}$, and let $(x_1^\out,z_1^\out,\cdots,x_n^\out,z_n^\out)$ be the teleportation errors. For each $i \in [n]$, let $\widehat{C}_i^\out \coloneqq C_i^\out X^{x_i^\out}Z^{z_i^\out}$.
        \item Send $\sk_\QGC,\{\widehat{C}_i^\out\}_{i \in [n] \setminus \cH}$ to the adversary. 
    \end{itemize}
    \item \textbf{Output Reconstruction}
    \begin{itemize}
        \item Receive $\{\ct_{y,i}\}_{i \in \cH}$ from the adversary. 
        \item For each $i \in \cH$, decrypt $\ct_{y,i}$ using $\sk_\QGC$ to obtain $\by_{\out,i}$. Then, compute $\widehat{C}_i^\out(\by_{\out,i})$ and measure whether the last $\secp$ trap qubits are all 0. If not, then send $\abort_i$ to the ideal functionality and otherwise send $\mathsf{ok}_i$.
    \end{itemize}
\end{itemize}

\paragraph{Notation.} For any adversary $\{\cA_\secp\}_{\secp \in \bbN}$ and inputs $(\bx_1,\dots,\bx_n,\baux_\cA,\baux_\cD)$, we partition the distributions $\Real_{\Pi,\Q}(\cA_\secp,\{\bx_i\}_{i \in [n]},\baux_\cA)$ and  $\Ideal_{\Pi,\Q}(\Sim,\{\bx_i\}_{i  \in [n]},\baux_\cA)$ by the set of teleportation errors $\{x_i^\Circle,z_i^\Circle,x_i^\inp,z_i^\inp\}_{i \in [n]}$ broadcast by all parties throughout the protocol. That is, we define the distribution $\Real_{\Pi,\Q}^{\{x_i^\Circle,z_i^\Circle,x_i^\inp,z_i^\inp\}_{i \in [n]}}(\cA_\secp,\{\bx_i\}_{i \in [n]},\baux_\cA)$ to be $\Real_{\Pi,\Q}(\cA_\secp,\{\bx_i\}_{i \in [n]},\baux_\cA)$ except that the output of the distribution (which is a state $\bz$) is replaced with $\bot$ if the set of teleportation errors broadcast during execution of the protocol were not exactly $\{x_i^\Circle,z_i^\Circle,x_i^\inp,z_i^\inp\}_{i \in [n]}$. We define the distribution $\Ideal_{\Pi,\Q}^{\{x_i^\Circle,z_i^\Circle,x_i^\inp,z_i^\inp\}_{i \in [n]}}(\Sim,\{\bx_i\}_{i \in [n]},\baux_\cA)$ analogously. 


We now prove the following lemma, which is the main part of the proof of security for this case. For notational convenience, we drop the indexing of inputs and teleportation errors by $\secp$.

\begin{lemma}
\label{lemma:distribution-bot-mpc}
For any QPT $\cA = \{\cA_\secp\}_{\secp \in \bbN}$, QPT distinguisher $\cD = \{\cD_\secp\}_{\secp \in \bbN}$, inputs $(\bx_1,\dots,\bx_n,\baux_\cA,\baux_\cD)$, and teleportation errors $\{x_i^\Circle,z_i^\Circle,x_i^\inp,z_i^\inp\}_{i \in [n]}$, there exists a negligible function $\mu$ such that

\begin{align*}
&\bigg|\Pr\left[\cD_\secp\left(\baux_\cD,\Real_{\Pi,\Q}^{\{x_i^\Circle,z_i^\Circle,x_i^\inp,z_i^\inp\}_{i \in [n]}}\left(\cA_\secp,\{\bx_i\}_{i \in [n]},\baux_\cA\right)\right) = 1\right]\\ &- \Pr\left[\cD_\secp\left(\baux_\cD,\Ideal_{\Pi,\Q}^{\{x_i^\Circle,z_i^\Circle,x_i^\inp,z_i^\inp\}_{i \in [n]}}\left(\Sim_\secp,\{\bx_i\}_{i \in [n]},\baux_\cA\right) \right) = 1  \right] \bigg| \leq \frac{\mu(\secp)}{2^{2nv+2m}}.
\end{align*}
\end{lemma}

\begin{proof}
First note that by the definition of $\secplev$, a $\cD$ violating the lemma distinguishes with probability at least $\left(\frac{1}{\poly(\secp)}\right)2^{-\secplev^{(1/c)})} \geq \frac{1}{2^{\secplev^{\epsilon}}}.$

Now fix any collection $\cD,\cA,\{\bx_i\}_{i \in [n]},\baux_\cA,\baux_\cD,\{x_i^\Circle,z_i^\Circle,x_i^\inp,z_i^\inp\}_{i \in [n]}$.  We show the indistinguishability via a sequence of hybrids, where $\cH_0$ is the distribution $\Real_{\Pi,\Q}^{\{x_i^\Circle,z_i^\Circle,x_i^\inp,z_i^\inp\}_{i \in [n]}}(\cA_\secp,\{\bx_i\}_{i \in [n]},\baux_\cA)$. In each hybrid, we describe the differences from the previous hybrid, and why each is indistinguishable.

\begin{itemize}
    \item $\cH_1$: \underline{Switch half of QMFHE ciphertexts to encryptions of 0}. Let $t^\Circle \coloneqq (x_1^\Circle,z_1^\Circle,\dots,x_n^\Circle,z_n^\Circle)$ and let $t^\inp \coloneqq (x_1^\inp,z_i^\inp,\dots,x_n^\inp,z_n^\inp)$. For each $i \in [2nv]$, compute $$\ct^\Circle_{i,1-t^\Circle_i} \gets \qmfhe.\CEnc(\pk^\Circle_{i,1-t^\Circle_i},0),$$ and for each $i \in [2m]$, compute $$\ct^\inp_{i,1-t^\inp_i} \gets \qmfhe.\CEnc(\pk^\inp_{i,1-t^\inp_i}, 0).$$ Note that the corresponding secret keys are not needed to produce the distribution conditioned on the set of teleportation errors being $\{x_i^\Circle,z_i^\Circle,x_i^\inp,z_i^\inp\}_{i \in [n]}$, so indistinguishability of $\cH_0$ and $\cH_1$ reduces to semantic security of $\qmfhe$.
    \item $\cH_2$: \underline{Simulate the two classical garbled circuits}. 
    \begin{itemize}
        \item During Offline Round 2, the classical functionality computes $$U_\test \coloneqq f_{\test}[\{C_i^\Circle\}_{i \in [n]},r,U_\enc,\{C_i^T\}_{i \in [n]},M,\pi](\{x_i^\Circle,z_i^\Circle\}_{i \in [n]})$$ and then $$\left(\{\widetilde{\lab^\test_i}\}_{i \in [2nv]},\widetilde{f}_{\test}\right) \gets \GSim\left(1^\secplev,1^{2nv},1^{|f_{\test}|},U_{\test}\right).$$
        \item During Online Round 1, the classical functionality computes $$(U_{\mathsf{garble}},D_0,\widetilde{g}_1,\dots,\widetilde{g}_d) \coloneqq f_\QGC[U_\enc,\{C_i^\inp,C_i^\out\}_{i \in [n]}](\{x_i^\inp,z_i^\inp\}_{i \in [n]})$$ and then $$(\{\widetilde{\lab^\QGC_i}\}_{i \in [2nv]},\widetilde{f}_{\QGC}) \gets \GSim(1^\secplev,1^{2m},1^{|f_{\QGC}|},(U_{\mathsf{garble}},D_0,\widetilde{g}_1,\dots,\widetilde{g}_d)).$$
    \end{itemize}
    Indistinguishability of $\cH_1$ and $\cH_2$ reduces to the security of the garbling scheme.
    \item $\cH_3$: \underline{Switch half of QMFHE ciphertexts to encryptions of simulated labels}. For each $i \in [2nv]$, compute $$\ct^\Circle_{i,1-t^\Circle_i} \gets \qmfhe.\CEnc(\pk^\Circle_{i,1-t^\Circle_i},\widetilde{\lab^\test_i}),$$ and for each $i \in [2m]$, compute $$\ct^\inp_{i,1-t^\inp_i} \gets \qmfhe.\CEnc(\pk^\inp_{i,1-t^\inp_i}, \widetilde{\lab^\QGC_i}).$$ Indistinguishability of $\cH_2$ and $\cH_3$ reduces to semantic security of $\qmfhe$.
    \item $\cH_4$: \underline{Re-define $C^{\Circle}_k$}. First, sample $U_\test \gets \mathscr{C}_v$ and then define \begin{align*}C^{\Circle}_k \coloneqq &\left(X^{x_n^\Circle} Z^{z_n^\Circle}C^{\Circle}_n\dots X^{x_{k+1}^\Circle} Z^{z_{k+1}^\Circle} C^{\Circle}_{k+1}X^{x^\Circle_k}Z^{z^\Circle_k}\right)^\dagger U_\test^\dagger (U_\enc \ \otimes \ \bbI)(\bbI \ \otimes \ U_\chck)\\&\left(X^{x_{k-1}^\Circle}Z^{z_{k-1}^\Circle}C^{\Circle}_{k-1}\dots X^{x_{1}^\Circle}Z^{z_{1}^\Circle}C^{\Circle}_1\right)^\dagger.\end{align*}
    Note that when $C^{\Circle}_k$ is plugged into $U_\mathsf{dec}$ (see Online Round 1 of \proref{fig:classical}), we have that $$\left(U_{\mathsf{enc}} \otimes \bbI\right)\left(\bbI \otimes U_{\mathsf{check}}\right)U_{\mathsf{dec}} \coloneqq U_\test.$$ This switch is perfectly indistinguishable given the fact that in $\cH_3$, $C_k^\Circle$ is a uniformly random Clifford, and in $\cH_4$, $U_\test$ is a uniformly random Clifford. 
    \item $\cH_5$: \underline{Re-define $U_\enc$}. First, sample $U_{\mathsf{garble}} \gets \mathscr{C}_{m+\secp n + k_0 + k_T}$, then compute $(E_0,D_0,\widetilde{g}_1,\dots,\widetilde{g}_d) \gets \QGarble(1^\secp,Q[\dst,\{C_i^\inp,C_i^\out\}_{i \in [n]}])$, and set $U_\enc \coloneqq U_{\mathsf{garble}}^\dagger E_0\left(X^{x_1^\inp}Z^{z_1^\inp} \otimes \dots \otimes X^{x_n^\inp}Z^{z_n^\inp} \otimes \bbI\right)$. Note that (see Online Round 2 of \proref{fig:classical}) $E_0\left(X^{x_1^\inp}Z^{z_1^\inp} \otimes \dots \otimes X^{x_n^\inp}Z^{z_n^\inp} \otimes \bbI\right)U_\enc^\dagger \coloneqq U_{\mathsf{garble}}$. Moreover, we can now write $C^{\Circle}_k$ as \begin{align*}C^{\Circle}_k \coloneqq &\left(X^{x^\Circle_k}Z^{z^\Circle_k}{C^{\Circle}_{k+1}}^\dagger X^{x_{k+1}^\Circle} Z^{z_{k+1}^\Circle} \dots {C^{\Circle}_n}^\dagger X^{x_n^\Circle} Z^{z_n^\Circle} \right) U_\test^\dagger \left(U_{\mathsf{garble}}^\dagger \otimes \bbI \right)\left(E_0 \otimes \bbI \right)\\&\left(X^{x_1^\inp}Z^{z_1^\inp} \otimes \dots \otimes X^{x_n^\inp}Z^{z_n^\inp} \otimes \bbI\right)(\bbI \ \otimes \ U_\chck)\left({C_1^{\Circle}}^\dagger X^{x_{1}^\Circle}Z^{z_{1}^\Circle} \dots {C_{k-1}^{\Circle}}^\dagger X^{x_{k-1}^\Circle}Z^{z_{k-1}^\Circle} \right).\end{align*} This switch is perfectly indistinguishable given the fact that in $\cH_4$, $U_\enc$ is a uniformly random Clifford and in $\cH_5$, $U_{\mathsf{garble}}$ is a uniformly random Clifford.
    \item $\cH_6$: \underline{Simulate the quantum garbled circuit}. Rather than directly apply $C^{\Circle}_k$ as described above during Offline Round 2, the simulator will do the following.
    \begin{itemize}
        \item Apply the first part $$\left(X^{x_1^\inp}Z^{z_1^\inp} \otimes \dots \otimes X^{x_n^\inp}Z^{z_n^\inp} \otimes \bbI\right)(\bbI \ \otimes \ U_\chck)\left({C_1^{\Circle}}^\dagger X^{x_{1}^\Circle}Z^{z_{1}^\Circle} \dots {C_{k-1}^{\Circle}}^\dagger X^{x_{k-1}^\Circle}Z^{z_{k-1}^\Circle} \right)$$ to obtain state $\left(\bx^{\gray{N},\gray{Z_\inp},\gray{T_\inp}},\bz^{\gray{Z_\test},\gray{T_{\test,1}},\dots,\gray{T_{\test,n}}}\right)$.
        \item Compute the circuit $Q_\dst[\{C_i^\inp,C_i^\out\}_{i \in [n]}]$ on $\bx$ to obtain output $\by$.
        \item Compute $(\widetilde{\bx}^{\gray{N},\gray{Z_\inp},\gray{T_\inp}},D_0,\widetilde{g}_1,\dots,\widetilde{g}_d) \gets \QGSim\left(1^\secplev,Q_\dst,\by\right)$.
        \item Apply the final part $$\left(X^{x^\Circle_k}Z^{z^\Circle_k}{C_{k+1}^{\Circle}}^\dagger X^{x_{k+1}^\Circle} Z^{z_{k+1}^\Circle} \dots {C_n^{\Circle}}^\dagger X^{x_n^\Circle} Z^{z_n^\Circle} \right) U_\test^\dagger \left(U_{\mathsf{garble}}^\dagger \otimes \bbI \right)$$ to state $\left(\widetilde{\bx},\bz\right)$.
    \end{itemize}
    Note that $(D_0,\widetilde{g}_1,\dots,\widetilde{g}_d)$ are later used to simulate $\widetilde{f}_\QGC$. Indistinguishability of $\cH_5$ and $\cH_6$ follows from the security of the quantum garbled circuit.
    
    \item $\cH_7$: \underline{Move the teleportation errors $\{x_i^\Circle,z_i^\Circle\}_{i \in [n]}$}. 
    
    Let $\widehat{x}_1, \widehat{z}_1$ be such that $$(\bbI \ \otimes \ U_\chck){C_1^\Circle}^\dagger X^{x_1^\Circle}Z^{z_1^\Circle} \cdots {C_{k-1}^\Circle}^\dagger X^{x_{k-1}^\Circle}Z^{z_{k-1}^\Circle} = X^{\widehat{x}_1}Z^{\widehat{z}_1}(\bbI \ \otimes \ U_\chck){C_1^\Circle}^\dagger \cdots {C_{k-1}^\Circle}^\dagger.$$ Write $\widehat{x}_1$ as $\widehat{x}_{\inp,1},\dots,\widehat{x}_{\inp,n},\widehat{x}_z,\widehat{x}_t,\widehat{x}_{\test}$, and same for $\widehat{z}_1$. Here, each $\widehat{x}_{\inp,i} \in \{0,1\}^{m_i + \secp}$, $\widehat{x}_z \in \{0,1\}^{k_0}$, $\widehat{x}_t \in \{0,1\}^{k_T}$, and $\widehat{x}_{\test} \in \{0,1\}^{k_0 + 3 \secp n}$. 
    
    Let $\widehat{x}_2, \widehat{z}_2$ be such that $$\left(\bbI \otimes X^{\widehat{x}_\test} Z^{\widehat{z}_\test}\right)U_\test X^{x_{n}^\Circle}Z^{z_{n}^\Circle}C^{\Circle}_{n}\cdots X^{x_{k+1}^\Circle}Z^{z_{k+1}^\Circle}C^{\Circle}_{k+1}X^{x_k^\Circle}Z^{z_k^\Circle} =  X^{\widehat{x}_2}Z^{\widehat{z}_2} U_\test C^{\Circle}_n\cdots C^{\Circle}_{k+1}.$$ Write $\widehat{x}_2$ as $\widehat{x}_{\mathsf{garble}},\widehat{x}_{\test,z},\widehat{x}_{\test,T,1},\dots,\widehat{x}_{\test,T,n}$, and same for $\widehat{z}_2$. Here, $\widehat{x}_{\mathsf{garble}} \in \{0,1\}^{m + \secp n + k_0 + k_T}$, $\widehat{x}_{\test,Z} \in \{0,1\}^{k_0 + \secp n}$, and each $\widehat{x}_{\test,T,i} \in \{0,1\}^{2\secp}$.

    Now, the first part of $C^{\Circle}_k$ will be computed as $$\left(X^{x_1^\inp}Z^{z_1^\inp}X^{\widehat{x}_{\inp,1}} Z^{\widehat{z}_{\inp,1}} \otimes \dots \otimes X^{x_n^\inp}Z^{z_n^\inp}X^{\widehat{x}_{\inp,n}} Z^{\widehat{z}_{\inp,n}} \otimes X^{\widehat{x}_t}Z^{\widehat{z}_t}\otimes X^{\widehat{x}_z}Z^{\widehat{z}_z} \otimes \bbI\right)(\bbI \ \otimes \ U_\chck){C_1^\Circle}^\dagger \cdots {C_{k-1}^\Circle}^\dagger,$$
    
    and the final part of $C_k$ will be computed as $${C_{k+1}^{\Circle}}^\dagger \dots {C_n^{\Circle}}^\dagger U_\test^\dagger \left(U_{\mathsf{garble}}^\dagger \otimes \bbI \right).$$
    Moreover, define 
    \begin{itemize}
        \item $\widehat{U}_{\mathsf{garble}} \coloneqq U_{\mathsf{garble}}X^{\widehat{x}_{\mathsf{garble}}}Z^{\widehat{z}_{\mathsf{garble}}}$, and use $\widehat{U}_{\mathsf{garble}}$ in place of $U_{\mathsf{garble}}$ when simulating $\widetilde{f}_\QGC$,
        \item $\widehat{r} \coloneqq r \oplus \widehat{x}_{\test,z}$, and use $\widehat{r}$ in place of $r$ in Online Round 2,
        \item for each $i \in [n]$, $\widehat{C}_i^T \coloneqq C_i^T X^{\widehat{x}_{\test,T,i}}Z^{\widehat{z}_{\test,T,i}}$ and use $\widehat{C}_i^T$ in place of $C_i^T$ in Online Round 1.
    \end{itemize}
    This switch is perfectly indistinguishable. In particular, $U_{\mathsf{garble}}$ and the $C_i^T$ are uniformly random Cliffords, so adversary will not notice the switch to $\widehat{U}_{\mathsf{garble}}$ and $\widehat{C}_i^T$. Also the $Z$-test positions are randomized with $X^r Z^s$ during computation of $U_{\chck}$, where $r$ and $s$ are uniformly random, so adversary will not notice the switch to $\widehat{r}$.
    \item $\cH_8$: \underline{Introduce new Pauli errors.} For each $i \in [n]$, sample $x_i^\out,z_i^\out \gets \{0,1\}^{\ell_i + \secp}$. Now, rather than running $\QGSim$ directly on $\by$, run it on $X^{x_1^\out}Z^{z_1^\out}\dots X^{x_n^\out}Z^{z_n^\out}(\by)$. Then, for each $i \in [n]$, define $\widehat{C}_i^\out \coloneqq C_i^\out X^{x_i^\out}Z^{z_i^\out}$, and use $\widehat{C}_i^\out$ in place of $C_i^\out$ in Online Round 2.
    This switch is perfectly indistinguishable given that $\{C_i^\out\}_{i \in [n]}$ are uniformly random.
    \item $\cH_9$: \underline{Introduce new EPR pairs.} Prepare $\be_{\Sim,1},\be_{\Sim,2}$ as described in the simulator. Let $\{x_i^\out,z_i^\out\}_{i \in [n]}$ now be the result of teleportation errors obtained across these EPR pairs, as described in the simulator. Now, we can delay the computation of $$\left(X^{x_1^\inp}Z^{z_1^\inp}X^{\widehat{x}_{\inp,1}} Z^{\widehat{z}_{\inp,1}} \otimes \dots \otimes X^{x_n^\inp}Z^{z_n^\inp}X^{\widehat{x}_{\inp,n}} Z^{\widehat{z}_{\inp,n}} \otimes X^{\widehat{x}_t}Z^{\widehat{z}_t}\otimes X^{\widehat{x}_z}Z^{\widehat{z}_z} \otimes \bbI\right)$$ and $Q_\dst[\{C_i^\inp,C_i^\out\}_{i \in [n]}]$ on state $(\bn^{\gray{N}},\bz^{\gray{Z_\inp}},\bt^{\gray{T_\inp}})$ to Online Round 2. This distribution is identical to the previous one.
    \item $\cH_{10}$: \underline{Switch honest party inputs to $\mathbf{0}$}. For each $i \in \cH$, in Online Round 1, teleport $C_i^\inp(\mathbf{0}^{m_i + \secp})$ to Party 1, rather than $C_i^\inp(\bx_i,\mathbf{0}^{\secp})$. Now, during Online Round 2, instead of directly applying $Q_\dst[\{C_i^\inp,C_i^\out\}_{i \in [n]}]$ to $(\bn,\bz,\bt)$, this hybrid will do the following. First apply ${C_1^\inp}^\dagger \otimes \dots \otimes {C_n^\inp}^\dagger$ to $\bn$. Then, swap out the honest party input registers for $\{\bx_i\}_{i \in \cH}$, and continue with the computation of $Q_\dst[\{C_i^\inp,C_i^\out\}_{i \in [n]}]$ on $(\bn,\bz,\bt)$. $\cH_9$ is statistically indistinguishable from $\cH_{10}$ due to properties of the Clifford authentication code. In particular, since the code is perfectly hiding, the adversary cannot tell that the inputs where switched to $\mathbf{0}$. Thus the adversary can only distinguish if the output of $Q_\dst[\{C_i^\inp,C_i^\out\}_{i \in [n]}]$ differs between $\cH_9$ and $\cH_{10}$. However, if any Clifford authentication test that happens within $Q_\dst[\{C_i^\inp,C_i^\out\}_{i \in [n]}]$ fails, then the output is $(\bot \dots \bot)$. In both $\cH_9$ and $\cH_{10}$, conditioned on these tests passing, the honest party inputs to $Q_\dst$ are statistically close to $\{\bx_i\}_{i \in \cH}$, due to the authentication property of the Clifford code.
    \item $\cH_{11}$: \underline{Query ideal functionality}. Consider the EPR pairs halves $\be_S$ sent from party $k$ to party $k-1$ in Offline Round 1. In this hybrid, we alter the computation that is performed on $\be_S$ in Offline Round 2 and Online Round 2. First, in Offline Round 2, this hybrid computes $$(\bn,\bz,\bt,\bz_\test,\bt_{\test,2},\dots,\bt_{\test,n}) \coloneqq (\bbI \ \otimes \ U_\chck){C_1^\Circle}^\dagger \cdots {C_{k-1}^\Circle}^\dagger (\be_S).$$ Then, $\bz, \bt$ are discarded, and $\bz_\test,\bt_{\test,2},\dots,\bt_{\test,n}$ are operated on as in $\cH_{10}$. Then, in Online Round 2, this hybrid does the following. 
    \begin{itemize}
        \item Parse $\bn \coloneqq (\bn_1,\dots,\bn_n)$. For each $i \in [n]$, compute $(\bx_i,\bz_i) \coloneqq C_i^\inp X_i^\inp Z_i^\inp X^{\widehat{x}_{\inp,i}} Z^{\widehat{z}_{\inp,i}}(\bn_i)$ and measure $\bz_i$. If any measurements are not all 0, then set $(\by_1,\dots,\by_n) \coloneqq (\bot,\cdots,\bot)$. Otherwise, query the ideal functionality with $\{\bx_i\}_{i \in [n] \setminus \{\cH\}}$ and let $\{\by_i\}_{i \in [n] \setminus \{\cH\}}$ be the output. Set $\by_i \coloneqq \mathbf{0}^{\ell_i}$ for each $i \in \cH$.
        \item Sample $\{C_i^\out \gets \mathscr{C}_{\ell_i + \secp}\}_{i \in [n]}$.
        \item Teleport $(C_1^\out(\by_1,\mathbf{0}^\secp),\cdots,C_n^\out(\by_n,\mathbf{0}^\secp))$ into $\be_{\Sim,1}$.
    \end{itemize}
    We also add the following behavior to Output Reconstruction.
    \begin{itemize}
        \item For each $i \in \cH$, decrypt $\ct_{y,i}$ using $\sk_\QGC$ to obtain $\by_i^\out$. Then, compute $\widehat{C}_i^\out(\by_i^\out)$ and measure whether the last $\secp$ trap qubits are all 0. If not, then send $\abort_i$ to the ideal functionality and otherwise send $\mathsf{ok}_i$.
    \end{itemize}
    Observe that one difference between $\cH_{10}$ and $\cH_{11}$ is that $\bz$ and $\bt$ are not used in the computation of $\CM$ circuit $Q_\dst$. Instead, the ideal functionality directly computes $Q$ on the inputs. This will result in statistically close outputs if i) the QGC satisfies statistical correctness, ii) $\bz$ is statistically close to $\mathbf{0}^{\otimes k_0}$, and iii) the result of applying the distillation circuit to $\bt$ is statistically close to $\mathbf{T}^{\otimes k_T / \secp}$. \cref{lemma:Ztest} implies that the second requirement holds conditioned on the adversary submitting the correct $r$ to the classical $\MPC$ in Online Round 1 (and otherwise, all honest parties abort). \cref{lemma:distillation} plus Clifford authentication implies that the third requirement holds conditioned on the honest party $T$-state checks in Online Round 2 all passing (and if any one of them fails, all honest parties abort).
    
    The other difference is that honest party outputs are determined by the ideal functionality's computation. First, the adversary cannot tell that the honest party outputs are switched to $\mathbf{0}$ within the quantum garbled circuit, by perfect hiding of the Clifford authentication code (using Cliffords $\{C_i^\out\}_{i \in \cH}$). Next, in $\cH_{10}$, the adversary cannot make an honest party accept a state noticeably far from their real output $\by_i$, by authentication of the Clifford code. Thus, $\cH_{11}$ is statistically close to $\cH_{10}$. This completes the proof, as $\cH_{11}$ is the simulator described above.
\end{itemize}

\end{proof}

\begin{lemma}
Let $\Pi$ be the protocol described in \cref{subsec:4-round-protocol} computing some quantum circuit $Q$. Then $\Pi$ satisfies \cref{def:mpqc} for any $\cA$ corrupting parties $M \subset [n]$ where $1 \in M$.

\end{lemma}

\begin{proof}
Assume towards contradiction the existence of a QPT $\cD = \{\cD_\secp\}_{\secp \in \bbN}$, a QPT $\cA = \{\cA_\secp\}_{\secp \in \bbN}$, and $(\bx_1,\dots,\bx_n,\baux_\cA,\baux_\cD)$ such that 
\begin{align*}
&\bigg|\Pr\left[\cD_\secp\left(\baux_\cD,\Real_{\Pi,\Q}\left(\cA_\secp,\{\bx_i\}_{i \in [n]},\baux_\cA\right)\right) = 1\right]\\ &- \Pr\left[\cD_\secp\left(\baux_\cD,\Ideal_{\Pi,\Q}\left(\Sim_\secp,\{\bx_i\}_{i \in [n]},\baux_\cA\right)\right) = 1 \right] \bigg| \geq 1/\poly(\secp).
\end{align*}
Define $\Real \coloneqq \Real_{\Pi,\Q}(\cA_\secp,\{\bx_i\}_{i \in [n]},\baux_\cA)$ and $\Ideal \coloneqq \Ideal_{\Pi,\Q}(\Sim_\secp,\{\bx_i\}_{i \in [n]},\baux_\cA)$. Furthermore, let $\mathbf{E}^{\{x_i^\Circle,z_i^\Circle,x_i^\inp,z_i^\inp\}_{i \in [n]}}_\Real$ be the event that the parties report teleportation errors $\{x_i^\Circle,z_i^\Circle,x_i^\inp,z_i^\inp\}_{i \in [n]}$ in $\Real$ and define $\mathbf{E}^{\{x_i^\Circle,z_i^\Circle,x_i^\inp,z_i^\inp\}_{i \in [n]}}_\Ideal$ analogously. Let $\mathbf{E}^{(\abort)}_\Real$ and $\mathbf{E}^{(\abort)}_\Ideal$ be the event that the adversary fails to report some of its teleporation errors, causing the honest parties to abort. The above implies that either there exists some $\{x_i^\Circle,z_i^\Circle,x_i^\inp,z_i^\inp\}_{i \in [n]}$ such that 
\begin{align*}
&\bigg|\Pr\left[\cD_\secp(\baux_\cD,\Real) = 1 \big| \mathsf{E}_\Real^{\{x_i^\Circle,z_i^\Circle,x_i^\inp,z_i^\inp\}_{i \in [n]}}\right]\Pr\left[\mathsf{E}_\Real^{(x_\inp,z_\inp)}\right] \\&- \Pr\left[\cD_\secp(\baux_\cD,\Ideal) = 1 \big| \mathsf{E}_\Ideal^{\{x_i^\Circle,z_i^\Circle,x_i^\inp,z_i^\inp\}_{i \in [n]}} \right]\Pr\left[\mathsf{E}_\Ideal^{(x_\inp,z_\inp)}\right] \bigg| \geq \frac{1}{\poly(\secp)(2^{2nv + 2m}+1)}
\end{align*}
or that 
\begin{align*}
&\bigg|\Pr\left[\cD_\secp(\baux_\cD,\Real) = 1 \big| \mathsf{E}_\Real^{(\abort)}\right]\Pr\left[\mathsf{E}_\Real^{(\abort)}\right] \\&- \Pr\left[\cD_\secp(\baux_\cD,\Ideal) = 1 \big| \mathsf{E}_\Ideal^{(\abort)} \right]\Pr\left[\mathsf{E}_\Ideal^{(\abort)}\right] \bigg| \geq \frac{1}{\poly(\secp)(2^{2nv + 2m}+1)}.
\end{align*}

Simulating the distribution conditioned on an abort is trivial, so the second case cannot occur, and the first case immediately contradicts~\cref{lemma:distribution-bot-mpc}, completing the proof.
\end{proof}

\subsubsection{Case 2: $P_1$ is the only honest party}

\paragraph{Simulator.} The simulator will act as party 1 and maintain the classical $\MPC$ oracle. It will compute $\MPC$ honestly throughout, and will compute honest party 1 actions throughout except for what is described below.

\begin{itemize}
    \item \textbf{Online Round 1.} Rather than Clifford-encoding and teleporting in $P_1$'s input $\bx_1$, the simulator will teleport $C_1^\inp(\mathbf{0}^{\ell_1 + \secp})$.
    \item \textbf{Online Round 2.} Rather than homomorphically evaluating the quantum garbled circuit $(U_{\mathsf{garble}},D_0,\widetilde{g}_1,\dots,\widetilde{g}_d)$ on its encrypted state $\by$, the simulator will run the following computation (homomorphically) on $\by$.
    \begin{itemize}
        \item Compute $(\bn_1,\dots,\bn_n,\bz,\bt) \coloneqq E_0^\dagger U_{\mathsf{garble}}(\by)$, where $(\bn_1,\dots,\bn_n)$ are the parties' Clifford-encoded inputs. 
        \item For each $i \in [n]$, compute $(\bx_i,\bz_i) \coloneqq {C_i^\inp}^\dagger (\bn_i)$ and measure $\bz_i$. If any measurements are not all 0, then set $(\by_1,\dots,\by_n) \coloneqq (\bot,\cdots,\bot)$. Otherwise, query the ideal functionality with $\{\bx_i\}_{i \in [2,\dots,n]}$ and let $\{\by_i\}_{i \in [2,\dots,n]}$ be the output.
        \item For each $i \in [2,\dots,n]$ set $\by_i^\out \coloneqq C_i^\out(\by_i,\mathbf{0}^\secp)$, set $\by_1^\out \coloneqq C_1^\out(\mathbf{0}^{\ell_1+\secp})$, and continue as the honest party 1.
    \end{itemize}
\end{itemize}

\begin{lemma}
Let $\Pi$ be the protocol described in \cref{subsec:4-round-protocol} computing some quantum circuit $Q$. Then $\Pi$ satisfies \cref{def:mpqc} for any $\cA$ corrupting parties $[2,\dots,n]$.

\end{lemma}

\begin{proof}
We consider a sequence of hybrid distributions, where $\cH_0$ is $\Real_{\Pi,\Q}(\cA_\secp,\{\bx_i\}_{i \in [n]},\baux_\cA)$, i.e. the real interaction between $\cA_\secp(\{\bx_i\}_{i \in [2,\dots,n]},\baux_\cA)$ and an honest party $P_1(1^\secp,\bx_1)$. In each hybrid, we describe the differences from the previous hybrids.

\begin{itemize}
    \item $\cH_1:$ \underline{Directly compute $Q_\dst[\{C_i^\inp,C_i^\out\}]$ in place of garbled circuit evaluation} During Online Round 2, this hybrid computes $(\bn_1,\dots,\bn_n,\bz,\bt) \coloneqq E_0^\dagger U_{\mathsf{garble}}(\by)$ and then applies $Q_\dst[\{C_i^\inp,C_i^\out\}]$ to produce outputs $(\by_1^\out,\dots,\by_n^\out)$. Statistical indistinguishability follows from the statistical correctness of the QGC, and statistical ciphertext re-randomization of $\qmfhe$.
    \item $\cH_2:$ \underline{Replace $P_1$'s input with $\mathbf{0}$} During Online Round 1, this hybrid teleports in $C_1^\inp(\mathbf{0}^{m_i + \secp})$. Then, during Online Round 2, this hybrid inserts $P_1$'s input $\bx_1$ before the computation of $Q$. This switch is perfectly indistinguishable due to the perfect hiding of the Clifford code.
    \item $\cH_3:$ \underline{Query ideal functionality} During Online Round 2, this hybrid computes $Q_\dst[\{C_i^\inp,C_i^\out\}]$ as described in the simulator, by using the ideal functionality to compute $Q$. This switch is statistically indistinguishable as long as i) $\bz$ is statistically close to $\mathbf{0}^{\otimes k_0}$ (which follows from \cref{lemma:Ztest}), and ii) the result of applying the distillation circuit to $\bt$ to statistically close to $\mathbf{T}^{\otimes k_T / \secp}$ (which follows from \cref{lemma:distillation}, as $P_1$ checks its own subset of T state). This hybrid is the simulator, completing the proof.
\end{itemize}

\end{proof}

\section{Two Rounds Without Pre-Processing: Challenges and Possibilities}

\subsection{An Oblivious Simulation Barrier for Two Round Protocols}\label{subsec:two-round-impossibility}

We begin with our negative result showing that any two-round 2PQC protocol with an \emph{oblivious simulator} supporting general quantum functionalities would imply new protocols for the setting of \emph{instantaneous non-local quantum computation}~\cite{PhysRevLett.90.010402,Beigi_2011,DBLP:conf/tqc/Speelman16,DBLP:journals/tit/GonzalesC20}.

\paragraph{Instantaneous Non-local Quantum Computation.} Instantaneous non-local quantum computation  of a unitary $U$ on $n_A + n_B$ qubits is an information-theoretic task where parties $A$ and $B$, who may share some initial entangled quantum state, receive as input quantum states $\bx_A,\bx_B$ and wish to compute the functionality $U(\bx_A,\bx_B) = (\by_A,\by_B)$ with only one round of simultaneous communication. For a family of unitaries $\{U_\secp\}_{\secp \in \bbN}$ on $\{n_{A,\secp} + n_{B,\secp}\}_{\secp \in \bbN}$ qubits, we say that an instantaneous non-local quantum computation protocol must satisfy the following properties:

\begin{itemize}
    \item \textbf{Correctness.} For all input states $(\bx_{A,\secp},\bx_{B,\secp})$, the joint outputs $(\by_{A,\secp}',\by_{B,\secp}')$ obtained by $A$ and $B$ after engaging in the protocol are such that $(\by_{A,\secp}',\by_{B,\secp}') \approx_s (\by_{A,\secp},\by_{B,\secp})$, where $(\by_{A,\secp},\by_{B,\secp}) \coloneqq U_\secp(\bx_{A,\secp},\bx_{B,\secp})$.
    \item \textbf{Efficiency.} The size of the entangled quantum state initially shared by $A$ and $B$ in the protocol for computing $U_\secp$ is bounded by some polynomial in $\secp$ (note that the running time of $A$ or $B$ in the protocol does not need to be polynomial).
\end{itemize}

\begin{conjecture}
\label{conjecture: instantaneous}
There exists a family of efficiently-computable unitaries $\{U_\secp\}_{\secp \in \bbN}$ for which no \emph{correct} and \emph{efficient} instantaneous non-local quantum computation protocol exists.
\end{conjecture}

As noted in the introduction, the best known instantaneous non-local quantum computation protocols for general functionalities on $n$-qubit inputs for $n>2$ require exponentially many EPR pairs in both $n$ and in $\log(1/\epsilon)$, where $\epsilon$ is the desired correctness error \cite{Beigi_2011}. Moreover, there has been recent progress on proving lower bounds for particular classes of unitaries~\cite{DBLP:journals/tit/GonzalesC20}. While current lower bounds on the size of input-independent pre-processing are linear in the number of input qubits, the current state of the art leaves open the possibility that exponentially-many EPR pairs are necessary for general functionalities. To the best of the authors' knowledge, known results give no indication as to whether \cref{conjecture: instantaneous} is more likely to be true or false. Nevertheless, the fact that it remains unresolved provides some indication that positive progress on two-round 2PQC with oblivious simulation will require new ideas.






\paragraph{Two-Round 2PQC in the CRS Model.} Consider a generic two-round two-party protocol for computing an arbitrary functionality $U$ in the (classical) CRS model assuming simultaneous messages. Such a protocol is described by the algorithms $(A_1,A_2,A_3,B_1,B_2,B_3)$ where $A_1,A_2,A_3$ are (respectively) Alice's first message algorithm, second message algorithm, and output reconstruction algorithm (and likewise for Bob with $B_1,B_2,B_3$). As usual, Alice's input is $\bx_A$ and Bob's input is $\bx_B$. They compute a unitary $U$ and obtain $U(\bx_A,\bx_B) = (\by_A,\by_B)$ where $\by_A$ and $\by_B$ are their respective outputs. We stress that since this model does not allow pre-processing, Alice and Bob \emph{may not share entanglement} before receiving their inputs.

An execution of such a two-round protocol proceeds as follows:
\begin{enumerate}
    \item \textbf{Setup.} Run $\crs \gets \Gen$.
    \item \textbf{Round 1.} Alice and Bob generate their first round messages and leftover states as $(\bm_1^{(A)},\bst_1^{(A)}) \gets A_1(\crs,\bx_A)$ and $(\bm_1^{(B)},\bst_1^{(B)}) \gets  B_1(\crs,\bx_B)$. They send their messages to each other, which has the effect of interchanging/swapping $\bm_1^{(A)}$ and $\bm_1^{(B)}$.
    \item \textbf{Round 2.} Alice and Bob generate their second round message and leftover states as $(\bm_2^{(A)},\bst_2^{(A)}) \gets  A_2(\bst_1^{(A)},\bm_1^{(B)})$ and $(\bm_2^{(B)},\bst_2^{(B)}) \gets  B_2(\bst_1^{(B)},\bm_1^{(A)})$.
    They send their messages to each other, which swaps $\bm_2^{(A)}$ and $\bm_2^{(B)}$.
    \item \textbf{Output.} $\by_A \gets A_3(\bst_2^{(A)},\bm_2^{(B)})$
    and $\by_B \gets B_3(\bst_2^{(B)},\bm_2^{(A)})$.
\end{enumerate}

\paragraph{Oblivious Simulation.} We now define a natural class of black-box, straight-line simulators that we call \emph{oblivious} simulators. Recall that a simulator is \emph{black-box} if it only makes query access to the attacker (and does not need the code/state of the attacker), and is \emph{straight-line} if it only runs a single time in the forward direction. The defining property of an oblivious simulator is that it learns which player (out of $A$ or $B$) is corrupted only {\em after} it has generated (and committed to) a simulated CRS. No matter which party is corrupted, such a simulator must use its committed CRS to generate a view for the corrupt party that is computationally indistinguishable from the party's view in the real world.

As discussed in~\cref{subsec:tech-overview-two-round}, a negative result for oblivious simulation demonstrates that a natural strategy for constructing two-round two-party computation in the \emph{classical} setting does not extend to the quantum setting.

The following definition specifies the \emph{additional requirements} for a simulator to be ``oblivious''; an oblivious simulator must still satisfy the standard real/ideal indistinguishability notion in~\cref{def:mpqc}, which we will not repeat here.

\begin{definition}[Syntactic Requirements for Oblivious Simulation]
A simulator for a two-round two-party quantum computation protocol in the classical CRS model is \emph{oblivious} if it can be described by a tuple of algorithms $(\Sim_0,\Sim^{(A)},\Sim^{(B)})$ where $\Sim^{(A)} = (\Sim_1^{(A)},\Sim_2^{(A)},\Sim_3^{(A)})$ and $\Sim^{(B)} = (\Sim_1^{(B)},\Sim_2^{(B)},\Sim_3^{(B)})$, simulation proceeds as follows.
\begin{enumerate}
\item The simulator runs $(\crs,\bst_0^{(\Sim)}) \gets \Sim_0$ to generate the CRS and leftover simulator state $\bst_0^{(\Sim)}$.
\end{enumerate}
Next, the simulator ``learns'' whether it should simulate the view of party $A$ or party $B$. If the simulator is simulating the view of party $A$, it proceeds using $\Sim^{(A)} = (\Sim_1^{(A)},\Sim_2^{(A)},\Sim_3^{(A)})$, and if it is simulating the view of party $B$, it proceeds with $\Sim^{(B)} = (\Sim_1^{(B)},\Sim_2^{(B)},\Sim_3^{(B)})$. We  write out the case for simulating the view of party $A$ below (the case for party $B$ is identical).
\begin{enumerate}
\setcounter{enumi}{1}
\item $(\bm_1^{(B)},\bst_1^{(\Sim)}) \gets \Sim_1^{(A)}(\bst_0^{(\Sim)})$\\ Then query $A_1$ on $(\crs,\bm_1^{(B)})$ and receive $\bm_1^{(A)}$.
\item $(\bx_A,\bst_2^{(\Sim)}) \gets \Sim_2^{(A)}(\bst_1^{(\Sim)},\bm_1^{(A)})$
\\ Then query the ideal functionality on $\bx_A$ and receive $\by_A$
\item $\bm_2^{(B)} \gets \Sim_3^{(A)}(\bst_2^{(\Sim)},\by_A)$
\\ Then query $A_2$ on $\bm_2^{(B)}$.
\end{enumerate}

\end{definition}

In short, for an oblivious simulator, the distribution of the simulated CRS is completely independent of whether $A$ is corrupt or $B$ is corrupt. Moreover, because the simulator is straight-line, it is possible to define a (possibly inefficient) algorithm $\Sim_{\comb}$ that computes $\left(\crs,\st_1^{(\Sim,A)},\st_1^{(\Sim,B)}\right),$ where each of the two simulator states computed is with respect to the {\em same} classical $\crs$. This can be done for example by running many iterations of $\Sim_0$ until two of them output the same classical $\crs$. Thus, one would obtain a ``first-round-only'' simulator with the following syntax for the first round:

\begin{enumerate}
    \item $\left(\crs,\st_1^{(\Sim,A)},\st_1^{(\Sim,B)}\right) \gets \Sim_{\comb}$.
    \\ {\em (send $\crs$ to $A_1$ and receive $\bm_1^{(A)}$ and send $\crs$ to $B_1$ and receive $\bm_1^{(B)}$)}
    \item $(\bx_A,\st_2^{(\Sim,A)}) \gets \Sim_2^{(A)}(\st_1^{(\Sim,A)},\bm_1^{(A)})$.
        \\ Then send $\bx_A$ to the ideal functionality and receive $\by_A$
    \item $(\bx_B,\st_2^{(\Sim,B)}) \gets \Sim_2^{(B)}(\st_1^{(\Sim,B)},\bm_1^{(B)})$.
        \\ Then send $\bx_B$ to the ideal functionality and receive $\by_B$
\end{enumerate}

\paragraph{Non-Local Computation from Two-Round 2PQC with Oblivious Simulation.} We now describe how to turn two-round 2PQC for general functionalities with oblivious simulation (in the CRS model) into an instantaneous non-local quantum communication protocol. In the following theorem, we will only make use of the ``first-round-only'' simulator discussed above.

\begin{theorem}
\label{thm:lower-bound}
Assuming \cref{conjecture: instantaneous}, there does not exist a two-round two-party quantum computation protocol for general functionalities in the classical CRS model, with an oblivious simulator.\footnote{A previous version of this work incorrectly claimed an unconditional version of this theorem. We thank Mi-Ying Huang as well as anonymous reviewers for pointing this issue out to us.}
\end{theorem}

\begin{proof}

Given any family of unitaries $U = \{U_\secp\}_{\secp \in \bbN}$ on $\{n_{A,\secp} + n_{B,\secp}\}_{\secp \in \bbN}$ qubits promised by \cref{conjecture: instantaneous}, we define the functionality $\text{C-SWAP-U} = \{\text{C-Swap-U}_\secp\}_{\secp \in \bbN}$ as follows. $\text{C-SWAP-U}_\secp$ takes a $(n_{A,\secp} + n_{B,\secp})$-qubit state $(\bx_A,\bx_B)$ as input along with an additional two classical bits of input $z_A,z_B$. If $z_A \oplus z_B = 0$, it applies $U_\secp$ to $(\bx_A,\bx_B)$ to produce $(\by_A,\by_B)$, and then swaps the output states, outputting $(\by_B,\by_A)$. If $z_A \oplus z_B = 1$, it simply swaps the input states, outputting $(\bx_B,\bx_A)$. In what follows, we will show that any two-round two-party quantum computation protocol for $\text{C-SWAP-U}$ implies a correct and efficient instantaneous non-local quantum computation protocol for $U$, violating \cref{conjecture: instantaneous}.

Consider the oblivious simulator for the protocol computing $\text{C-SWAP-U}$. We will only be interested in the simulated first round and subsequent input extraction. We will not be concerned with simulating the second round at all. Furthermore, we will only care about simulating the view of a specific type of adversary: one that simply runs the honest $A$ (resp. $B$) algorithm. Such an ``adversary'' does not rush, i.e. the first message algorithm of the (honestly behaving) adversary is independent of $B$'s first round message $\bm_1^{(B)}$. Therefore for simplicity we will suppress mention of this message being generated by the simulator (since we are also not concerned with simulation of the second round).

Now, we will combine the ``first-round-only'' simulator discussed above with the first-message algorithms of parties $A$ and $B$ to produce the following algorithm $U_{\extract}$ (which can be written as a unitary), which will be applied to $(\bx_A,\bx_B)$ (tensored with sufficiently many $\Zstate$ states, which we write as $\Zstate^*$). Technically, $U_{\extract}$ is a family of unitaries parameterized by the security parameter $\secp$, and the inputs $(\bx_A,\bx_B)$ are families of input states, though for simplicity we will drop the explicit indexing by $\secp$.

$U_{\extract}$, on input $(\bx_A,\bx_B,\Zstate^*)$ works as follows:
\begin{enumerate}
    \item Compute $\left(\crs,\st_1^{(\Sim,A)},\st_1^{(\Sim,B)}\right) \gets \Sim_\comb$.
    \item Compute $(\bm_1^{(A)},\bst_1^{(A)}) \gets A_1(\bx_A,\crs)$.
    \item Compute $(\bm_1^{(B)},\bst_1^{(B)}) \gets B_1(\bx_B,\crs)$.
    \item Compute $(\bx_A',\bst_2^{(\Sim,A)}) \gets \Sim_2^{(A)}(\st_1^{(\Sim,A)},\bm_1^{(A)})$.
    \item Compute $(\bx_B',\bst_2^{(\Sim,B)}) \gets \Sim_2^{(B)}(\st_1^{(\Sim,B)},\bm_1^{(B)})$.
    \item Output $(\bx_A',\bx_B',\bst_1^{(A)},\bst_1^{(B)},\bst_2^{(\Sim,A)},\bst_2^{(\Sim,B)})$.
\end{enumerate}

Now, we show that for any pair of pure states $(\bx_A,\bx_B)$ that can be deterministically efficiently generated (i.e. can be generated by applying an efficient unitary to $\mathbf{0}$ states), the (traced out) portion of $U_\extract(\bx_A,\bx_B,\Zstate^*)$ consisting of $(\bx_A',\bx_B')$ is statistically close to $(\bx_A,\bx_B)$. First, we argue that $\bx_A' \approx_s \bx_A$. Recall that regardless of $A$'s classical input bit $z_A$, there is always a possibility that, depending on $B$'s classical input bit $z_B$, the functionality computed will simply be swapping $\bx_A$ and $\bx_B$ (from the definition of our C-SWAP-U unitary). In this case, the value $\bx_A'$ queried by $\Sim$ to the ideal functionality will be forwarded to $B$ as its output in the simulated world. $B$'s output in the real world is $\bx_A$, and thus $\bx_A' \approx_s \bx_A$, since otherwise the real and ideal worlds would be distinguishable by the measurement $\{\bx_A, \bbI - \bx_A\}$; note that projecting onto $\bx_A$ can be performed efficiently since $\bx_A$ is a (deterministically) efficiently generated pure state. An identical argument shows that $\bx_B' \approx_s \bx_B$.



Now, we can apply \cref{lemma:no-cloning} below to $U_{\extract}$; since the above argument applies to the case where $\bx_A,\bx_B$ are deterministically efficiently generated pure states, it in particular applies to the states required by \cref{lemma:no-cloning} (i.e. all computational basis states and all uniform superpositions of two computational basis states). \cref{lemma:no-cloning} applied to $U_{\extract}$ allows us to conclude that for \emph{any} input state $(\bx_A,\bx_B)$, the states $(\bst_1^{(A)},\bst_1^{(B)},\bst_2^{(\Sim,A)},\bst_2^{(\Sim,B)})$ are (statistically) independent of $(\bx_A,\bx_B)$. This fact can be used to design a correct and efficient instantaneous non-local quantum computation protocol for $U$, as described below. 


\begin{itemize}
    \item Setup: Execute $U_{\extract}$ on all $\Zstate$ states to produce $(\Zstate',\Zstate',\bst_1^{(A)},\bst_1^{(B)},\bst_2^{(\Sim,A)},\bst_2^{(\Sim,B)})$, where $\Zstate'$ denotes a state that is statistically indistinguishable from $\Zstate$. Discard $(\Zstate',\Zstate')$, send $(\bst_1^{(B)},\bst_2^{(\Sim,A)})$ to party $A$, and send $(\bst_1^{(A)},\bst_2^{(\Sim,B)})$ to party $B$.
    \item Party $A$, on input $\bx_A$, does the following.
    \begin{enumerate}
        \item Compute $(\st_1^{(\Sim,A)},\bm_1^{(A)}) \coloneqq {\Sim_2^{(A)}}^\dagger(\bx_A,\bst_2^{(\Sim,A)})$.
        \item Compute $(\bm_2^{(B)},\bst_2^{(B)}) \gets B_2(\bst_1^{(B)},\bm_1^{(A)})$.
        \item Send $\bm_2^{(B)}$.
    \end{enumerate}
    \item Party $B$, on input $\bx_B$, does the following.
    \begin{enumerate}
        \item Compute $(\st_1^{(\Sim,B)},\bm_1^{(B)}) \coloneqq {\Sim_2^{(B)}}^\dagger(\bx_A,\bst_2^{(\Sim,B)})$.
        \item Compute $(\bm_2^{(A)},\bst_2^{(A)}) \gets A_2(\bst_1^{(A)},\bm_1^{(B)})$.
        \item Send $\bm_2^{(A)}$.
    \end{enumerate}
    \item Party $A$ computes and outputs $\by_B \gets B_3(\bst_2^{(B)},\bm_2^{(A)})$.
    \item Party $B$ computes and outputs $\by_A \gets A_3(\bst_2^{(A)},\bm_2^{(B)})$.
\end{itemize}

Observe that the above protocol produces a transcript that is statistically close to the transcript between an honest $A$ and $B$. Fix $A$ and $B$'s classical inputs $z_A,z_B$ to be such that $z_A \oplus z_B = 0$. Since $A$ is receiving $B$'s output and $B$ is receiving $A$'s output, the parties are then computing a statistically close approximation to $U(\bx_A,\bx_B)$ in one round of online communication. The size of the initial entangled state held by $A$ and $B$ is bounded by the size of the honest $A$ and $B$ algorithms and the size of the simulator algorithms, which are all polynomial-size. Thus, the above is a correct and efficient instantaneous non-local quantum computation protocol for computing $U$.

\end{proof}

\begin{lemma}\label{lemma:no-cloning}
Let $U^{\gray{AB}}$ be a unitary over registers $\mathsf{A},\mathsf{B}$, and suppose that there exists $\epsilon$ between $0$ and $1$ such that for any $\ket{x}^{\gray{A}}$ that is (the density matrix of) either a computational basis state $\ket{i}$, or a uniform superposition of two computational basis states $\frac{1}{\sqrt{2}}(\ket{0} + \ket{i})$,
\[ \left|\trace_{\gray{B}}(U^{\gray{AB}} (\ket{x}^{\gray{A}} \otimes \ket{0}^{\gray{B}})) - \ket{x}^{\gray{A}}\right|_1 \leq \epsilon.\]
Then there exists a constant $\delta > 0$ and a pure state $\ket{y}^{\gray{B}}$ such that for \emph{every} state $\ket{x}^{\gray{A}}$,
\[ \left|U^{\gray{AB}}(\ket{x}^{\gray{A}} \otimes \ket{0}^{\gray{B}}) - (\ket{x}^{\gray{A}} \otimes \ket{y}^{\gray{B}})\right|_1 \leq \epsilon^\delta. \]
\end{lemma}

\begin{proof}
If $U^{\gray{AB}}$ satisfies the conditions of the lemma statement, then for any computational basis state $\ket{i}^{\gray{A}}$ on the $\gray{A}$ registers, there exists a pure state $\ket{y_i}^{\gray{B}}$ and a polynomial $\poly(\cdot)$ such that
\[\left|U^{\gray{AB}}(\ket{i}^{\gray{A}} \otimes \ket{0}^{\gray{B}}) - (\ket{i}^{\gray{A}} \otimes \ket{y_i}^{\gray{B}})\right|_1 \leq \poly(\epsilon). \]


Moreover, for each $i$ there must exist a state $\ket{y_{0,i}}^{\gray{B}}$ such that

$$\left|U^{\gray{AB}}\left(\frac{1}{\sqrt{2}}(\ket{0}^{\gray{A}} + \ket{i}^{\gray{A}}) \otimes \ket{0}^{\gray{B}}\right) - \left(\frac{1}{\sqrt{2}}(\ket{0}^{\gray{A}} + \ket{i}^{\gray{A}}) \otimes \ket{y_{0,i}}^{\gray{B}}\right)\right|_1 \leq \poly(\epsilon).$$


By linearity, we also have that $$\left|U^{\gray{AB}}\left(\frac{1}{\sqrt{2}}(\ket{0}^{\gray{A}} + \ket{i}^{\gray{A}}) \otimes \ket{0}^{\gray{B}}\right) - \left(\frac{1}{\sqrt{2}}(\ket{0}^{\gray{A}} \otimes \ket{y_0}^{\gray{B}}) +\frac{1}{\sqrt{2}}(\ket{i}^{\gray{A}} \otimes \ket{y_i}^{\gray{B}}) \right)\right|_1 \leq \poly(\epsilon).$$ 

This implies that for each $i$, $\ket{y_{0,i}}^{\gray{B}}$ is within $\poly(\epsilon)$ trace distance of both $\ket{y_0}^{\gray{B}}$ and $\ket{y_i}^{\gray{B}}$, which means that $\ket{y_0}^{\gray{B}}$ is within $\poly(\epsilon)$ trace distance of $\ket{y_i}^{\gray{B}}$. Thus, $\ket{y_0}^{\gray{B}}$ satisfies the condition in the lemma statement, since for an arbitrary state $\ket{x}^{\gray{A}} = \sum_i \alpha_i\ket{i}^{\gray{A}}$, we have that $$\left|U^{\gray{AB}}\left(\sum_i \alpha_i\ket{i}^{\gray{A}} \otimes \ket{0}^{\gray{B}}\right) - \sum_i \alpha_i(\ket{i}^{\gray{A}} \otimes \ket{y_i}^{\gray{B}})\right|_1 \leq \poly(\epsilon),$$ which implies that 
 
$$\left|U^{\gray{AB}}\left(\sum_i \alpha_i\ket{i}^{\gray{A}} \otimes \ket{0}^{\gray{B}}\right) - (\ket{x}^{\gray{A}} \otimes \ket{y_0}^{\gray{B}})\right|_1 \leq \poly(\epsilon).$$

\end{proof}


\subsection{A Two-Round Protocol from Quantum VBB Obfuscation}\label{subsec:vbb-protocol}
In what follows, we describe a two-round two-party protocol in the common reference string model, assuming the existence of a (strong form of) VBB obfuscation of quantum circuits. We note that it is straightforward to adapt this protocol to the multi-party setting. However, the main idea is already present in the two-party case, so for simplicity we describe only this case.

\subsubsection{VBB Obfuscation of Quantum Circuits}

We consider virtual black-box obfuscation of quantum circuits, which was defined (and shown to be impossible in general) by~\cite{Alagic2016OnQO}. In fact, we consider a potentially stronger version than that given by~\cite{Alagic2016OnQO}, who only consider VBB obfuscation of unitaries. We consider quantum functionalities $Q$ from $n$ qubits to $n$ qubits that include not just unitary gates, but also \emph{measurement} gates, and unitary gates that may be \emph{classically controlled} on the outcome of the measurement gates. While one can always push any measurement to the end of the computation so that the circuit becomes unitary, doing so would not necessarily preserve the security of obfuscation, as it would introduce new auxiliary input registers that a malicious evaluator may initialize in a non-zero state. Thus, obfuscation for unitary + measurement circuits is potentially stronger than obfuscation for unitaries.

We model black-box access to a unitary+measurement circuit as an oracle that accepts a quantum state on $n$ registers, manipulates it according to $Q$, and returns those same $n$ registers. We allow the obfuscation itself to be either a quantum circuit with a purely classical description, or a quantum circuit along with some quantum state. We refer to this obfuscation as $\cO(Q)$, and write $\cO(Q)(\bx)$ to indicate evaluation of the obfuscation on an $n$-qubit input $\bx$, with the understanding that this operation may either be directly applying a quantum circuit to $\bx$, or first augmenting $\bx$ with additional registers, applying a circuit to the expanded system, and then discarding the extra registers. 

\begin{definition}[Quantum VBB Obfuscation]\label{defn:vbb}
Let $\{\cQ_n\}_{n \in \bbN}$ be a family of polynomial-size quantum circuits, where each $Q \in \cQ_n$ maps $n$ qubits to $n$ qubits.
A quantum black-box obfuscator $\cO$ is a quantum algorithm that takes as input an input length $n \in \bbN$, a security parameter $\secp \in \bbN$, and a quantum circuit $Q$, and outputs an obfuscated quantum circuit. $\cO$ should satisfy the following properties.
\begin{itemize}
    \item Polynomial expansion: for every $n,\secp \in \bbN$ and $Q \in \cQ_n$, the size of $\cO(1^n,1^\secp,Q)$ is at most $\poly(n,\secp)$.
    \item Functional equivalence: for every $n,\secp \in \bbN$, $Q \in \cQ_n$, and $\bx$ on $n$ qubits, $\cO(1^n,1^\secp,Q)(\bx) \approx_s Q(\bx)$.
    \item Virtual black-box: for every (non-uniform) QPT $\cA$, there exists a (non-uniform) QPT $\cS$ such that for each $n \in \bbN$ and $Q \in \cQ_n$, $$\left|\Pr[\cA(\cO(1^n,1^\secp,Q)) = 1] - \Pr[\cS^Q(1^n,1^\secp) = 1]\right| = \negl(\secp).$$
\end{itemize}
\end{definition}

We now make a few remarks on the definition that will allow us to simplify the constructions given in the next section.

\begin{itemize}
    \item We will consider functionalities that discard, or trace out, some subset of registers. In order to implement this with a circuit from $n$ qubits to $n$ qubits, we can have the functionality measure the subset of qubits to be traced out and then ``randomize'' the outcomes (since we don't want the evaluator to know these measurement results) by applying Hadamard to each register and measuring again. Thus, we will consider obfuscation of functionalities from $n$ qubits to $m \leq n$ qubits.
    \item We will consider functionalities represented by quantum circuits that require the use of auxiliary $\Zstate$ states. We do not want the evaluator to be able to run such a functionality using non-zero auxiliary states, so we'll have the functionality first measure any auxiliary states input by the evaluator. If all measurements are 0, then the circuit will be run on the all registers, otherwise the functionality can ``abort'' by discarding all registers as explained above. Thus, we will suppress mention of auxiliary input registers, and assume the functionality has access to any auxiliary $\Zstate$ states that it needs.
    \item We will consider functionalities that can sample classical bits uniformly at random. This can be accomplished by applying Hadamard to a $\Zstate$ state and measuring. Note that this is a uniquely quantum phenomenon - one cannot obfuscate classical circuits that produce their own randomness. 
\end{itemize}

\subsubsection{The Protocol}

We present a two-round protocol for two-party quantum computation in the common reference string model. Let $Q$ be the two-party quantum functionality to be computed, and assume for simplicity that it takes $n$ qubits from each party and outputs $n$ qubits to each party. The common reference string will consist of obfuscations of six quantum functionalities $$\cF_{A,\inp}^{(b)},\cF_{B,\inp}^{(b)},\cF_{A,\cmp},\cF_{B,\cmp},\cF_{A,\out}^{(b)},\cF_{B,\out}^{(b)},$$ three to be used by each party. Each functionality has hard-coded some subset of 8 PRF keys $$k_{\inp}^{(A,A)},k_{\inp}^{(A,B)},k_{\inp}^{(B,A)},k_{\inp}^{(B,B)},k_{\out}^{(A,A)},k_{\out}^{(A,B)},k_{\out}^{(B,A)},k_{\out}^{(B,B)}.$$ We take each $\PRF(k_{\inp}^{(\cdot,\cdot)},\cdot)$ to be a mapping from a $\secp$-bit string to a classical description of a Clifford $C \in \mathscr{C}_{n+\secp}$. Each PRF key is used in two of the six obfuscated circuits, and the pair of letters in the superscript refers to the identity of the party associated with the first obfuscation it is used in, followed by the identify of the party associated with the second obfuscation it is used in.

Below we describe only  $\cF_{A,\inp}^{(b)},\cF_{A,\cmp}$, and $\cF_{A,\out}^{(b)}$ since $\cF_{B,\inp}^{(b)},\cF_{B,\cmp}$, and $\cF_{B,\out}^{(b)}$ are defined exactly the same with $A$ and $B$ switched.

\begin{itemize}
    \item $\cF_{A,\inp}^{(b)}\left[k_{\inp}^{(A,A)},k_{\inp}^{(A,B)}\right]$:
    \begin{enumerate}
        \item Take as input $(\bx_A,\bd_A)$ which consists of $A$'s input $\bx_A$ on $n$ qubits and a ``dummy'' input $\bd_A$ on $n$ qubits.
        \item Sample classical strings $r_{\inp}^{(A,A)},r_{\inp}^{(A,B)} \gets \{0,1\}^\secp$.
        \item Compute $C_{\inp}^{(A,A)} \coloneqq \PRF(k_{\inp}^{(A,A)},r_{\inp}^{(A,A)}),C_{\inp}^{(A,B)} \coloneqq \PRF(k_{\inp}^{(A,B)},r_{\inp}^{(A,B)})$.
        \item Output $$\begin{cases}\left(r_{\inp}^{(A,A)},C_{\inp}^{(A,A)}(\bx_A,\Zstate^\secp),r_{\inp}^{(A,B)},C_{\inp}^{(A,B)}(\bd_A,\Zstate^\secp)\right) \text{ if } b=0 \\ \left(r_{\inp}^{(A,A)},C_{\inp}^{(A,A)}(\bd_A,\Zstate^\secp),r_{\inp}^{(A,B)},C_{\inp}^{(A,B)}(\bx_A,\Zstate^\secp)\right) \text{ if } b=1 \end{cases}.$$
    \end{enumerate} 
    \item $\cF_{A,\cmp}\left[Q,k_{\inp}^{(A,A)},k_{\inp}^{(B,A)},k_{\out}^{(A,A)},k_{\out}^{(A,B)}\right]$:
    \begin{enumerate}
        \item Take as input $\left(r_{\inp}^{(A,A)},\widehat{\bx}_A,r_{\inp}^{(B,A)},\widehat{\bx}_B\right)$, where $\widehat{\bx}_A$ and $\widehat{\bx}_B$ are $(n+\secp)$-qubit states.
        \item Compute $C_{\inp}^{(A,A)} \coloneqq \PRF(k_{\inp}^{(A,A)},r_{\inp}^{(A,A)}),C_{\inp}^{(B,A)} \coloneqq \PRF(k_{\inp}^{(B,A)},r_{\inp}^{(B,A)})$.
        \item Compute $C_{\inp}^{(A,A)}(\widehat{\bx}_A)$ and measure the final $\secp$ qubits. If each is zero, let $\bx_A$ be the remaining $n$-qubit state. Otherwise, abort.
        \item Compute $C_{\inp}^{(B,A)}(\widehat{\bx}_B)$ and measure the final $\secp$ qubits. If each is zero, let $\bx_B$ be the remaining $n$-qubit state. Otherwise, abort.
        \item Compute $(\by_A,\by_B) \coloneqq Q(\bx_A,\bx_B)$.
        \item Sample classical strings $r_{\out}^{(A,A)},r_{\out}^{(A,B)} \gets \{0,1\}^\secp$.
        \item Compute $C_{\out}^{(A,A)} \coloneqq \PRF(k_{\out}^{(A,A)},r_{\out}^{(A,A)}),C_{\out}^{(A,B)} \coloneqq \PRF(k_{\out}^{(A,B)},r_{\out}^{(A,B)})$.
        \item Output $$\left(r_{\out}^{(A,A)},C_{\out}^{(A,A)}(\by_A,\Zstate^\secp),r_{\out}^{(A,B)},C_{\out}^{(A,B)}(\by_B,\Zstate^\secp)\right).$$
    \end{enumerate}
    \item $\cF^{(b)}_{A,\out}\left[k_{\out}^{(A,A)},k_{\out}^{(B,A)}\right]:$
    \begin{enumerate}
        \item Take as input $\left(r_{\out}^{(A,A)},\widehat{\by}_A^{(0)},r_{\out}^{(B,A)},\widehat{\by}_A^{(1)}\right)$, where $\widehat{\by}_A^{(0)}$ and $\widehat{\by}_A^{(1)}$ are $(n+\secp)$-qubit states.
        \item Compute $C_{\out}^{(A,A)} \coloneqq \PRF(k_{\out}^{(A,A)},r_{\out}^{(A,A)}),C_{\out}^{(B,A)} \coloneqq \PRF(k_{\out}^{(B,A)},r_{\out}^{(B,A)})$.
        \item Compute $C_{\out}^{(A,A)}(\widehat{\by}_A^{(0)})$ and measure the final $\secp$ qubits. If each is zero, let $\by_A^{(0)}$ be the remaining $n$-qubit state. Otherwise, abort.
        \item Compute $C_{\out}^{(B,A)}(\widehat{\by}_A^{(1)})$ and measure the final $\secp$ qubits. If each is zero, let $\by_A^{(1)}$ be the remaining $n$-qubit state. Otherwise, abort.
        \item Output $\by_A^{(b)}$.
    \end{enumerate}
\end{itemize}

Now we are ready to describe the protocol.

\protocol
{\proref{fig:vbb-protocol}}
{Two-round two-party quantum computation.}
{fig:vbb-protocol}
{
\textbf{Common Information:} Quantum circuit $Q$ to be computed with $2n$ input qubits and $2n$ output qubits.\\

\textbf{Party A Input:} $\bx_A$\\
\textbf{Party B Input:} $\bx_B$\\

\underline{\textbf{The Protocol:}}

\textbf{Setup.}
\begin{enumerate}
    \item Sample 8 PRF keys $k_{\inp}^{(A,A)},k_{\inp}^{(A,B)},k_{\inp}^{(B,A)},k_{\inp}^{(B,B)},k_{\out}^{(A,A)},k_{\out}^{(A,B)},k_{\out}^{(B,A)},k_{\out}^{(B,B)}.$
    \item Publish the following obfuscations:
    \begin{gather*}
        \cO_{A,\inp} \coloneqq \cO\left(\cF^{(1)}_{A,\inp}\left[k_{\inp}^{(A,A)},k_{\inp}^{(A,B)}\right]\right), \cO_{B,\inp} \coloneqq \cO\left(\cF^{(0)}_{B,\inp}\left[k_{\inp}^{(B,B)},k_{\inp}^{(B,A)}\right]\right) \\
        \cO_{A,\cmp} \coloneqq \cO\left(\cF_{A,\cmp}\left[Q,k_{\inp}^{(A,A)},k_{\inp}^{(B,A)},k_{\out}^{(A,A)},k_{\out}^{(A,B)}\right]\right), \\ \cO_{B,\cmp} \coloneqq \cO\left(\cF_{B,\cmp}\left[Q,k_{\inp}^{(B,B)},k_{\inp}^{(A,B)},k_{\out}^{(B,B)},k_{\out}^{(B,A)}\right]\right), \\
        \cO_{A,\out} \coloneqq \cO\left(\cF^{(1)}_{A,\out}\left[k_{\out}^{(A,A)},k_{\out}^{(B,A)}\right]\right), \cO_{B,\out} \coloneqq \cO\left(\cF^{(0)}_{B,\out}\left[k_{\out}^{(B,B)},k_{\out}^{(A,B)}\right]\right).
    \end{gather*}
\end{enumerate}

\textbf{Round 1.}

\emph{Party $A$:}
\begin{enumerate}
\item Compute $(st_{A,1},\bst_{A,1},m_{A,1},\bm_{A,1}) \gets \cO_{A,\inp}(\bx_A,\Zstate^n)$, where $\Zstate^n$ is the ``dummy input''. 
\item Send to Party $B$: $(m_{A,1},\bm_{A,1})$.
\end{enumerate}

\emph{Party $B$:}
\begin{enumerate}
\item Compute $(st_{B,1},\bst_{B,1},m_{B,1},\bm_{B,1}) \gets \cO_{B,\inp}(\bx_B,\Zstate^n)$, where $\Zstate^n$ is the ``dummy input''. 
\item Send to Party $A$: $(m_{B,1},\bm_{B,1})$.
\end{enumerate}

\textbf{Round 2.}

\emph{Party $A$:}
\begin{enumerate}
\item Compute $(st_{A,2},\bst_{A,2},m_{A,2},\bm_{A,2}) \gets \cO_{A,\cmp}(st_{A,1},\bst_{A,1},m_{B,1},\bm_{B,1})$.
\item Send to Party $B$: $(m_{A,2},\bm_{A,2})$.
\end{enumerate}

\emph{Party $B$:}
\begin{enumerate}
\item Compute $(st_{B,2},\bst_{B,2},m_{B,2},\bm_{B,2}) \gets \cO_{B,\cmp}(st_{B,1},\bst_{B,1},m_{A,1},\bm_{A,1})$. 
\item Send to Party $A$: $(m_{B,2},\bm_{B,2})$.
\end{enumerate}

\textbf{Output Reconstruction.}
\begin{itemize}
\item \emph{Party $A$:} Compute $\by_A \gets \cO_{A,\out}(st_{A,2},\bst_{A,2},m_{B,2},\bm_{B,2})$.
\item \emph{Party $B$:} Compute $\by_B \gets \cO_{B,\out}(st_{B,2},\bst_{B,2},m_{A,2},\bm_{A,2})$.
\end{itemize}
}
\clearpage

\begin{theorem}
Assuming quantum VBB obfuscation (\cref{defn:vbb}), the protocol described in \proref{fig:vbb-protocol} satisfies security against a malicious $A$ and malicious $B$.
\end{theorem}

\begin{proof}
(Sketch) First observe where the computation is taking place in an honest execution of the protocol. In the first round, $A$ executes $\cO_{A,\inp}$ and sends a Clifford encoding of its input $\bx_A$ to $B$, while $B$ executes $\cO_{B,\inp}$ and sends a Clifford encoding of its dummy input $\Zstate^n$ to $A$ while keeping a Clifford encoding of its input $\bx_B$ in its state. In the second round, $B$ executes $\cO_{B,\cmp}$ to produce a Clifford encoding of the output $(\by_A,\by_B)$, while $A$ executes $\cO_{A,\cmp}$ to produce a Clifford encoding of a dummy output. $B$ sends the encoding of $\by_A$, which is decrypted by $A$ using $\cO_{A,\out}$ and $B$ decrypts its output $\by_B$ using $\cO_{B,\out}$.

Now note that it is straightforward to perfectly simulate a malicious $A$. The simulator can sample the CRS (so in particular it knows all the PRF keys) and then emulate an honest $B$, receiving $A$'s encoded input, decrypting it, evaluating the circuit, and finally encoding $A$'s output and sending it back.

Now consider sampling the CRS to be obfuscations of the functionalities $$\cF_{A,\inp}^{(0)},\cF_{B,\inp}^{(1)},\cF_{A,\cmp},\cF_{B,\cmp},\cF_{A,\out}^{(0)},\cF_{B,\out}^{(1)}.$$ This only differs from the real protocol in the superscript $b$ values of the $\inp$ and $\out$ functionality. However, this completely reverses the flow of computation in an honest execution of the protocol. Now, $A$ is computing the functionality on the actual inputs, while $B$ is computing on the dummy inputs. Thus with this sampling of the CRS, it is straightforward to perfectly simulate a malicious $B$. It remains to show that these two methods of sampling the CRS are computationally indistinguishable.

Consider an adversary that can distinguish obfuscations of 

$$\cF_{A,\inp}^{(1)},\cF_{B,\inp}^{(0)},\cF_{A,\cmp},\cF_{B,\cmp},\cF_{A,\out}^{(1)},\cF_{B,\out}^{(0)}$$ from obfuscations of $$\cF_{A,\inp}^{(0)},\cF_{B,\inp}^{(1)},\cF_{A,\cmp},\cF_{B,\cmp},\cF_{A,\out}^{(0)},\cF_{B,\out}^{(1)}.$$ By the security of (our strong form of) VBB obfuscation, such an adversary implies a distinguisher that is only given oracle access to these functionalities. Now, since this distinguisher only has oracle access to the PRFs, one can replace them with truly random functions. At this point, due to the perfect hiding of the Clifford code, an adversary cannot obtain any information from the outputs of $(\cF_{A,\inp}^{(0)},\cF_{B,\inp}^{(1)},\cF_{A,\cmp},\cF_{B,\cmp})$, except with negligible probability (this non-zero probability is due to the possibility of collision in sampling the $r$ values used as inputs to the PRFs / random functions). Furthermore, due to the statistical authentication of the Clifford code, an adversary can only obtain an output from $(\cF_{A,\out}^{(b)},\cF_{B,\out}^{(1-b)})$ with non-negligible probability if it emulates an honest execution of the protocol, starting with arbitrary inputs $(\bx_A,\bd_A)$ to $\cF_{A,\inp}^{(b)}$ and $(\bx_B,\bd_B)$ to $\cF_{B,\inp}^{(1-b)}$. But regardless of the value of $b$, the outputs of $(\cF_{A,\out}^{(b)},\cF_{B,\out}^{(1-b)})$ will be $(\by_A,\by_B) \coloneqq Q(\bx_A,\bx_B)$. Thus switching the value of $b$ at this point will be statistically indistinguishable, completing the proof.

\end{proof}

\section*{Acknowledgements}

We thank anonymous reviewers for pointing out an issue with~\cref{thm:lower-bound} in an earlier draft of this work.

This material is based on work supported in part by DARPA under Contract No. HR001120C0024 (for DK). Any opinions, findings and conclusions or recommendations expressed in this material are those of the author(s) and do not necessarily reflect the views of the United States Government or DARPA. A.C. is supported by DoE under grant DOD ONR Federal. Part of this work was done while the authors were visiting the Simons Institute for the Theory of Computing in Spring 2020.

\ifsubmission
\else
\bibliographystyle{alpha}
\fi

\bibliography{abbrev3,custom,crypto,main}

\appendix

\ifsubmission

\else
\fi

\newpage

\clearpage

\end{document}